\PassOptionsToPackage{dvipsnames}{xcolor}
\documentclass[11pt,a4paper]{article}
\usepackage[margin=2.5cm]{geometry}
\usepackage[utf8]{inputenc}
\usepackage{authblk}

\usepackage{amsthm}
\usepackage{amssymb}
\usepackage{amstext}
\usepackage{amsmath}
\usepackage{url}
\usepackage{hyperref}
\usepackage{bookmark}
\usepackage{xspace}
\usepackage{caption}
\usepackage{subcaption}
\usepackage{lineno}

\usepackage[T1]{fontenc}
\usepackage{lmodern}
\usepackage{bm}
\usepackage{dsfont}

\usepackage{xcolor}

\usepackage{algorithm}
\usepackage{algpseudocode}
\algrenewcommand\algorithmicrequire{\textbf{Input:}}
\algrenewcommand\algorithmicensure{\textbf{Output:}}

\usepackage{graphicx}
\graphicspath{{./images/}}
\usepackage[outline]{contour}
\contourlength{0.8pt}

\usepackage{enumerate}
\usepackage{enumitem}

\usepackage{booktabs}
\usepackage{array}
\usepackage{longtable}

\usepackage{cleveref}

\usepackage{scalerel}
\usepackage{stackengine}

\stackMath
\newcommand\widecheck[1]{%
  \savestack{\tmpbox}{\stretchto{%
    \scaleto{%
      \scalerel*[\widthof{#1}]{\kern-.6pt\bigvee\kern-.6pt}%
      {\rule[-\textheight/2]{1ex}{\textheight}}%
    }{\textheight}%
  }{0.5ex}}%
  \stackon[1pt]{#1}{\tmpbox}%
}

\newtheorem{proposition}{Proposition}

\newtheorem{lemma}{Lemma}

\newtheorem{assumption}{Assumption}
\newtheorem{remark}{Remark}

\crefname{assumption}{assumption}{assumptions}
\Crefname{assumption}{Assumption}{Assumptions}

\def\argmin{\mathop{\rm arg \; min}\limits}%
\DeclareMathOperator{\arctantwo}{arctan2}

\newcommand{\D}{\mathcal{D}}

\newcommand{\R}{{\mathbb R}}

\newcommand{\N}{{\mathbb N}}
\newcommand{\Z}{{\mathbb Z}}

\newcommand{\Id}{{\mathrm{Id}}}
\newcommand{\indic}{{\mathds{1}}}
\newcommand{\dint}{{\mathrm{d}}}
\newcommand{\err}{{\mathtt{err}}}

\newcommand{\bc}{\bm{c}}

\newcommand{\be}{\bm{e}}
\newcommand{\bbf}{\bm{f}}

\newcommand{\bh}{\bm{h}}

\newcommand{\bn}{\bm{n}}
\newcommand{\bp}{\bm{p}}
\newcommand{\bq}{\bm{q}}
\newcommand{\bs}{\bm{s}}

\newcommand{\bv}{\bm{v}}

\newcommand{\bx}{\bm{x}}

\newcommand{\bz}{\bm{z}}

\newcommand{\bA}{\bm{A}}

\newcommand{\bR}{\bm{R}}

\newcommand{\dsop}{\mathord{\downarrow_U}}

\newcommand{\balpha}{\bm{\alpha}}

\newcommand{\bPhi}{\bm{\Phi}}

\def\argmin{\mathop{\rm arg \; min}\limits}%
\newcommand{\supp}{\mathop{\mathrm{supp}}}
\DeclareMathOperator{\diam}{diam}
\DeclareMathOperator*{\esssup}{ess\,sup}
\newcommand{\dist}{\mathrm{dist}}

\usepackage{tikz}
\usetikzlibrary{calc}
\usepackage{pgfplots}
\usepgfplotslibrary{groupplots}
\pgfplotsset{compat=newest}
\usetikzlibrary{plotmarks}
\usepgfplotslibrary{colorbrewer}
\usetikzlibrary{pgfplots.statistics, pgfplots.colorbrewer}
\pgfplotsset{cycle list/Set1-8}
\usepackage{pgfplotstable}
\usepackage{grffile}
\usepackage{pgffor,etoolbox}
\usepackage{python}

\usetikzlibrary{3d}
\usetikzlibrary{shapes.geometric, arrows.meta, positioning, calc,fit}

\tikzset{
  noblock/.style={
    minimum width=1.2cm, minimum height=0.9cm,
    text centered,
    font=\footnotesize
  },
  block/.style={
    rectangle, draw, thick,
    minimum width=1.2cm, minimum height=0.9cm,
    text width=2cm,
    text centered,
    font=\footnotesize
  },
  junction/.style={
    circle, fill=black, inner sep=0pt, minimum size=7pt
  },
  arr/.style={-{Stealth[length=6pt]}, thick},
  lin/.style={thick},
  dashbox/.style={
    rectangle, draw, dashed, thick,
    inner sep=15pt, rounded corners=0pt
  }
}

\hypersetup{colorlinks=true, allcolors=blue}

\newcommand{\orcid}[1]{}
\newcommand{\ack}[1]{\section*{Acknowledgments}#1}
\newcommand{\funding}[1]{\section*{Funding}#1}
\newcommand{\data}[1]{\section*{Data availability}#1}

\title{Scalable photoacoustic tomography implementations accounting for the spatial impulse response of transducers}

\author[1,4,6]{Trung-Thai Do}
\author[1]{Paul Escande}
\author[2,6]{Caroline Chaux}
\author[3]{Jérôme Gateau}
\author[4,5,6,7,8]{Hwee Kuan Lee}

\affil[1]{Institut de Mathématiques de Toulouse; UMR 5219, Université de Toulouse, CNRS; UPS, F-31062 Toulouse Cedex 9, France}
\affil[2]{Aix Marseille Univ, CNRS, Institut de Mathématiques de Marseille, Marseille, France}
\affil[3]{Sorbonne Université, CNRS, Inserm, Laboratoire d'Imagerie Biomédicale, LIB, F-75006, Paris, France}
\affil[4]{Bioinformatics Institute, Agency for Science, Technology and Research, 30 Biopolis Street, 138671, Singapore}
\affil[5]{Centre for Frontier AI Research, Agency for Science, Technology and Research, 1 Fusionopolis Way, 138671, Singapore}
\affil[6]{International Research Laboratory on Artificial Intelligence, Agency for Science, Technology and Research, 1 Fusionopolis Way, 138671, Singapore}
\affil[7]{School of Biological Sciences, Nanyang Technological University, 60 Nanyang Dr, 639798, Singapore}
\affil[8]{School of Computing, National University of Singapore, 13 Computing Dr, 117417, Singapore}

\date{}

\begin{document}

\maketitle

\begin{abstract}
Iterative model-based reconstruction in photoacoustic tomography repeatedly applies the forward operator mapping the initial pressure to the transducer signals, and its adjoint. At the scale of current three-dimensional systems, this operator cannot be stored and must be evaluated matrix-free, while accounting for the finite, focused surface of the transducers, whose spatial impulse response degrades the resolution when ignored.
Representing the initial pressure by compactly supported radial functions, we show that the measured signal is exactly a temporal convolution between a system kernel gathering the radial function and the electrical impulse response, and a purely geometric quantity accounting for the portion of the transducer surface reached by the wave emitted from a voxel during one time step. Two implementations are proposed, differing only in how this quantity is evaluated: a quadrature over points of the surface, as in existing works, or a closed-form area, which never discretizes the surface.
We derive closed forms for planar and cylindrically focused transducers and provide, in the latter case, two accelerations of the resulting elliptic integrals, a lookup table and a trapezoidal approximation, together with the piecewise planar approximation customary in the literature.
These implementations reduce the per-voxel geometric computations and are released as an open-source Python package for graphics processing units.
The performance of these operators is first demonstrated on a synthetic phantom, where the lookup-table-based operator reaches the accuracy of the exact evaluation ten times faster and outperforms the point discretization on both accuracy and runtime. A second experiment shows that they enable the processing of a realistic vascular phantom at full scale, with a higher peak signal-to-noise ratio and a better resolution than the back-projection counterpart.
The released implementations are an important step towards the adoption of three-dimensional model-based photoacoustic reconstructions with finite and focused transducers.

\end{abstract}

\section{Introduction}
\label{sec:introduction}
Photoacoustic tomography (PAT) is a hybrid biomedical imaging modality that combines optical contrast with ultrasonic resolution \cite{beard2011biomedical, wang2012photoacoustic}: by recording the acoustic waves generated following pulsed optical illumination of a sample, it reconstructs the spatial distribution of optical absorption at depths of several centimeters with sub-millimeter resolution. The resulting optical-absorption contrast, sensitive to endogenous chromophores such as hemoglobin, has made PAT a versatile tool for structural, functional and molecular imaging \cite{xia2014photoacoustic, yao2014photoacoustic, ntziachristos2010molecular}, and its non-ionizing nature has supported a growing number of preclinical \cite{perez2025dual} and clinical applications \cite{park2025clinical, attia2019review}.

\paragraph{Physical principle}
Photoacoustic systems are based on a pulsed optical excitation, typically of nanosecond duration, which ensures an efficient conversion of absorbed optical energy into ultrasonic waves under the stress- and thermal-confinement conditions \cite{wang2017photoacoustic}.
The generation of these waves in soft tissues can be described in two stages \cite{beard2011biomedical}: (i)~incident light is absorbed by optical chromophores, locally depositing heat inside the tissue; (ii)~the resulting thermoelastic expansion generates an initial pressure distribution $p_0 : \R^3 \to \R$, which subsequently propagates through the tissue as an ultrasonic wave.
These pressure waves are detected by an array of ultrasonic transducers surrounding the sample, whose piezoelectric elements convert the acoustic waves into measurable electrical signals $s$.

\paragraph{PAT reconstruction as an inverse problem}
Given the initial pressure $p_0$, the measured signals $s$ are determined by the acoustic wave equation and the detector model.
PAT image reconstruction addresses the converse direction: recovering $p_0$ from the measured signals $s$ alone \cite{poudel2019survey, anastasio2007application}.
Under the standard acoustic-wave model in a homogeneous, non-attenuating medium \cite{wang2017photoacoustic}, the complete mapping from the initial pressure distribution to the recorded measurements is linear and can be represented by a forward operator $\bA$.
The reconstruction problem therefore amounts to inverting
\begin{equation}
    \label{eq:inverse-problem}
    s = \bA p_0 + \eta,
\end{equation}
where the forward operator $\bA$ models acoustic propagation together with the transducer response, while $\eta$ accounts for measurement noise and modeling errors.
Although the continuous operator is closely related to the spherical-mean Radon transform, whose inversion has been extensively studied \cite{kuchment2008mathematics}, practical PAT systems differ significantly from this idealized setting. Attenuation, limited-view acquisition, finite transducer bandwidth and measurement noise \cite{arridge2016adjoint} make \eqref{eq:inverse-problem} severely ill-posed, so stable reconstruction relies on regularized iterative schemes that repeatedly evaluate $\bA$ and its adjoint $\bA^{*}$.
For realistic three-dimensional systems, applying $\bA$ and $\bA^*$ can be computationally expensive and require prohibitive amounts of memory.
Efficient evaluation of the forward operator and its adjoint is therefore essential for practical model-based PAT reconstruction.
The present work develops and analyzes efficient matrix-free implementations of these operators for realistic transducer geometries, aiming to reduce both computational cost and memory requirements.

\paragraph{Spatial impulse response of ultrasound detectors}
Practical ultrasound detectors are far from idealized point receivers.
Because the sensitivity per unit area of a piezoelectric element is limited, a finite aperture is used to collect sufficient acoustic signal.
Cylindrically focused surfaces are commonly employed to further benefit from acoustic focusing gain.

Such a detector integrates the incident pressure field over its surface before converting it into an electrical signal.
The resulting response depends on the position of the acoustic source relative to the detector and is described by the spatial impulse response (SIR) \cite{Wang2011,rosenthal2011model, mitsuhashi2014investigation}.
The signal recorded from a point source is not the pressure at a single location but a filtered version that depends on the geometry of the transducer and the angle under which the source is seen.
The SIR is therefore part of the forward operator $\bA$ in \eqref{eq:inverse-problem}.
Ignoring the SIR introduces a systematic loss of resolution that no amount of iterative refinement can recover \cite{queiros2014modeling, roitner2014deblurring}.

Existing strategies to account for the SIR range from post-acquisition deblurring with a measured or simulated SIR \cite{roitner2014deblurring}, which decouples the correction from the inversion, to explicit surface integration inside model-based pipelines \cite{piwakowski1989method,piwakowski1999new,rosenthal2011model,ding2017efficient, wise2019representing}.
The latter all hinge on a discrete sampling of the detector surface into quadrature points, whose density grows with the curvature of the surface and with the upper cut-off frequency of the detector, and thus directly controls the cost of every application of $\bA$ and $\bA^{*}$.
Closed-form expressions of the SIR, which would remove this quadrature altogether, are known only for a handful of canonical geometries (planar piston \cite{lockwood1973high,san1992diffraction,schmerr2007ultrasonic}, focused bowl \cite{hunt1983ultrasound}, rectangular double curved \cite{baek2012spatial}).
For the cylindrically focused detectors considered in this paper, the SIR is usually obtained by approximating the curved surface with small flat elements, for which a closed form is known \cite{jensen1992calculation,jensen1996field,wu1999spatial}. These detectors are therefore left without an exact closed-form expression suited to matrix-free forward operators.

\paragraph{Forward operators for PAT: state of the art}
Iterative regularized reconstruction algorithms \cite{dean2022practical, do2025reconstruction, paltauf2002iterative} are known to outperform direct analytical inversion in three-dimensional (3D) settings, especially for partial views and non-canonical detector geometries, but they only become tractable when the forward operator and its adjoint can be evaluated efficiently.
Existing implementations of the forward operator fall into four broad families: closed-form analytical formulas, full-wave numerical solvers, matrix-free Green's-function implementations, and learning-based reconstructions.
The first three differ in how they discretize the continuous wave-equation forward map while the fourth replaces part of the inversion pipeline with a neural network but, as discussed below, still relies on an efficient $\bA$ at training and inference time.

\begin{description}
    \item[\normalfont \textit{Closed-form analytical formulas}] such as universal back-projection \cite{xu2006photoacoustic, xu2005universal, burgholzer2007exact} exploit the analytical link between the photoacoustic measurement and the spherical-mean Radon transform of $p_0$ \cite{kuchment2008mathematics, finch2004determining, kunyansky2007explicit, haltmeier2014universal}. They are extremely fast but rest on idealized assumptions: homogeneous medium, point-like detectors that neglect the SIR and are arranged on a canonical full-view geometry (sphere, cylinder, plane), and infinite acquisition bandwidth.
    Beyond these regimes, the reconstructions degrade rapidly and exhibit characteristic limited-view artifacts \cite{paltauf2007experimental, xu2004reconstructions}.
Nevertheless, their computational efficiency makes them particularly attractive in practice for large-scale problems, where the computational and memory costs of model-based methods can become prohibitive.

    \item[\normalfont \textit{Full-wave numerical solvers}] numerically solve the acoustic wave equation by discretizing the spatial domain and propagating the pressure field over discrete time steps.
Finite-difference time-domain (FDTD) schemes approximate both the spatial and temporal derivatives by finite differences \cite{sheu2008simulations, huang2013full, mitsuhashi2017forward}.
Finite-element variants \cite{Yuan2007threeD, Yao2009finite} accommodate complex detector geometries at the price of mesh generation.
Both methods require a large number of points per wavelength for accurate discretization.
Moreover, the time step is constrained by a Courant--Friedrichs--Lewy (CFL) stability condition and must decrease as the spatial resolution is refined.
This leads to a larger number of time steps and prohibitive computational costs.

The dominant variant in PAT is the $k$-space pseudo-spectral scheme \cite{mast2001k, cox2007k} popularized by the open-source toolbox k-Wave \cite{treeby2010kwave, treeby2012modeling} and adopted in recent open-source pipelines such as PATATO \cite{else2024patato} and j-Wave \cite{stanziola2023j}: spatial derivatives are evaluated in the Fourier domain via the Fast-Fourier Transform (FFT), while the time derivative is approximated by a leapfrog scheme.
This drastically reduces the number of spatial points required. Together with a relaxed CFL condition, it permits larger time steps, allowing both the spatial and temporal discretizations to remain coarse and substantially reducing the computational cost.
The ability to handle arbitrary geometries and heterogeneous media, together with the availability of mature open-source code, has made these solvers a \textit{de facto} standard in many PAT studies \cite{poudel2019survey}. Time-reversal reconstruction \cite{hristova2008reconstruction, burgholzer2007exact} is built on the same machinery: it re-propagates the measured field backward in time with such a solver, and therefore inherits both their flexibility (arbitrary geometries and heterogeneous media) and their computational cost.

However, two features make these solvers ill-suited for the evaluation of the forward operator at a finite set of detector positions in large-scale 3D problems.
First, the simulation grid must encompass the sample, the entire detector array, and an absorbing layer, potentially resulting in one to two orders of magnitude more grid points than the region of interest. This is particularly inefficient since only the values at the detector surfaces are ultimately required.
Second, the treatment of finite-aperture detectors adds some technicalities.
In particular, the detector surface is not aligned with the grid, and representing it accurately requires off-grid quadrature points mapped onto the grid through band-limited interpolants \cite{wise2019representing}. This machinery is inexpensive in itself, but its accuracy is linked to the grid spacing and usually requires finer spatial sampling and correspondingly more time steps.

For these reasons, full-wave solvers are widely reported to be prohibitive for iterative reconstruction of large 3D systems \cite{li2026slingbag,wang2026gpair}, and require distributed memory \cite{jaros2016fullwave}.
    \item[\normalfont \textit{Matrix-free Green's-function implementations}]
Alternative strategies leverage the Green kernel of the wave equation by expressing each transducer signal as an integral of spherical waves emitted from each voxel of the sample over the detector surface \cite{rosenthal2010fast,buehler2011model,caballero2013model,queiros2014modeling,ding2017efficient, Hauptmann2018,ding2020model,li2026slingbag,wang2026gpair}.
These approaches only require the computational domain to encompass the sample, avoiding the need to extend the simulation grid to include the transducers.
Moreover, since the Green's function provides the pressure field at arbitrary observation times, the wave equation does not need to be solved using an explicit time discretization.
This avoids the large number of time steps imposed by the CFL condition and can substantially reduce the computational cost, and has been applied successfully to large-scale 3D reconstruction \cite{ding2020model,li2026slingbag,li2026slingbagpro,wang2026gpair}.

The remaining computational bottleneck is the surface integration itself: each detector surface must be discretized into quadrature points at which the contributions from all voxels are evaluated.
Accurate quadrature over curved, extended detector surfaces requires a sufficiently dense discretization, substantially increasing the computational cost.

A recent line of work chooses to represent the initial pressure field as an expansion in radial basis functions (RBFs).
For example, \cite{diebold1991photoacoustic} considers indicator functions, \cite{Lewitt1990attenuation,wang2014discrete} use Kaiser--Bessel kernels, while \cite{ding2020model} considers hat functions, and \cite{li2026slingbag,li2026slingbagpro,wang2026gpair} consider isotropic Gaussian kernels.
This approach is closely related to Kansa's radial basis function collocation \cite{kansa1990multiquadrics, franke1998solving}.
For a radially symmetric source, the emitted pressure field is a spherical wave whose temporal profile is determined by the radial profile of the source.
This property allows the spatial representation of the initial pressure and its acoustic propagation to be treated separately and in closed form. The resulting operations are well suited to efficient implementations on graphics processing units (GPUs).
However, existing approaches derive closed-form expressions for specific choices of radial basis functions. They work with isotropic point-detectors or approximate the transducer surface using quadrature points, and do not quantify the approximation error with respect to the continuous model.

Symmetries of the acquisition geometry can further compress these operators.
For detectors regularly sampled on a plane, Fourier-based techniques yield very fast
operators \cite{kostli2001temporal,cox2005fast,hauptmann2018approximate,kucukkomurcu2026depth}.
They rely, however, on the translation invariance of this specific geometry and do not
extend to general acquisition schemes.

\item[\normalfont \textit{Learning-based reconstruction}]
A complementary family of methods replaces part or all of the inversion pipeline with neural networks. Early works approached PAT reconstruction as an image-to-image task, training convolutional networks to map an analytical back-projection or the raw sinogram to the target image \cite{antholzer2019deep, davoudi2019deep, schwab2019learned, allman2018photoacoustic, lan2020net}. More recent unrolled and model-based architectures \cite{Hauptmann2018, hauptmann2020deep} embed the physical forward operator inside learned iterations, achieving better generalization across acquisition geometries but inheriting the same computational requirement: each unrolled step still requires one application of $\bA$ and $\bA^{*}$ at training and inference time. The same observation applies to the score- and diffusion-based posteriors, both generic \cite{song2024solving} and PAT-specific \cite{dey2024score}. Comprehensive surveys of the field are given in \cite{yang2021review, grohl2021deep, deng2021deep}. The efficient evaluation of $\bA$ is therefore not made obsolete by deep learning. On the contrary, it is a prerequisite for training and deploying efficient physics-aware reconstruction networks at 3D resolution.

A comparison of these four families through the lens of an iterative reconstruction pipeline reveals complementary limitations. Analytical inversion is fast but rigid and accuracy-limited outside canonical geometries. Full-wave solvers are flexible but oversized for repeated forward evaluations at a fixed sparse set of detectors. Matrix-free Green's-function schemes are efficient on point detectors but lose accuracy (or runtime) as soon as the detector surface is sampled too coarsely, with no \emph{a priori} bound on the resulting error. Learning-based reconstruction inherits the cost of $\bA$ at training time and, for unrolled architectures, at inference time as well. To the best of our knowledge, no existing method combines a closed-form treatment of extended, curved detectors, a quantified approximation error with respect to the continuous wave-equation model, and GPU-friendly arithmetic at the volumes encountered in modern 3D photoacoustic systems. This is the gap this paper addresses.

\end{description}

\paragraph{Contributions}

This paper provides fast matrix-free implementations of the forward operator $\bA$ and its adjoint $\bA^*$, which account for the SIR of the transducers, with low memory requirements that scale well to large 3D systems.
Based on a radial basis expansion of the initial pressure fields, these methods are accompanied by a rigorous mathematical framework in which their accuracy and their complexity are analyzed, and by a numerical validation at the scale of a real system.
In particular, our contributions are the following:
\begin{itemize}
        \item \emph{A rigorous modeling framework.}
        The discrete forward operators are derived from the wave equation to the final sampled digital signal, with the surface of the transducer explicitly incorporated through an integration of the pressure field over the surface to account for its SIR.
        The model also incorporates the electrical impulse response (EIR) of the transducers.
        Finally, the proposed formulation leverages the Green's function of the wave equation together with a Kansa-type expansion of the initial pressure field for arbitrary compactly supported radial bases.
        The most common approach to modeling finite-aperture transducers in the photoacoustic literature, that is, the point discretization of the detector surface, is recovered within this framework and forms the basis of the first of our two implementations.

    \item \emph{A new treatment of the transducer surface.}
    For our second implementation, we show that the measured signal can be expressed as a temporal convolution between a kernel and a sequence of coefficients indexed by the time steps. Each coefficient is the area of the detector surface enclosed between its intersections with two spheres centered at a given voxel, whose radii correspond to the distances traveled by the acoustic wave between consecutive time steps. We derive this area in closed form for planar and cylindrically focused transducers and present a general methodology that extends to arbitrary transducer geometries.

    \item \emph{Quantified approximation errors.}
Error bounds between the computed measurements and the continuous wave-equation model are established for both implementations, together with their computational and memory complexities. The two bounds differ by a quadrature term, present for the point method and removed by the analytical treatment of the surface, making the accuracy--cost trade-off of each implementation explicit.

    \item \emph{A GPU implementation validated at the scale of a real system.}
    Both implementations reduce to per-voxel computations followed by batched convolutions, making them well suited to GPU architectures.
    The forward and adjoint operators provided by both implementations are first compared on a reduced-size experiment mimicking a realistic acquisition with cylindrical transducers.
    The surface method achieves better reconstruction quality while requiring less computation.
    It is further demonstrated on a realistic large-scale 3D reconstruction scenario, which, to the best of our knowledge, is the first model-based reconstruction demonstrated at this scale that accounts for the SIR of detectors.
    The implementation is released as a Python package at \url{https://patminton.readthedocs.io/en/latest/}.
\end{itemize}

\paragraph{Paper structure}
The paper is organized as follows. \Cref{sec:forward_model} introduces the photoacoustic forward model, the Kansa discretization of $p_{0}$ and the convolutional reformulation of the measurements. \Cref{sec:implementations} develops the matrix-free forward operator in two variants, the point method and the surface method, with their complexities and error bounds. \Cref{sec:usual_surfaces} specializes the surface method to planar and cylindrical detectors. Numerical experiments on synthetic and realistic phantoms are reported in \Cref{sec:experiments}. Finally, \Cref{sec:conclusion} concludes the paper while sketching some future research directions.

\section{Modeling of PAT systems}
\label{sec:forward_model}
This section introduces the mathematical modeling of the PAT systems used in this paper. Additional details on its derivation can be found in \cite[Chapter 3]{wang2017photoacoustic} or \cite{bridal2024innovative}.

\subsection{Notations}
\label{sec:notations}

Vectors are denoted by bold lowercase letters, e.g. $\bv$, and matrices by bold uppercase letters, e.g. $\bm{M}$. For a vector $\bv$, $\bv[i]$ denotes its $i$-th component. For a matrix $\bm{M}$, $\bm{M}[i,j]$ denotes the entry in the $i$-th row and $j$-th column.
The $\ell_p$ norms of vectors will be denoted by $\| \cdot \|_p$.

Given any positive integer $n$, the notation $[n]$ refers to the set of the first $n$ integers, that is, $[n] = \{1, \ldots, n\}$. A family indexed by a multidimensional grid of $n$ points is always enumerated by the single flattened index set $[n]$, in an arbitrary but fixed order.

Let $a,b$ be two non-negative functions of a parameter $\alpha$ living in a set $A$. The notation $a(\alpha) \lesssim b(\alpha)$ means that there exists a constant $C > 0$ such that $a(\alpha) \leq C b(\alpha)$ for all $\alpha \in A$, with $C$ independent of $\alpha$.
Moreover, we shall often use the notation $C(\beta_1)$ or $C(\beta_1, \beta_2)$ to emphasize the dependence of the constants on $\beta_1$ and $(\beta_1, \beta_2)$ respectively, where each $\beta_1, \beta_2$ may stand for a domain, a function, a vector, etc.

For $1 \leq q \leq \infty$, $L^q(\Omega)$ stands for the usual Lebesgue space of real-valued functions $u : \Omega \subset \mathbb{R}^d \to \R$, equipped with the norm $\| \cdot \|_{L^q(\Omega)}$ defined by

\begin{equation*}
\|u\|_{L^q(\Omega)} = \biggl( \int_\Omega |u(\bx)|^q \, \dint \bx \biggr)^{1/q}, \; 1 \leq q < \infty \; \text{ and } \; \|u\|_{L^\infty(\Omega)} = \esssup_{\bx\in\Omega} |u(\bx)|.
\end{equation*}

When the domain $\Omega$ is clear from the context, its dependence will be dropped and we will write $ \|\cdot\|_{L^q}$.
The Lebesgue spaces of vector-valued functions are defined componentwise,

\begin{equation*}
       L^q(\Omega ; \R^m) = \bigl\{ u : \Omega \to \R^m \ \big| \ u_j \in L^q(\Omega) \ \text{ for all } j \in [m] \bigr\},
\end{equation*}
equivalently described by $\| u(\cdot) \|_2 \in L^q(\Omega)$ since all norms on $\R^m$ are equivalent; it is equipped with the norm $\| u \|_{L^q(\Omega;\R^m)} = \bigl\| \, \| u(\cdot) \|_2 \, \bigr\|_{L^q(\Omega)}$.

The Sobolev space $W^{k,q}(\Omega ; \R^m)$, with $1 \leq q \leq \infty$ and $k \in \N$, is defined as the set of functions $u \in L^q(\Omega ; \R^m)$ whose weak partial derivatives up to the order $k$ are in $L^q(\Omega ; \R^m)$. It is a Banach space, equipped with the norm
\begin{equation*}
       \| u \|_{W^{k,q}(\Omega ; \R^m)} = \| u \|_{L^q(\Omega ; \R^m)} + |u|_{W^{k,q}(\Omega ; \R^m)},
       \quad \text{where } |u|_{W^{k,q}(\Omega ; \R^m)} = \sum_{1 \le |\alpha| \le k} \| \partial^\alpha u \|_{L^q(\Omega ; \R^m)}
\end{equation*}
is the associated seminorm, $\alpha = (\alpha_1, \ldots, \alpha_d)$, $|\alpha| = \sum_{i=1}^{d} \alpha_i$ and $\partial^\alpha = \partial^{\alpha_1}_1 \ldots \partial^{\alpha_d}_d$. For a vector-valued function, $\partial^\alpha u = (\partial^\alpha u_1, \ldots, \partial^\alpha u_m)$. When $m=1$, the associated Sobolev space will be written $W^{k,q}(\Omega)$.

The main notations used throughout the paper are collected in \Cref{sec:notations_table}, each symbol is also defined at its first use.
\subsection{Optical generation of ultrasound waves}
\label{subsec:wave_prop}

\paragraph{Geometric setup}
The geometry of the photoacoustic system is depicted in \Cref{fig:geometry}.
We work in a single Cartesian coordinate system on $\R^3$, in which any point of space, whether inside the sample or at a transducer location, is represented by the same vector variable $\bx \in \R^3$. When a source point and a detection point appear in the same expression, the source point is written $\bx'$; the grid points introduced in \Cref{subsec:discretization} are written $\bx_i$.

The optical absorbers are confined to a bounded \emph{sample region} $\Omega \subset \R^3$, taken to be a rectangular box centered at the origin.

The acoustic field is recorded by $K$ ultrasound transducers, modeled as two-dimensional surfaces $\Sigma_1, \ldots, \Sigma_K \subset \R^3$, all disjoint from $\Omega$.

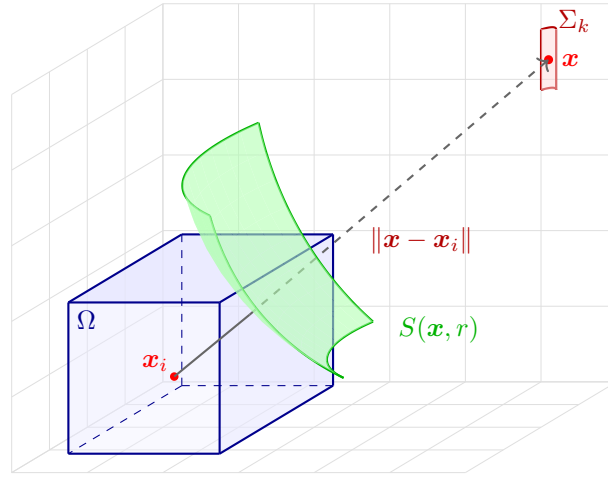
\begin{figure}[!htb]
       \centering
       \begin{tikzpicture}[
       x={(1cm, 0cm)},
       y={(-0.5cm, -0.3cm)},
       z={(0cm, 1cm)},
       scale=1.0,
       every node/.style={font=\small}
       ]
              \foreach \x in {-0.5,0.5,...,5.5}
              \draw[gray!25, very thin] (\x, -0.5, -0.5) -- (\x, 3.5, -0.5);
              \foreach \y in {-0.5,0.5,1.5,2.5,3.5}
              \draw[gray!25, very thin] (-0.5, \y, -0.5) -- (5.5, \y, -0.5);
              \foreach \x in {-0.5,0.5,...,5.5}
              \draw[gray!25, very thin] (\x, -0.5, -0.5) -- (\x, -0.5, 4.5);
              \foreach \z in {-0.5,0.5,...,4.5}
              \draw[gray!25, very thin] (-0.5, -0.5, \z) -- (5.5, -0.5, \z);
              \foreach \y in {-0.5,0.5,1.5,2.5,3.5}
              \draw[gray!25, very thin] (-0.5, \y, -0.5) -- (-0.5, \y, 4.5);
              \foreach \z in {-0.5,0.5,...,4.5}
              \draw[gray!25, very thin] (-0.5, -0.5, \z) -- (-0.5, 3.5, \z);

              \coordinate (P000) at (-1, -2, -1);
              \coordinate (P100) at ( 1, -2, -1);
              \coordinate (P110) at ( 1,  1, -1);
              \coordinate (P010) at (-1,  1, -1);
              \coordinate (P001) at (-1, -2,  1);
              \coordinate (P101) at ( 1, -2,  1);
              \coordinate (P111) at ( 1,  1,  1);
              \coordinate (P011) at (-1,  1,  1);

              \fill[blue!8,  opacity=0.5] (P000) -- (P100) -- (P101) -- (P001) -- cycle;
              \fill[blue!12, opacity=0.4] (P100) -- (P110) -- (P111) -- (P101) -- cycle;
              \fill[blue!6,  opacity=0.4] (P001) -- (P101) -- (P111) -- (P011) -- cycle;

              \fill[blue!8,  opacity=0.5] (P000) -- (P010) -- (P011) -- (P001) -- cycle;
              \fill[blue!12, opacity=0.4] (P111) -- (P101) -- (P001) -- (P011) -- cycle;
              \fill[blue!6,  opacity=0.4] (P000) -- (P100) -- (P110) -- (P010) -- cycle;

              \draw[blue!55!black, thick] (P100) -- (P110);
              \draw[blue!55!black, thick] (P100) -- (P101);
              \draw[blue!55!black, thick] (P001) -- (P101);
              \draw[blue!55!black, thick] (P101) -- (P111);
              \draw[blue!55!black, thick] (P110) -- (P111);
              \draw[blue!55!black, thick] (P111) -- (P011);
              \draw[blue!55!black, thick] (P001) -- (P011);
              \draw[blue!55!black, thick] (P010) -- (P110);
              \draw[blue!55!black, thick] (P010) -- (P011);

              \draw[blue!55!black, thin, dashed] (P000) -- (P010);
              \draw[blue!55!black, thin, dashed] (P000) -- (P001);
              \draw[blue!55!black, thin, dashed] (P000) -- (P100);

              \node[blue!55!black, anchor=north west, inner sep=3pt] at (P011) {$\Omega$};

              \coordinate (Xp) at (-0.3, -0.4, -0.4);
              \fill[red] (Xp) circle (1.6pt);
              \node[red, anchor=south east, inner sep=2pt] at (Xp) {$\bx_i$};

              \draw[red!70!black, thick] plot[smooth, samples=15, domain=-30:30]
                     ({5 + 0.2*sin(\x)}, {0.4 - 0.2*cos(\x*2)}, 3.6);
              \draw[red!70!black, thick] plot[smooth, samples=15, domain=-30:30]
                     ({5 + 0.2*sin(\x)}, {0.4 - 0.2*cos(\x*2)}, 4.4);
              \draw[red!70!black, thick] (4.9, 0.30, 3.6) -- (4.9, 0.30, 4.4);
              \draw[red!70!black, thick] (5.1, 0.30, 3.6) -- (5.1, 0.30, 4.4);
              \fill[red!15, opacity=0.5]
                     (4.9, 0.30, 3.6) -- (5.1, 0.30, 3.6) --
                     (5.1, 0.30, 4.4) -- (4.9, 0.30, 4.4) -- cycle;

              \coordinate (X) at (5, 0.3, 4);
              \fill[red] (X) circle (1.7pt);
              \node[red, anchor=west, inner sep=2pt] at (5.1, 0.3, 4) {$\bx$};
              \node[red!70!black, anchor=west, inner sep=2pt] at (5.2, 0.6, 4.6) {$\Sigma_k$};

              \draw[thick, black!60] (Xp) -- (1,0.21,0.6);
              \draw[->, dashed, thick, black!60] (1,0.21,0.6) -- (X);
              \node[anchor=north west, inner sep=3pt, sloped, red!70!black]
                     at ($(Xp)!0.5!(X)$) {$\|\bx - \bx_i\|$};

              \pgfmathsetmacro{\cxsp}{5}
              \pgfmathsetmacro{\cysp}{0.3}
              \pgfmathsetmacro{\czsp}{4}
              \pgfmathsetmacro{\spxv}{1.5}
              \pgfmathsetmacro{\spyv}{1.5}
              \pgfmathsetmacro{\spzv}{1.5}
              \pgfmathsetmacro{\Rxv}{\spxv-\cxsp}
              \pgfmathsetmacro{\Ryv}{\spyv-\cysp}
              \pgfmathsetmacro{\Rzv}{\spzv-\czsp}
              \pgfmathsetmacro{\Rsph}{sqrt(\Rxv*\Rxv+\Ryv*\Ryv+\Rzv*\Rzv)}
              \pgfmathsetmacro{\thp}{acos(\Rzv/\Rsph)}
              \pgfmathsetmacro{\php}{atan2(\Ryv,\Rxv)}
              \def\dth{20}
              \def\dph{30}

               \draw[green!70!black, thick] plot[smooth, samples=20, domain={\thp-\dth}:{\thp+\dth}]
                  ({\cxsp + \Rsph*sin(\x)*cos(\php-\dph)},
                   {\cysp + \Rsph*sin(\x)*sin(\php-\dph)},
                   {\czsp + \Rsph*cos(\x)});
              \draw[green!70!black, thick] plot[smooth, samples=20, domain={\thp-\dth}:{\thp+\dth}]
                  ({\cxsp + \Rsph*sin(\x)*cos(\php+\dph)},
                   {\cysp + \Rsph*sin(\x)*sin(\php+\dph)},
                   {\czsp + \Rsph*cos(\x)});
              \draw[green!70!black, thick] plot[smooth, samples=20, domain={\php-\dph}:{\php+\dph}]
                  ({\cxsp + \Rsph*sin(\thp-\dth)*cos(\x)},
                   {\cysp + \Rsph*sin(\thp-\dth)*sin(\x)},
                   {\czsp + \Rsph*cos(\thp-\dth)});
              \draw[green!70!black, thick] plot[smooth, samples=20, domain={\php-\dph}:{\php+\dph}]
                  ({\cxsp + \Rsph*sin(\thp+\dth)*cos(\x)},
                   {\cysp + \Rsph*sin(\thp+\dth)*sin(\x)},
                   {\czsp + \Rsph*cos(\thp+\dth)});

              \node[green!70!black, anchor=west, inner sep=2pt]
                  at ({\cxsp + \Rsph*sin(\thp+\dth+3)*cos(\php+\dph+3)},
                      {\cysp + \Rsph*sin(\thp+\dth+3)*sin(\php+\dph+3)},
                      {\czsp + \Rsph*cos(\thp+\dth+3)})
                  {$S(\bx, r)$};

              \foreach \i in {0,1,...,7} {
                \foreach \j in {0,1,...,11} {
                  \pgfmathsetmacro{\ti}{\thp - \dth + 2*\dth * \i / 8}
                  \pgfmathsetmacro{\tip}{\thp - \dth + 2*\dth * (\i+1) / 8}
                  \pgfmathsetmacro{\fj}{\php - \dph + 2*\dph * \j / 12}
                  \pgfmathsetmacro{\fjp}{\php - \dph + 2*\dph * (\j+1) / 12}
                  \fill[green!30!white, opacity=0.6]
                    ({\cxsp + \Rsph*sin(\ti)*cos(\fj)},
                     {\cysp + \Rsph*sin(\ti)*sin(\fj)},
                     {\czsp + \Rsph*cos(\ti)})
                    -- ({\cxsp + \Rsph*sin(\tip)*cos(\fj)},
                        {\cysp + \Rsph*sin(\tip)*sin(\fj)},
                        {\czsp + \Rsph*cos(\tip)})
                    -- ({\cxsp + \Rsph*sin(\tip)*cos(\fjp)},
                        {\cysp + \Rsph*sin(\tip)*sin(\fjp)},
                        {\czsp + \Rsph*cos(\tip)})
                    -- ({\cxsp + \Rsph*sin(\ti)*cos(\fjp)},
                        {\cysp + \Rsph*sin(\ti)*sin(\fjp)},
                        {\czsp + \Rsph*cos(\ti)})
                    -- cycle;
                }
              }

       \end{tikzpicture}
       \caption{Geometric setup of the photoacoustic system. The sample region $\Omega \subset \R^3$ is a rectangular box centered at the origin. All points are expressed in the same coordinate system: here $\bx$ is taken on a transducer surface $\Sigma_k$ (e.g.\ a cylindrically focused element), and $\bx_i$ inside $\Omega$ is a grid point (\Cref{subsec:discretization}). The green doubly-curved patch depicts a portion of the sphere $S(\bx, r)$ centered at $\bx$.}
       \label{fig:geometry}
\end{figure}

\paragraph{Wave propagation}
In this context, the initial pressure $p_0 : \R^3 \to \R$ is proportional to the product of the density of absorbers, their optical absorption cross-section and the local optical fluence; most photoacoustic imaging approaches focus on reconstructing $p_0$, which can be subsequently corrected by the optical fluence to retrieve the absorber density. We assume throughout that $p_0$ is compactly supported in the sample region $\Omega$ introduced above.

The medium between $\Omega$ and the transducer is assumed acoustically homogeneous and lossless. When $\Omega$ is excited by a short laser pulse at time $t = 0$, the pressure wave $p(\bx, t)$ can be modeled by the wave equation
\begin{equation} \label{eq:wave_equation}
       \left\{
              \begin{aligned}
                     &\partial_t^2 p(\bx,t) - c^2 \Delta p(\bx,t) = 0, \quad \bx \in \mathbb{R}^3, t > 0, \\
                     &p(\bx,0) = p_0(\bx), \quad \partial_t p(\bx,0) = 0
              \end{aligned}
       \right.
\end{equation}
where $c > 0$ is the speed of sound in the medium and $\Delta$ is the Laplacian operator.

The solution of the wave equation \eqref{eq:wave_equation} is given by Kirchhoff's and Poisson's formulas \cite[Section 2.4]{evans2022partial}. More precisely, the acoustic pressure field at position $\bx \in \mathbb{R}^3$ at time $t>0$ can be written as:
\begin{equation} \label{eq:forward}
       p(\bx,t) = \frac{1}{4 \pi c} \partial_t a(\bx, ct) \quad \text{with} \quad a(\bx, r) = \frac{1}{r} \int_{S(\bx,r)} p_0(\bx') \dint S_r(\bx'),
\end{equation}
where $S(\bx,r) = \left\{ \bx' \in \R^{3} \,|\, \|\bx' - \bx\| = r \right\}$ denotes the sphere of radius $r$ centered at $\bx$ and $\dint S_r$ is the surface element on this sphere (cf. \Cref{fig:geometry}).

\subsection{Pressure field measurement}
\label{subsec:measurement}

The acoustic pressure field is then measured over time, as electrical signals, by a collection of $K$ ultrasound transducers surrounding $\Omega$.
More precisely \cite{ding2017efficient}, the transducers are typically described by two-dimensional surfaces $\Sigma_k$ with surface element $\dint \sigma_k$.
Each transducer records a signal equal to the pressure field $p$ integrated over its surface.

Moreover, the piezoelectric material of the transducers acts as a band-pass filter, modeled by the \emph{electrical impulse response} (EIR) $e : \R \to \R$, a real-valued function of time depending on the piezoelectric material, the acoustic matching layers, and the associated electronic circuits \cite{schmerr2007ultrasonic,Wang2011}.

Therefore, the signal $s_k$ measured by the $k$-th transducer will be modeled by
\begin{equation} \label{eq:measurement}
       s_k(t) =  (e \star m_k)(t) = \int_{\R} e(t-t') \, m_k(t') \dint t', \quad \text{where } \quad  m_k(t) = \int_{\Sigma_k} p(\bx, t) \dint\sigma_k(\bx), \quad \text{for } t > 0.
\end{equation}

Following the modeling of \cite[Section 4.8.2]{oppenheim2010discrete}, each signal $s_k$ is converted by an analog-to-digital converter of sampling frequency $F_s$. This conversion is described through a sampling operator $\mathcal{S} : L^1(\R) \to \R^L$ defined as
\begin{equation} \label{eq:def_operatorS}
       (\mathcal{S} s)[l] = F_s \int_{t_{l-1}}^{t_l} s(t) \dint t, \quad l \in [L], \quad \text{where } t_l = t_0 + l/F_s,
\end{equation}
for some $t_0 \geq 0$ defining the origin time of recording.

Applying $\mathcal{S}$ to each $s_k$ gives the discrete measurements vector $\bs \in \R^{KL}$ defined as
\begin{equation} \label{eq:measurement_discrete}
       \bs[k,l] = (\mathcal{S}s_k)[l] \quad \text{for } l \in [L],\ k \in [K].
\end{equation}

\subsection{Discretization of $p_0$}
\label{subsec:discretization}

The forward modeling finally requires a spatial discretization of the initial pressure field $p_0$. We adopt a Kansa method \cite{kansa1990multiquadrics,franke1998solving}, that is, a representation of $p_0$ as a finite linear combination of radial functions. This approach, which has already been applied in the PAT context \cite{ding2020model,li2026slingbag}, offers several advantages discussed at the end of this section.

This discretization is performed by first choosing a set of points. We consider a uniform Cartesian grid $\{ \bx_i \}_{i \in [N]}$ of $N$ points covering $\Omega$, obtained by partitioning the box into equal voxels and taking their centers. Following \Cref{sec:notations}, the three grid indices are flattened into the single index $i \in [N]$: the unknown $\bp_0$ is therefore a vector of $\R^N$, not a third-order tensor.

Next, given a radial function $\phi : [0,\infty) \to \R$ and a coefficient vector $\bp_0 \in \R^N$, we represent the initial pressure field as
\begin{equation} \label{eq:continuous_discretized_field}
       p_0(\bx) = \sum_{i \in [N]} \bp_0[i] \phi( \| \bx - \bx_i\| ).
\end{equation}
The relation between the coefficients $\bp_0$ and the nodal values $(p_0(\bx_i))_i$, the implications of this discretization, and guidelines for the choice of $\phi$ and of its support radius $\kappa$, are discussed in \Cref{subsec:radial_function}.

\subsection{Setting}

In this paper, we will consider the setting where the transducers are sufficiently separated from $\Omega$, the support of $p_0$.
This assumption is formalized through the minimal grid-to-transducer distance $\dist_{\min}$ defined by
\begin{equation} \label{eq:min_dist}
       \dist_{\min} = \min_{k \in [K]} \, \min_{i \in [N]} \, \inf_{\bx \in \Sigma_k} \| \bx - \bx_i \|.
\end{equation}
In the remainder of the paper we will consider the setting formalized by the next assumption.
\begin{assumption}[Setting] \label{ass:setting}
       Assume that
       \begin{enumerate}[label=(\roman*)]
              \item (radial function) $\phi: [0,\infty) \to \R$ in \eqref{eq:continuous_discretized_field} is bounded, compactly supported in $[0,\kappa]$ with $\kappa > 0$; \label{ass:setting:phi}
              \item (EIR) $e : \R \to \R$ is compactly supported in $[0,T_e]$ and integrable; \label{ass:setting:e}
              \item (separation) $\dist_{\min} \geq \kappa $. \label{ass:setting:separation}
       \end{enumerate}
\end{assumption}
Common models describe the EIR as a damped oscillator \cite{arnau2004piezoelectric}. Although such models are not strictly compactly supported, their exponential decay allows them to be accurately approximated by functions with truncated support. Throughout this work, we assume that the truncation time $T_e$ is chosen sufficiently large so that the approximation error introduced by truncation can be neglected.
\Cref{ass:setting}\ref{ass:setting:separation} formalizes the separation of the support of $\bp_0$ and the transducers.
Smoothness conditions on $\phi$ and $e$ will be imposed in \Cref{sec:implementations}.

\subsection{Decoupling interpolation and propagation}
\label{subsec:decoupling}

In the case where the initial pressure field is an isolated emitting point $\bx'$, that is $p_0 = \delta_{\bx'}$, the solution formula \eqref{eq:forward} yields
\begin{equation} \label{eq:propagation_point2point}
       p(\bx, t) = \partial_t \left(\frac{1}{4 \pi c^2 t} \delta\left( \| \bx - \bx'\| - ct  \right)\right),
\end{equation}
meaning that $p$ is a spherical wave of center $\bx'$ propagating at speed $c$.

With the discretization \eqref{eq:continuous_discretized_field}, the pressure field $p(\bx,t)$ now depends non-trivially on the interpolation function $\phi$, resulting in a coupling between discretization and wave propagation.
In particular, the contribution of $\bx_i$ via \eqref{eq:forward} involves the integral of $\phi(\| \cdot - \bx_i\|)$ along the sphere $S(\bx, ct)$.
By rotational invariance, this integral does not depend on the direction of $\bx - \bx_i$, as already observed in \cite{ding2017efficient,ding2020model}.
It does not depend on the curvature of the sphere, which is governed by the distance $ct$.
This strategy allows treating each $\bx_i$ as an isolated point source, as in \eqref{eq:propagation_point2point}, while correcting their contribution through a convolution kernel that only depends on $\phi$.
Such a decoupling of wave propagation and interpolation has already appeared in the PAT context \cite{diebold1991photoacoustic,wang2014discrete,ding2020model,li2026slingbag} for specific choices of radial functions $\phi$.

\begin{proposition} \label{prop:point_value_pressure_field}
       Given an initial pressure field of the form \eqref{eq:continuous_discretized_field}, let $p$ be the associated pressure field given by \eqref{eq:forward}. Then, under \Cref{ass:setting}\ref{ass:setting:phi} and \ref{ass:setting:separation}
       \begin{equation*}
              p(\bx, t) = \sum_{i \in [N]} g\left( t - \frac{\|\bx - \bx_i\|}{c} \right) \frac{\bp_0[i]}{\| \bx - \bx_i \|}, \quad \text{for almost every } t > 0, \, \forall \bx \in \bigcup_{k \in [K]} \Sigma_k,
       \end{equation*}
       with $g : \R \to \R$ defined by
       \begin{equation*}
              g(t) = - \frac{1}{2} ct \phi(|ct|).
       \end{equation*}
\end{proposition}
When $\phi$ is continuous and vanishes at $\kappa$, the identity holds for every $t > 0$.
The proof is postponed to \Cref{sec:prop1}.

It is worth mentioning that the radial symmetry of $\phi$ ensures that $g$ does not depend on the direction $\bx - \bx_i$.
This decoupled structure is particularly suited for GPU implementations. On the one hand, arrival times between each detection point $\bx$ and the voxels $\bx_i$ can be efficiently parallelized. On the other hand, the convolution structure is highly optimized on modern architectures.
However, any angular dependence would require evaluating the function $g$ at the corresponding time of arrival for each $(\bx, \bx_i)$ pair, destroying memory coalescing and significantly degrading computational efficiency. This is the case, for example, where $p_0$ is discretized with tensor products of one dimensional hat functions \cite{Wang_2013, ding2017efficient}.
Note that for $\phi = \indic_{[0,\kappa]}$, the kernel $g$ recovers the N-shape wave of the ball \cite{diebold1991photoacoustic}.

\subsection{Reference experimental setup}
\label{subsec:lib_setup}
In the remainder of the paper, we instantiate our framework on the photoacoustic acquisition system developed at the Laboratoire d'Imagerie Biom\'edicale (LIB) \cite{linger2023volumetric}. This system fits the modeling assumptions introduced above and provides the reference parameters $(N, K, L)$ against which the complexities of \Cref{sec:implementations} will be quoted.

\paragraph{Geometry and acquisition}
The sample is immersed in a water tank in which the speed of sound is assumed constant, $c = 1500\,\mathrm{m\,s^{-1}}$.
The initial pressure field $p_0$ is supported in a cube $\Omega$ of side $10\,\mathrm{mm}$, discretized on a uniform Cartesian grid $(200, 200, 200)$ of $N = 8 \cdot 10^{6}$ points and voxel spacing $w = 50\,\mu\mathrm{m}$.
The acoustic field is recorded by the $64$ central elements of a linear ultrasound array (L7-4, ATL, USA): each element is a transducer $0.25\,\mathrm{mm}$ wide and cylindrically focused at $25\,\mathrm{mm}$ through a $7.5\,\mathrm{mm}$ elevation aperture. The inter-element pitch is $0.298\,\mathrm{mm}$. An assembly of motorized stages sweeps the array in a translate-rotate scanning scheme over $12$ rotation angles and $15$ translation steps at each angle, so each acquisition yields $K = 64 \times 12 \times 15 = 11\,520$ signals.
The focal zone is swept across $\Omega$ for the different angles and the distance between the elements and the center of $\Omega$ is equal to or larger than the focal distance. The transducers have a limited bandwidth (2 to $10\,\mathrm{MHz}$ in reception) and a center frequency of $5\,\mathrm{MHz}$ (acoustic wavelength $\approx 0.30\,\mathrm{mm}$ in water). Each signal is sampled following \eqref{eq:measurement_discrete} at $F_s = 62.5\,\mathrm{MHz}$ and time-gated to the window during which spherical wavefronts emitted from $\Omega$ can intersect the transducer surface, yielding on average $L \approx 1\,000$ samples per element.

\paragraph{Orders of magnitude}
The forward operator $\bA : \R^N \mapsto \R^{KL}$ maps $N = 8 \cdot 10^{6}$ unknowns to $KL \approx 1.2 \cdot 10^{7}$ time samples. Storing $\bA$ as a dense matrix would require $NKL \approx 9 \cdot 10^{13}$ entries ($\approx 740\,\mathrm{TB}$ in double precision), well out of reach. Each implementation of \Cref{sec:implementations} will be analyzed against the reference parameters $(N, K, L)$ listed here.

\subsection{Choice of the radial function} \label{subsec:radial_function}

The expansion \eqref{eq:continuous_discretized_field} is supported by the well-established theory of radial basis functions \cite{wendland2005scattered}, in particular for compactly supported radial basis functions \cite{wendland1995piecewise}.
The values $\bp_0[i]$ coincide with the nodal values $p_0(\bx_i)$ when the function $\phi$ is interpolatory, that is when $\phi(0) = 1$ and $\phi$ vanishes at the neighboring nodes.
Otherwise, the coefficients $\bp_0[i]$ should be obtained by solving the linear system $\bPhi \bp_0 = \balpha$ where $\balpha = (p_0(\bx_i)) \in \R^N$ contains the nodal values and $\bPhi \in \R^{N \times N}$ is the Gram matrix associated to $\phi$ defined as $\bPhi[i,j] = \phi(\|\bx_i - \bx_j\|)$.
Whenever $\phi$ is positive definite on $\R^3$ the system is uniquely solvable \cite{wendland2005scattered}.

Since $\phi$ is compactly supported, and assuming that $\kappa$ is not too large compared to the grid point spacing, the matrix $\bPhi$ is sparse.
The system can therefore be solved efficiently using an iterative sparse linear solver.
Conversely, recovering the nodal values or converting the representation to another discretization, such as voxel representation, amounts to applying another sparse matrix to $\bp_0$.
Both operations will be negligible in cost compared with one application of $\bA$.

Popular choices include the compactly supported Wendland functions \cite{wendland1995piecewise}.
They are the minimal-degree piecewise polynomials that are positive definite in a prescribed spatial dimension, have controlled smoothness, and satisfy the compact-support requirement of \Cref{ass:setting}\ref{ass:setting:phi}.
For example, in $\R^3$, the function $\phi(r) = \max\left(1 - \tfrac{r}{\kappa},0\right)^2$ is positive definite.
Its associated native space is norm-equivalent to the Sobolev space $W^{2,2}(\R^3)$, providing a natural setting for approximation error analysis.

Two approximation regimes are commonly distinguished in the literature. In the non-stationary regime \cite{wendland2005scattered,narcowich2006sobolev}, the support radius $\kappa$ is kept fixed while the spacing between the centers tends to zero. Approximation errors measured in the native space converge at a rate determined by the smoothness of the kernel and the target functions.
In the stationary regime, where $\kappa$ is kept proportional to the spacing between the grid points as the latter tends to zero, such convergence is generally not guaranteed \cite{buhmann2003radial,schaback2006kernel} as the family $\{ \phi(\|\cdot - \bx_i\|) \}_{i \in [N]}$ is not a partition of unity.

In the present setting however, the transducer EIR acts as a band-pass filter, and this provides guidelines to choose $\kappa$.
It must be small enough for the basis functions to represent the frequencies of the transducer bandwidth.
Further decreasing $\kappa$ introduces higher-frequency components that are suppressed by the EIR, while increasing the number of basis functions $N$ and hence the computational cost.

Reducing the complexity of the method also motivates the use of the stationary regime, as it uses significantly fewer grid points than the non-stationary one.
Although it lacks convergence guarantees on $p_0$ itself, the conclusions differ for the other quantity of interest, the measured signals $s$.
The failure of partition of unity is carried in the high-frequencies which the EIR band-limits.
A complete approximation theory for this regime, quantifying the error with the band-limiting action of the EIR and leading to an optimal choice of the triple $\phi, \kappa$ and grid spacing, is beyond the scope of this paper. In what follows, $p_0$ is therefore taken as given in the form \eqref{eq:continuous_discretized_field}, so that the error bounds of \Cref{sec:implementations} quantify the discretization of the forward operator, and not that of the initial pressure itself.

\section{Implementation of the forward operator}
\label{sec:implementations}

This section is dedicated to the derivation of efficient implementations of the forward operator mapping $\bp_0 \in \R^N \mapsto \bs \in \R^{KL}$ and its adjoint.
In this work, the signal recorded by each transducer is computed independently. Consequently, to simplify the notation, the dependence on the transducer index $k$ is dropped throughout this section: $\Sigma$ denotes a generic transducer $\Sigma_k$, $m$ its signal \eqref{eq:measurement} and $\bs \in \R^L$ its samples \eqref{eq:measurement_discrete}.

Two implementations are described: a point detector approximation (\Cref{subsec:point_method}) and a scheme exploiting the transducer surface (\Cref{subsec:surface_method}).
Prior to their descriptions, we provide an overview of the implementations (\Cref{subsec:overview}).
We use the following assumption on the transducer surface $\Sigma$.

\begin{assumption} \label{ass:smoothness_transducer}
       The transducer $\Sigma$ is defined by
       \begin{equation*}
              \Sigma = \left\{ \sigma(u, v) \, | \, (u, v) \in D \right\},
       \end{equation*}
       where $D \subset \R^{2}$ is a connected bounded open set and $\sigma : D \to \R^3$ is an injective immersion that is further assumed to be in $W^{2, \infty}(D ; \R^3)$.
\end{assumption}

Since $\sigma \in W^{2,\infty}(D;\R^3)$, it follows in particular that it is $C^1$. Combined with the immersion condition, this ensures that the surface element $\dint\sigma$ is well-defined, so that surface integrals on $\Sigma$ are properly defined.
Together with injectivity of $\sigma$ and the fact that $D$ is open, these assumptions are standard hypotheses for a parametrized surface.
The additional $W^{2,\infty}$ regularity will be used later to derive approximation bounds.

\subsection{Overview}
\label{subsec:overview}

The expression of the pressure field on $\Sigma$ in \Cref{prop:point_value_pressure_field} can alternatively be written, with distributions, as

\begin{equation*}
p(\bx, t) = (g \star f(\bx, \cdot))(t), \quad \text{with } \quad f(\bx, \tau) = \sum_{i \in [N]} \frac{\bp_0[i]}{\| \bx - \bx_i \|} \delta_{\tau = \frac{\| \bx - \bx_i\|}{c}}
\end{equation*}
leading to a measured signal $m$ written as
\begin{equation} \label{eq:def_f_sigma}
       m(t) = \int_{\Sigma} p(\bx, t) \dint\sigma = g \star f_{\Sigma} (t), \quad \text{with } \quad f_{\Sigma}(\tau) = \int_{\Sigma} f(\bx, \tau) \dint\sigma.
\end{equation}

Since $f(\bx,\cdot)$ is a weighted sum of Dirac masses, $f(\bx,\cdot)$ and $f_{\Sigma}$ are understood as distributions in the time variable.
The two identities above, and the convolution $\star$ of such distributions, are defined and justified in \Cref{sec:conv_form}.

Next, the distribution $f_{\Sigma}$ will be approximated by a piecewise constant \emph{function} defined on a fine time grid.
This serves two purposes: the approximation is described by a finite vector of coefficients, and the convolution with $g$, the convolution with the EIR $e$ \eqref{eq:measurement} and the sampling \eqref{eq:def_operatorS} are then computed \emph{exactly} by a discrete convolution.

\paragraph{Approximation of $f_{\Sigma}$}

Let $U \in \N$ be an upsampling factor and consider the upsampled time grid $(\tau_l)_{l=0}^{UL}$, sampled at frequency $U F_s$, that is $\tau_0 = t_0$ and $\tau_l = t_0 + l/(UF_s)$ for $l \in [UL]$.
Both methods developed in this paper replace $f_{\Sigma}$ by a function that is piecewise constant on this grid, that is of the form
\begin{equation} \label{eq:approx_u_sigma}
       \widehat{f}_{\Sigma} = \sum_{l \in [UL]} \bbf[l] \indic_{[\tau_{l-1},\tau_l)}, \qquad \bbf \in \R^{UL}.
\end{equation}
They differ only in the way the coefficients $\bbf$ are obtained: by a quadrature on $\Sigma$ (\Cref{subsec:point_method}) or by exploiting its surface (\Cref{subsec:surface_method}).
The rest of this overview uses \eqref{eq:approx_u_sigma} alone, so $\widehat{f}_{\Sigma}$ and $\bbf$ denote interchangeably the quantities produced by either method.
The usefulness of this upsampled grid will become clear in the following.

\paragraph{Discrete convolution}

Substituting the approximation $\widehat{m} = g \star \widehat{f}_\Sigma$ for the exact signal $m$ in \eqref{eq:measurement} yields the approximate measurements

\begin{equation} \label{eq:def_h}
       \widehat{\bs} = \mathcal{S}(e \star \widehat{m}) = \mathcal{S}(h \star \widehat{f}_{\Sigma}), \quad \text{with} \quad h = e \star g,
\end{equation}
that is $\widehat{\bs}[l] = F_s \int_{t_{l-1}}^{t_l} (h \star \widehat{f}_{\Sigma}) (t) \dint t$.
The system kernel $h$ gathers the radial function $\phi$, through $g$, and the EIR $e$.

Because $\widehat{f}_{\Sigma}$ is a piecewise constant function on the upsampled grid \eqref{eq:approx_u_sigma}, the coefficients $\widehat{\bs}$ can be exactly computed through

\begin{equation} \label{eq:discrete_conv}
       \widehat{\bs} = \dsop (\bh \star \bbf)
\end{equation}
where $\star$, $\bh$ and $\dsop$ are defined as follows.
First, $\star$ denotes the discrete convolution of sequences indexed by $\Z$, with $\bbf \in \R^{UL}$ being embedded in $\Z$ as a sequence of support $[UL]$, that is
\begin{equation} \label{eq:discrete_conv_2}
       (\bh \star \bbf)[l_2] = \sum_{l_1 \in [UL]} \bh[l_2 - l_1] \bbf[l_1],  \, \forall l_2 \in \Z.
\end{equation}
Second, $\bh : \Z \to \R$ is the discrete kernel defined by
\begin{equation} \label{eq:def_bh}
       \bh[l_2 - l_1] = U F_s \int_{0}^{1/(U F_s)} \int_{0}^{1/(U F_s)} h\left(t-t'+\tfrac{l_2 - l_1}{UF_s}\right) \dint t \, \dint t'.
\end{equation}
Third, $\dsop$ is the downsampling by the factor $U$, which maps a sequence $\bv : \Z \to \R$ to the averages of its consecutive blocks of $U$ samples,
\begin{equation} \label{eq:def_downsampling}
       (\dsop \bv) [l] = \frac{1}{U} \sum_{l' \in [U]} \bv[U(l-1) + l'], \quad \forall l \in \Z.
\end{equation}
The identity \eqref{eq:discrete_conv} is a classical result of signal processing detailed in \Cref{sec:discrete_conv}, see also \cite[Theorem 1]{chacko2013discretization}.

Moreover, since $g$ is supported in $[-\kappa c^{-1}, \kappa c^{-1}]$ and $e$ is supported in $[0,T_e]$, the vector $\bh : \Z \to \R$ is supported in $\Z \cap \left[-\kappa c^{-1} UF_s -1, (\kappa c^{-1} + T_e) UF_s +1 \right]$. As a consequence, the vector $\widehat{\bs}$ is also compactly supported and may have nonzero entries at negative indices.
In the remainder of the paper, the indexing is chosen such that the support of $\widehat{\bs}$ is contained in $[L]$.
This amounts to shifting the indexing of $\bbf$ and performing suitable left and right padding of $\bbf$.
With this convention, both $\bbf$ and $\bh \star \bbf$ can be stored by vectors of length $UL$, with their supports fully contained. Similarly, $\widehat{\bs}$ can be stored by a vector of length $L$ obtained by downsampling $\bh \star \bbf$ following \eqref{eq:def_downsampling} that simplifies to
\begin{equation*} \label{eq:downsampling_conv}
       \dsop(\bh \star \bbf) [l] = \frac{1}{U} \sum_{l' \in [U]} (\bh \star \bbf)[U(l-1) + l'], \quad \forall l \in [L].
\end{equation*}

\paragraph{Summary}

The construction of the approximating $\widehat{\bs}$ can therefore be summarized by the following diagram: \\

\begin{center}
\resizebox{0.96\linewidth}{!}{%
              \begin{tikzpicture}
  \node (M) [noblock, minimum width=0pt] {$m$} ;
  \node (FF)   [block,right=0.5cm of M]                        {\Cref{prop:point_value_pressure_field}};
  \node (WM) [noblock, right=0.5cm of FF] {$\displaystyle m = g \star f_{\Sigma}$ \eqref{eq:def_f_sigma}} ;

 \node (PA)   [block,  right=0.5cm of WM]       {Piecewise Approximation};
 \node (HM)   [noblock,  right=0.5cm of PA]       {$ \displaystyle \widehat{m} = g \star \widehat{f}_{\Sigma}$ \eqref{eq:approx_u_sigma}};

 \node (S) [noblock, minimum width=0pt, below=1.5cm of M, font=\small] {$\bs$} ;
 \node (HS) [noblock, below=1.5cm of HM, font=\small] {$\widehat{\bs} = \dsop(\bh \star \bbf)$ \eqref{eq:discrete_conv}} ;

 \draw[arr] (M) -- (FF) ;

 \draw[arr] (FF) -- (WM) ;
 \draw[arr] (WM) -- (PA) ;
 \draw[arr] (PA) -- (HM) ;

 \draw[arr] (M) -- node[right,font=\footnotesize]{$ \displaystyle \mathcal{S} (e \star  m)$ \eqref{eq:def_operatorS}} (S) ;
 \draw[arr] (HM) -- node[right,font=\footnotesize]{$ \displaystyle \mathcal{S} (e \star  \widehat{m})$ \eqref{eq:def_operatorS}} (HS) ;
\end{tikzpicture}
}
\end{center}
The integrated pressure field $m$ is first written as $m = g \star f_{\Sigma}$ using the result of \Cref{prop:point_value_pressure_field}. The resulting field is then approximated by $\widehat{m} = g \star \widehat{f}_{\Sigma}$, where $\widehat{f}_{\Sigma}$ is piecewise constant on the upsampled grid \eqref{eq:approx_u_sigma} and thus entirely described by the coefficients $\bbf$.
The latter structure is particularly convenient, since it allows one to exactly compute the coefficients $\widehat{\bs} = \mathcal{S} (e \star \widehat{m})$ by a discrete convolution between $\bbf$ and a discrete filter $\bh$ defined in \eqref{eq:def_bh}.

The proposed approximations build a forward operator $\bA : \bp_0 \to \widehat{\bs}$ by decoupling the wave propagation from the radial discretization and the EIR of the transducer.
Specifically, the propagation from the grid points $(\bx_i)$ to the transducer surface is encoded in $\bbf$ while the effects of the radial discretization $\phi$ and the EIR are incorporated through the convolution with the filter $\bh$.
The sampling operator $\mathcal{S}$ is naturally accounted for by the convolution and decimation operations.

From a numerical perspective, the application of $\bA$ can be formalized in \Cref{alg:forward_global}, which assembles the $K$ transducers and where the index $k$ is therefore restored: a subscript $k$, as in $\bbf_k$, refers to the quantity computed for the transducer $\Sigma_k$.
The kernel $\bh$ is computed once and reused across all transducers and for each evaluation of $\bA$.
The convolution $\bh \star \bbf_k$ can be implemented via FFT, so that its complexity is $O(UL \log (UL))$ operations.
Moreover, the $K$ convolutions can be computed simultaneously by exploiting batched FFT implementations, thereby taking advantage of parallelization.

The overall complexity of the forward operator is also driven by the cost of evaluating the coefficients $\bbf_k$.
Numerical experiments, including those presented in \Cref{sec:experiments}, suggest that the computation of the $\bbf_k$ is the dominant computational cost.

Iterative reconstruction applies the adjoint $\bA^{*} : \R^{KL} \to \R^N$ as often as $\bA$. It transposes the pipeline above and reuses the same geometric routines, so that every method below provides $\bA^{*}$ at the cost of $\bA$; the corresponding algorithms are given in \Cref{app:adjoint}.

\begin{algorithm}[ht]
\caption{Forward operator $\bA : \bp_0 \mapsto \widehat{\bs}$.}
\label{alg:forward_global}
\begin{algorithmic}[1]
\Require $\bp_0 \in \R^N$ ; $(\bx_i)_{i \in [N]}$ ; $(\Sigma_k)_{k\in[K]}$ ; a precomputed $\bh$ following \eqref{eq:def_bh} ;
       \Ensure Discrete measurements $\widehat{\bs} \in \R^{KL}$.
       \Statex
       \For{$k = 1, \ldots, K$} \Comment{Loop over transducers}
              \State Compute $\bbf_k \in \R^{UL}$ following \eqref{eq:approx_u_sigma}. \Comment{Method-specific: \Cref{subsec:point_method,subsec:surface_method}}
              \State $\widehat{\bs}_k \gets \dsop\!\left(\bh \star \bbf_k\right)$. \Comment{FFT convolution, then downsampling \eqref{eq:discrete_conv}}
       \EndFor
       \State \Return $\widehat{\bs} = \big(\widehat{\bs}_k\big)_{k \in [K]} \in \R^{KL}$.
\end{algorithmic}
\end{algorithm}

The two methods for computing the coefficients $\bbf_k$ are now described. For each of them, the computational complexity is analyzed, and an error bound between the measurements $\bs$ and the approximations $\widehat{\bs}$ is provided.
It will be shown that the upsampling factor $U$ is one of the parameters influencing the accuracy--complexity trade-off: larger values of $U$ can improve accuracy at the expense of increased computational cost.

\paragraph{Error bounds}
The error bounds will be obtained by applying the convolution with the system kernel $h = e \star g$ \eqref{eq:def_h} to the approximations of $f_\Sigma$.
The regularity of $g$ is controlled by the choice of $\phi$, and that of $e$ by the detector model: a damped oscillator gives $e \in W^{1,\infty}(\R)$, and the third-order Butterworth band-pass filter used in \Cref{sec:experiments} gives $e \in W^{2,\infty}(\R)$.

\subsection{Point detector approximation}
\label{subsec:point_method}

The first implementation relies on a discretization of the transducer surface $\Sigma$ by a family of points, which is the most common approach in the photoacoustic literature \cite{ding2017efficient,ding2020model,li2026slingbag} as it directly leverages the results of \Cref{prop:point_value_pressure_field}.
It consists in approximating
\begin{equation*}
       m(t) = \int_{\Sigma} p(\bx, t) \dint\sigma
\end{equation*}
by $\widetilde{m}$ through a quadrature rule on $Q$ sample points $(\bq_j)_{j \in [Q]}$ of $\Sigma$ and areas $(\Delta \bq_j)_{j \in [Q]}$, that is
\begin{equation*}
       \begin{aligned}
              \widetilde{m}(t) &= \sum_{j \in [Q]} p(\bq_j, t) \Delta \bq_j \\
              & = \sum_{j \in [Q], i \in [N]} g\left( t - \frac{\| \bq_j - \bx_i\|}{c} \right) \frac{\bp_0[i] \Delta \bq_j}{\| \bq_j - \bx_i\|} \\
              & = (g \star \widetilde{f}_\Sigma)(t),
       \end{aligned}
\end{equation*}
with $\widetilde{f}_{\Sigma}$ being defined by
\begin{equation*}
       \widetilde{f}_\Sigma =  \sum_{j \in [Q], i \in [N]} \frac{\bp_0[i] \Delta \bq_j}{\| \bq_j - \bx_i\|} \delta_{\frac{\| \bq_j - \bx_i\|}{c}}.
\end{equation*}
This equality is justified in \Cref{lem:conv_measure} of \Cref{sec:conv_form}.
The coefficients $\bbf$ are then set to
\begin{align}
              \bbf[l] &= U F_s \int_{\tau_{l-1}}^{\tau_l} \widetilde{f}_\Sigma (t) \dint t  \notag\\
              &=  U F_s \sum_{j \in [Q], i \in [N]} \frac{\bp_0[i] \Delta \bq_j}{\| \bq_j - \bx_i\|} \indic_{[\tau_{l-1}, \tau_l)}\left(\frac{\| \bq_j - \bx_i\|}{c}\right). \label{eq:def_f_point}
\end{align}
The last equality is justified in \Cref{lem:dirac_indicator} of \Cref{sec:conv_form}.

As can be seen from \eqref{eq:def_f_point}, the proposed method can be interpreted as summing the contributions of a collection of point detectors distributed across the transducer surface $\Sigma$.
This interpretation forms the basis of the implementation described in \Cref{alg:point_based}.
To compute $\bbf$, the distance $\| \bq_j - \bx_i\|$ must be evaluated between every point detector and every voxel.
This distance determines the time interval $[\tau_{l-1}, \tau_l)$, with $l = \lfloor (c^{-1}\| \bq_j - \bx_i\| - t_0) U F_s \rfloor + 1$, to which the contribution $\bp_0[i] \Delta \bq_j / \| \bq_j - \bx_i\|$ is accumulated.
The complexity of the algorithm is therefore $O(NQ)$.
Its adjoint, obtained by replacing the scatter into $\bbf$ by a gather over the voxels, is given in \Cref{alg:adjoint_point} of \Cref{app:adjoint} and has the same $O(NQ)$ complexity.

\begin{algorithm}[ht]
\caption{Point detector approximation of $\bbf$ for one transducer $\Sigma$.}
\label{alg:point_based}
\begin{algorithmic}[1]
       \Require $\bp_0 \in \R^N$ ; $(\bx_i)_{i \in [N]}$; quadrature $(\bq_j, \Delta \bq_j)_{j \in [Q]}$ of $\Sigma$;
       \Ensure Geometric coefficients $\bbf \in \R^{UL}$.
       \Statex
       \State \textbf{Initialization}: $\bbf \gets 0 \in \R^{UL}$
       \For{$i \in [N]$} \Comment{Loop over voxels}
              \For{$j \in [Q]$} \Comment{Loop over point detectors of $\Sigma$}
                     \State $r \gets \|\bq_j - \bx_i\|$.
                     \State $l \gets \lfloor (r/c - t_0)\, U F_s \rfloor + 1$.
                     \State $\bbf[l] \gets \bbf[l] + \bp_0[i] \Delta \bq_j / r$. \Comment{\eqref{eq:def_f_point}}
              \EndFor
       \EndFor
       \State $\bbf \gets U F_s \cdot \bbf$.
       \State \Return $\bbf$.
\end{algorithmic}
\end{algorithm}

The following proposition quantifies the error introduced by the point detector approximation.
\begin{proposition}
\label{prop:approximation_points}
Suppose that $D$ is partitioned into $Q$ convex cells $(D_j)_{j \in [Q]}$ such that $\forall j \in [Q]$, $\diam(D_j) \lesssim Q^{-\frac{1}{2}}$.
Suppose furthermore that $\bq_j = \sigma(u_j,v_j)$, with $(u_j,v_j)$ the barycenter of $D_j$ with respect to $\dint \sigma$ and define the area $\Delta \bq_j = | \sigma(D_j) |$.

Assuming that $h \in W^{2, \infty}(\R)$ together with \Cref{ass:setting,ass:smoothness_transducer}, the output $\widehat{\bs}$ of \Cref{alg:forward_global} restricted to one transducer with subroutine \Cref{alg:point_based} satisfies
       \begin{equation*}
              \| \bs - \widehat{\bs} \|_{\infty} \lesssim  |\Sigma| \, \| \bp_0 \|_1 \left( C(h, \dist_{\min}, \sigma) \frac{1}{Q} + \|h'\|_{L^\infty} \frac{1}{\dist_{\min} U F_s} \right),
       \end{equation*}
where $|\Sigma|$ denotes the area of $\Sigma$ and the constant $C(h, \dist_{\min}, \sigma)$ is made explicit in the proof.
\end{proposition}

\begin{proof}
       See \Cref{app:proof_prop_23}.
\end{proof}

The assumptions on the partition of $D$ are quite general and are satisfied, in particular, when $D$ is a rectangular domain and $(D_j)_{j\in[Q]}$ is the partition obtained as the tensor product of partitions of the two directions into $Q_1$ and $Q_2$ intervals, respectively, giving $Q = Q_1 Q_2$ cells. In this case, the barycenters $(u_j, v_j)$ are the centers of the subrectangles $D_j$. Moreover, if $Q_1$ and $Q_2$ are proportional, the diameters of $D_j$ are also proportional to $Q^{-\frac{1}{2}}$.

The error bound in \Cref{prop:approximation_points} consists of two terms.
The first one, of order $Q^{-1}$, is due to the quadrature approximation of the integral, over $\Sigma$, of the filtered pressure field $e \star p$.
Through the constant $C(h, \dist_{\min}, \sigma)$, this error depends on the curvature of the transducer surface and the curvature of the spherical waves centered at $\bx_i$, as well as on the second order smoothness of $h$.
The second term in $(U F_s)^{-1}$ arises from the quantization of the arrival times $ c^{-1} \| \bq_j - \bx_i \| $ into the time intervals $[\tau_{l-1}, \tau_l)$, and depends on the Lipschitz constant of $h$.
While $F_s$ is fixed by the acquisition system, the upsampling factor $U$ is a user-defined parameter that can be increased to improve the accuracy of the implementation.

The assumption $h \in W^{2, \infty}(\R)$ describes a shared regularity between $e$ and $g$. Since the derivatives of $h$ can be transferred onto $e$ or $g$, this can be achieved when
\begin{enumerate}
    \item $\phi \in W^{2, \infty}(\R_+)$ and $e \in L^1(\R)$ for example with Wendland functions with sufficient smoothness, see \Cref{subsec:radial_function};
    \item $\phi \in W^{1, \infty}(\R_+)$ and $e \in W^{1, 1}(\R)$ for example with $\phi$ the cone function;
    \item $\phi \in L^{\infty}(\R_+)$ and $e \in W^{2, 1}(\R)$ for $\phi$ the indicator of $[0,\kappa]$.
\end{enumerate}
Moreover, when $h$ fails to be in $W^{2, \infty}(\R)$ but still in $W^{1, \infty}(\R)$, a similar result holds with $Q^{-1}$ replaced by $Q^{-\frac{1}{2}}$, using only the Lipschitz continuity of $h$ (see the proof in \Cref{app:proof_prop_23}).

\subsection{Exploiting the surface of the transducer} \label{subsec:surface_method}

The scheme presented in \Cref{subsec:point_method} computes $\bbf$ by approximating the surface integral defining $f_{\Sigma}$ with a quadrature on $\Sigma$.
In this section, we introduce an alternative scheme that treats the surface integral analytically, therefore avoiding any discretization of $\Sigma$.
The analytical treatment depends on the specific transducer geometry and will therefore be instantiated for planar and cylindrical transducers in \Cref{sec:usual_surfaces}.
Nevertheless, the present section establishes a systematic procedure to derive these analytical formulations.

Recall that $m = g \star f_\Sigma$, as defined in \eqref{eq:def_f_sigma}.
The piecewise constant approximation of $f_\Sigma$ is obtained by integration on $[\tau_{l-1}, \tau_l)$ giving
\begin{equation*}
       \begin{aligned}
              \overline{\bbf}[l]  &= U F_s \int_{\tau_{l-1}}^{\tau_l} f_\Sigma (t) \dint t = \sum_{i \in [N]} \bp_{0}[i]\, \overline{\bbf}_{i}[l],
        \end{aligned}
\end{equation*}
with
\begin{equation*}
       \begin{aligned}
              \overline{\bbf}_{i}[l] = U F_{s} \int_{\Sigma}
                       \frac{1}{\|\bx - \bx_i\|}\,
                       \indic_{[\tau_{l-1},\, \tau_l)} \left( \frac{\|\bx - \bx_i\|}{c} \right)\,
                       \dint\sigma(\bx).
        \end{aligned}
\end{equation*}
This is obtained by interchanging again the order of integration and using the fact that
\begin{equation*}
       \int_{\tau_{l-1}}^{\tau_l} \delta_{\tau = \frac{\|\bx - \bx_i\|}{c}} \dint \tau = \indic_{[\tau_{l-1},\, \tau_l)} \left( \frac{\|\bx - \bx_i\|}{c} \right),
\end{equation*}
see again \Cref{lem:dirac_indicator} in \Cref{sec:conv_form} for more details. The coefficient $\overline{\bbf}_{i}$ can alternatively be written as
\begin{equation} \label{eq:def:bffi}
       \begin{aligned}
              \overline{\bbf}_{i}[l] = U F_{s} \int_{\mathcal{M}_i^l} \frac{1}{\|\bx - \bx_i\|} \dint\sigma(\bx), \quad \text{with } \, \mathcal{M}_i^l = \left\{ \bx \in \Sigma \,| \, c\tau_{l-1} \leq \|\bx - \bx_i\| < c\tau_l \right\}.
        \end{aligned}
\end{equation}
Moreover, on this set, $\|\bx - \bx_i\|^{-1}$ can be approximated by $\left(c \tau_{l-\frac{1}{2}}\right)^{-1}$ with a small error, where $\tau_{l-\frac{1}{2}} = \frac{\tau_{l-1} + \tau_l}{2}$ is the midpoint of $[\tau_{l-1}, \tau_l)$.
Overall, the proposed method sets
\begin{equation} \label{eq:ui_area}
       \bbf_{i}[l] = \frac{U F_{s}}{c\,\tau_{l-\frac{1}{2}}} \left| \mathcal{M}_i^l \right|,
       \qquad
       \bbf[l] = \sum_{i \in [N]} \bp_0[i]\, \bbf_{i}[l],
\end{equation}
where $\left| \mathcal{M}_i^l \right|$ denotes the area of $\mathcal{M}_i^l$.
Geometrically, $\mathcal{M}_i^l$ describes the domain of $\Sigma$ that is swept by the spherical wave emitted by $\bx_i$ with radii in $[c \tau_{l-1}, c \tau_l)$, as depicted in \Cref{fig:M_set}.
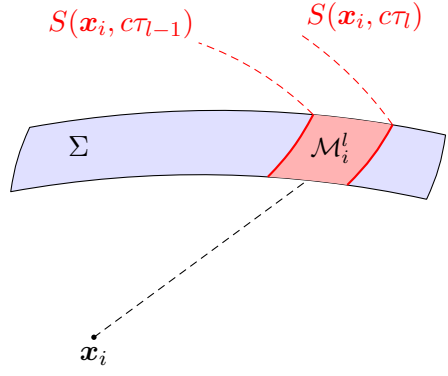
\begin{figure}[!htb] \centering
       \begin{tikzpicture}[>=latex, line join=round,
    declare function={
      yfront(\x) = \apex - \RR + sqrt(\RR*\RR - \x*\x);
      afront(\x) = acos(\x/\RR);
    }]
  \def\RR{14}
  \def\apex{1.0}
  \def\dx{0.25}\def\dy{0.85}
  \def\xL{-2.5}\def\xR{3.0}
  \def\xa{0.9}\def\xb{1.95}
  \def\bskew{0.35}
  \def\vx{-1.4}\def\vy{-1.15}

  \fill[blue!12, draw=black, thin]
    ({\xL},{yfront(\xL)}) arc[start angle={afront(\xL)}, end angle={afront(\xR)}, radius=\RR]
    to[bend right=8] ({\xR+\dx},{yfront(\xR)+\dy})
    arc[start angle={afront(\xR)}, end angle={afront(\xL)}, radius=\RR]
    to[bend right=8] cycle;
  \node[above] at ({\xL+0.9},{yfront(\xL-2)+\dy}) {$\Sigma$};

  \fill[red!30]
    ({\xa},{yfront(\xa)}) to[bend right=8] ({\xa+\bskew+\dx},{yfront(\xa+\bskew)+\dy})
    arc[start angle={afront(\xa+\bskew)}, end angle={afront(\xb+\bskew)}, radius=\RR]
    to[bend left=8] ({\xb},{yfront(\xb)})
    arc[start angle={afront(\xb)}, end angle={afront(\xa)}, radius=\RR] -- cycle;
  \draw[red,thick] ({\xa},{yfront(\xa)}) to[bend right=8] ({\xa+\bskew+\dx},{yfront(\xa+\bskew)+\dy});
  \draw[red,thick] ({\xb},{yfront(\xb)}) to[bend right=8] ({\xb+\bskew+\dx},{yfront(\xb+\bskew)+\dy});
  \node[scale=0.9] at ({(\xa+\xb)/2+(\dx+\bskew)/2},{yfront((\xa+\xb)/2)+\dy/2}) {$\mathcal{M}_i^l$};

  \pgfmathsetmacro\Bax{\xa+\bskew+\dx}\pgfmathsetmacro\Bay{yfront(\xa+\bskew)+\dy}
  \pgfmathsetmacro\Bbx{\xb+\bskew+\dx}\pgfmathsetmacro\Bby{yfront(\xb+\bskew)+\dy}
  \pgfmathsetmacro\ra{veclen(\Bax-\vx,\Bay-\vy)}\pgfmathsetmacro\anga{atan2(\Bay-\vy,\Bax-\vx)}
  \pgfmathsetmacro\rb{veclen(\Bbx-\vx,\Bby-\vy)}\pgfmathsetmacro\angb{atan2(\Bby-\vy,\Bbx-\vx)}
  \draw[red, densely dashed]
    ($(\vx,\vy)+(\anga:\ra)$) arc[start angle=\anga, end angle=\anga+25, radius=\ra]
    node[above left=-3pt] {$S(\bx_i, c\tau_{l-1})$};
  \draw[red, densely dashed]
    ($(\vx,\vy)+(\angb:\rb)$) arc[start angle=\angb, end angle=\angb+20, radius=\rb]
    node[above right=-3pt] {$S(\bx_i, c\tau_{l})$};

  \draw[densely dashed] (\vx,\vy) -- ({(\xa+\xb)/2},{yfront((\xa+\xb)/2)});
  \fill (\vx,\vy) circle (1pt) node[below] {$\bx_i$};
\end{tikzpicture}
       \caption{The set $\mathcal{M}_i^l$ (red) on a cylindrical transducer $\Sigma$, delimited by the intersections of $\Sigma$ with the spheres $S(\bx_i, c\tau_{l-1})$ and $S(\bx_i, c\tau_l)$ centered at the voxel $\bx_i$.}
       \label{fig:M_set}
\end{figure}

\begin{remark}[Connection with the SIR]
The coefficients $\overline{\bbf}_i$ are proportional to the averages of the spatial impulse response (SIR) of the transducer over the time intervals $[\tau_{l-1},\tau_l)$.
    In \cite{stepanishen1971transient}, the SIR of $\Sigma$ at a point $\bx$ is defined, up to a constant factor, as
    \begin{equation*}
        \mathrm{SIR}(\bx,t)= \int_\Sigma \frac{1}{2\pi\,\|\bx'-\bx\|} \delta\left(t-\frac{\|\bx'-\bx\|}{c}\right) \dint \sigma(\bx').
    \end{equation*}
Integrating over the time bin $[\tau_{l-1},\tau_l)$ and using the same
interchange of integrals as in \eqref{eq:def:bffi} yields
\begin{equation*}
    \int_{\tau_{l-1}}^{\tau_l} \mathrm{SIR}(\bx_i,t) \dint t = \frac{1}{2\pi} \int_{\mathcal{M}_i^l}\frac{1}{\|\bx-\bx_i\|} \dint \sigma(\bx)
\end{equation*}
so that
\begin{equation*}
  \overline{\bbf}_i[l]=2\pi UF_s \int_{\tau_{l-1}}^{\tau_l} \mathrm{SIR}(\bx_i,t) \dint t.
\end{equation*}
The coefficients $\bbf_i$ are therefore proportional to an approximation of these averages.
\end{remark}

These derivations lead to the following \Cref{alg:surface_based}, which computes the areas $\left| \mathcal{M}_i^l \right|$ for each time interval and accumulates the corresponding contributions \eqref{eq:ui_area} in $\bbf$.
The range of time intervals for which $\mathcal{M}_i^l$ is nonempty can be determined a priori. Indeed, $\mathcal{M}_i^l \neq \emptyset$ if and only if the sphere $S(\bx_i,ct)$ intersects the transducer surface $\Sigma$ for some $t \in [\tau_{l-1}, \tau_l)$.
More precisely, the quantities
\begin{equation} \label{eq:rmin_rmax}
       \begin{aligned}
       r_{\min}(i) &= \min_{\bx \in \Sigma} \|\bx - \bx_i\|, & l_{\min}(i) &= \lfloor \left( r_{\min}(i) c^{-1} - t_0\right) UF_s \rfloor + 1 \\
       r_{\max}(i) &= \max_{\bx \in \Sigma} \|\bx - \bx_i\|, & l_{\max}(i) &= \lceil \left( r_{\max}(i) c^{-1} - t_0 \right) UF_s \rceil
       \end{aligned}
\end{equation}
can be computed analytically for common surface geometries, see \Cref{sec:usual_surfaces}.
Its adjoint, which reuses the same bounds $l_{\min}(i), l_{\max}(i)$ and the same areas $\left| \mathcal{M}_i^l \right|$, is given in \Cref{alg:adjoint_surface} of \Cref{app:adjoint}.

\begin{algorithm}[ht]
\caption{Surface approximation of $\bbf$ for one transducer $\Sigma$.}
\label{alg:surface_based}
\begin{algorithmic}[1]
       \Require $\bp_0 \in \R^N$ ; $(\bx_i)_{i \in [N]}$ ; transducer surface $\Sigma$
       \Ensure Geometric coefficients $\bbf \in \R^{UL}$.
       \State \textbf{Initialization:} $\bbf \gets 0 \in \R^{UL}$.
       \For{$i \in [N]$} \Comment{Loop over voxels}
       \State Compute $l_{\min}(i)$ and $l_{\max}(i)$ as in \eqref{eq:rmin_rmax}
              \For{$l = l_{\min}(i), \ldots, l_{\max}(i)$} \Comment{Loop over time steps}
                     \State Compute $\left| \mathcal{M}_i^l \right|$ in \eqref{eq:def:bffi} \Comment{Surface dependent: \Cref{sec:usual_surfaces}}
                     \State $\bbf[l] \gets \bbf[l] + \bp_0[i]\, \left| \mathcal{M}_i^l \right| / (c\,\tau_{l-\frac{1}{2}})$ \Comment{\eqref{eq:ui_area}}
              \EndFor
       \EndFor
       \State $\bbf \gets U F_s \cdot \bbf$.
       \State \Return $\bbf$.
\end{algorithmic}
\end{algorithm}

The following proposition quantifies the error introduced by the proposed approximation.
\begin{proposition} \label{prop:approximation_surface}
Assuming that $h \in W^{1, \infty}(\R)$ together with \Cref{ass:setting,ass:smoothness_transducer}, the output $\widehat{\bs}$ of \Cref{alg:forward_global} restricted to one transducer with subroutine \Cref{alg:surface_based} satisfies
       \begin{equation*}
              \| \bs - \widehat{\bs} \|_{\infty} \lesssim C(h, \dist_{\min}) |\Sigma| \, \| \bp_0 \|_1  \frac{1}{\dist_{\min} (U F_s)}
       \end{equation*}
as soon as $c (UF_s)^{-1} \leq \dist_{\min}$. The constant $C(h,\dist_{\min})$ is made explicit in the proof.
\end{proposition}

\begin{proof}
       See \Cref{app:proof_prop_23}.
\end{proof}

The error bound in \Cref{prop:approximation_surface} is of order $\dist_{\min}^{-1} (UF_s)^{-1}$ and comes from two sources.
First, it arises from the quantization of the arrival times $c^{-1}\| \bx - \bx_i\|$, for $\bx \in \Sigma$, into the time intervals $[\tau_{l-1}, \tau_l)$.
Second, it is due to the approximation of $\| \bx - \bx_i\|^{-1}$ by $(c \tau_{l - \frac{1}{2}})^{-1}$ over $[\tau_{l-1}, \tau_l)$.
While $F_s$ is fixed by the acquisition system, the upsampling factor $U$ is a user-defined parameter that can be increased to improve the accuracy of the implementation.

The condition $c (UF_s)^{-1} \leq \dist_{\min}$  is technical and always satisfied in practice. It ensures that the midpoint radius $c\tau_{l-\frac12}$ remains comparable to $\dist_{\min}$.

The smoothness assumption $h \in W^{1, \infty}(\R)$ is weaker than that of \Cref{prop:approximation_points}.
It is achieved when
\begin{enumerate}
    \item $\phi \in W^{1, \infty}(\R_+)$ and $e \in L^1(\R)$, for example for $\phi$ the cone function or any Wendland function, see \Cref{subsec:radial_function};
    \item $\phi \in L^{\infty}(\R_+)$ and $e \in W^{1, 1}(\R)$, for $\phi$ the indicator of $[0,\kappa]$.
\end{enumerate}

The complexity of the algorithm is less straightforward to analyze.
Indeed, for every voxel $i$ and every transducer $k$, the routine computing the areas $\left| \mathcal{M}_i^l \right|$ has to be called $O(l_{\max}(i) - l_{\min}(i) + 1)$ times.
The complexity of computing one $\left| \mathcal{M}_i^l \right|$ primarily depends on the surface $\Sigma$ but also on the relative position of $\bx_i$ with respect to the transducer.
We assume it to be uniformly bounded by $C_{\mathcal{M}}$ for all $i$ and $l$.

Letting $\overline{B}$ be the average over the voxels of the number of time intervals $l_{\max}(i) - l_{\min}(i) + 1$ defined in \eqref{eq:rmin_rmax}, the complexity of \Cref{alg:surface_based} is $O\left(N \overline{B} C_\mathcal{M}\right)$.
Since $\overline{B}$ grows linearly with $U$, this complexity increases with the upsampling factor.
The upsampling factor thus sets the accuracy--complexity trade-off.

\begin{remark} \label{rmk:coarser_area_grid}
    Increasing $U$ can have a dramatic impact on the computational cost of the method.
    The following strategy takes advantage of an upsampled grid while maintaining a competitive complexity.
    One advantage of the upsampled grid is that it provides a finer localization of the wavefront arrival times, thus reducing the quantization error, in particular in the determination of $l_{\min}$ and $l_{\max}$.
    However, it may not be necessary to evaluate $\left| \mathcal{M}_i^l \right|$ for every $l$ to compute the coefficients $\bbf$ on the fine grid.
    A coarser evaluation of these areas may be sufficient to achieve a satisfactory accuracy.

    To this end, it is possible to fix a prescribed step $\Delta l \in [U]$, and evaluate instead the areas
    \begin{equation*}
        \left| \mathcal{M}_i^{l-\Delta l+1 : l} \right| = \left|  \left\{ \bx \in \Sigma \, | \, c \tau_{l-\Delta l} \leq \| \bx - \bx_i\| < c \tau_l \right \} \right| = \sum_{l' = l - \Delta l + 1}^{l} \left| \mathcal{M}_i^{l
        '} \right|.
    \end{equation*}
    This area is then split evenly over the $\Delta l$ time intervals $[\tau_{l'-1}, \tau_{l'})$ for $l - \Delta l + 1 \leq l' \leq l$.

    This reduces the number of area evaluations by a factor $\Delta l$ while preserving the time-of-arrival resolution of the fine grid. It does, however, introduce an approximation error that must be taken into account.
    Setting $\Delta l = 1$ recovers \Cref{alg:surface_based} exactly, whereas $\Delta l = U$ evaluates the areas only at the native sampling frequency $F_s$, resulting in the lowest computational cost.

    In the numerical experiments of this paper, setting $\Delta l = U$ already provided the best trade-off.
    However, this choice may be system dependent.
    In cases where more accuracy is needed, adaptive step sizes could be used to account for the sharp variations in $\left|\mathcal{M}_i^l\right|$ occurring near the first and last time steps at which the spherical waves intersect $\Sigma$.
\end{remark}

\subsection{Comparison of the two implementations}
\label{subsec:comparison}

The two implementations only differ in the computation of the coefficients $\bbf$, the convolution and downsampling steps of \Cref{alg:forward_global} being identical. The comparison therefore focuses on this computation and is summarized in \Cref{tab:methods_comparison}.

\begin{table}[!htb]
	\centering
	\begin{tabular}{l l l}
		\toprule
		 & Point detector (\Cref{subsec:point_method}) & Surface (\Cref{subsec:surface_method}) \\
		\midrule
		Applicability & any quadrature of $\Sigma$ & geometry-specific (\Cref{sec:usual_surfaces}) \\
		Error terms & $Q^{-1} + (U F_s)^{-1}$ & $(U F_s)^{-1}$ \\
		Accuracy parameter & $Q, U$ & $U$ \\
		Computation of $\bbf$ & $O(N Q)$ & $O(N \overline{B} \, C_{\mathcal{M}})$ \\
   		Memory & $O(UL)$ & $O(UL)$ \\
		\bottomrule
	\end{tabular}
	\caption{Comparison of the two implementations of the coefficients $\bbf$ for one transducer.}
	\label{tab:methods_comparison}
\end{table}

The point detector approximation applies to any transducer geometry: it only requires a quadrature $(\bq_j, \Delta \bq_j)_{j \in [Q]}$ of $\Sigma$, and \Cref{alg:point_based} is otherwise identical for all geometries, making it simple to implement. In contrast, the surface method relies on an analytical evaluation of the areas $\left| \mathcal{M}_i^l \right|$, which must be derived for each transducer geometry. Such derivations are provided in \Cref{sec:usual_surfaces} for planar and cylindrical transducers, at the price of a more involved implementation.

In terms of accuracy, the error bound of \Cref{prop:approximation_points} contains a quadrature term of order $Q^{-1}$ which is absent from \Cref{prop:approximation_surface}, the surface integral being treated analytically by the surface method. The accuracy of the point detector approximation is therefore controlled by two parameters, $Q$ and $U$, whereas that of the surface method is only controlled by $U$.

Regarding complexity, computing $\bbf$ with \Cref{alg:point_based} costs $O(NQ)$ operations, independently of $U$: increasing $U$ improves the temporal accuracy without increasing this cost.
The cost of \Cref{alg:surface_based} is $O(N \overline{B} \, C_{\mathcal{M}})$ and thus grows linearly with $U$. This can be mitigated using the strategy developed in \Cref{rmk:coarser_area_grid}.

Regarding memory usage, the two implementations are identical.
In addition to $\bp_0$, which requires $O(N)$ values, the memory footprint is dominated by the $K$ vectors $\bbf_k$, their Fourier transforms and the workspaces of the forward and inverse batched FFT plans, that is about $4KUL$ floating point values.

Regarding the smoothness assumptions on $h$, the $Q^{-1}$ rate in \Cref{prop:approximation_points} requires a stronger smoothness condition than that in \Cref{prop:approximation_surface}. However, when both results are considered under the same smoothness assumption, the convergence rate of the point method becomes significantly slower, dropping to $Q^{-1/2}$.

\section{Implementations for usual transducer geometries}
\label{sec:usual_surfaces}

As shown in \Cref{subsec:surface_method}, computing the coefficients $\bbf$ reduces to evaluating the areas $\left| \mathcal{M}_i^l \right|$ defined in \eqref{eq:def:bffi}.
This section derives closed-form formulas for these areas for rectangular plane (\Cref{subsec:plane_surface}) and cylindrical (\Cref{subsec:cylinder_method}) transducers.

Since the grid point $\bx_i$ only appears through its relative position with respect to $\Sigma$, we absorb its coordinates into the parametrization: given a parametrization $\sigma_\Sigma$ of $\Sigma$ as in \Cref{ass:smoothness_transducer}, we set $\sigma = \sigma_\Sigma - \bx_i$, so that $\| \sigma(u,v) \|$ is the distance from $\bx_i$ to the corresponding point of $\Sigma$. The index $i$ is omitted in the following.
With this parametrization, computing the area $\left| \mathcal{M}_i^l \right|$ requires the determination of the subset
\begin{equation*}
    D^l = \left\{ (u,v) \in D \, | \, \| \sigma(u,v) \| \in [c\tau_{l-1}, c\tau_l) \right\}.
\end{equation*}

This will be carried out by solving the equation $\| \sigma(u,v) \|^2 = r^2$ with respect to $v$, thus obtaining $v$ as a function of $(u,r)$, that is, a description of the level sets of $\| \sigma(u,v) \|$. The boundaries of $D^l$ being the level sets associated with $r = c\tau_{l-1}$ and $r = c\tau_l$, together with the edges of $D$ that clip them, this yields closed-form expressions for its area.

\paragraph{Separable form}

Both geometries admit the separable decomposition:
\begin{equation} \label{eq:separable}
    \left\{
    \begin{aligned}
         \sigma &: D \to \R^3 \\
         D &= (u_1, u_2) \times (v_1,v_2) \\
         \|\sigma(u,v) \|^2 &= \psi(u) + v^2, \quad \forall(u,v) \in D \\
         \dint \sigma (u,v) &= J_0 \dint u \dint v, \quad \forall(u,v) \in D \\
    \end{aligned}
    \right.
\end{equation}
for some non-negative even function $\psi$ that is strictly increasing on $[0, \max(|u_1|, |u_2|)]$ (thus invertible) and a constant $J_0 > 0$.

Since $\sigma$ is injective and $\dint \sigma = J_0 \dint u \dint v$, the change of variables gives $\mathcal{M}^l = \sigma(D^l)$ and $\left| \mathcal{M}^l \right| = J_0 \left| D^l \right|$, so that computing the area reduces to a Lebesgue measure in the parameter plane.

\paragraph{Four-quadrant decomposition}

The function $\psi$ being even, it is strictly increasing only on $u \geq 0$, so the equation $\| \sigma(u,v)\|^2 = r^2$ cannot be inverted globally on $D$. We therefore split $D$ along the signs of $u$ and $v$ into the four subdomains
\begin{equation*}
    \begin{aligned}
        D_{\xi_1, \xi_2} & = \{ (u,v) \in D \, | \, \xi_1 u \geq 0, \xi_2 v \geq 0 \}, \qquad \xi_1, \xi_2 \in \{+,-\},
    \end{aligned}
\end{equation*}
as depicted in \Cref{fig:quadrant_decomposition}. Since $D$ is an open rectangle, each $D_{\xi_1,\xi_2}$ is again a rectangle. Accordingly,
\begin{equation*}
    \begin{aligned}
        D^l_{\xi_1, \xi_2} & = D^l \cap D_{\xi_1, \xi_2}, \\
        D^l &= \bigcup_{\xi_1, \xi_2 \in \{+,-\}} D^l_{\xi_1, \xi_2}. \\
    \end{aligned}
\end{equation*}
In other words, $D^l$ is decomposed into four components corresponding to the possible sign combinations of $(u,v)$. Note that depending on $D$ and $l$, some $D^l_{\xi_1, \xi_2}$ could be empty.
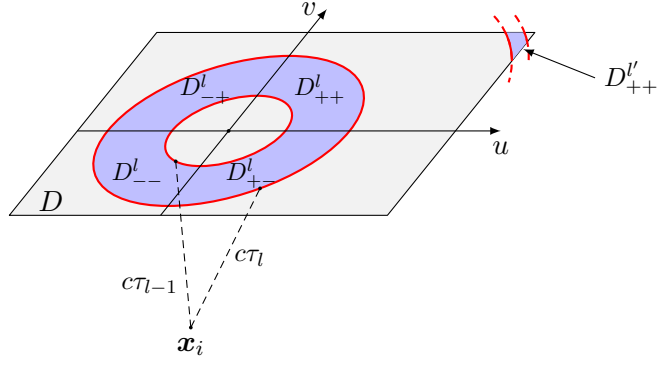
\begin{figure}[!htb] \centering
       \begin{tikzpicture}[>=latex, line join=round, scale=0.5]
  \def\sx{0.50}\def\sy{0.62}
  \def\umin{-4.0}\def\umax{6.0}\def\vmin{-3.6}\def\vmax{4.2}
  \def\ra{1.5}\def\rb{3.2}
  \def\rc{6.7}\def\rd{7.1}

  \begin{scope}[cm={1,0,\sx,\sy,(0,0)}]
    \fill[black!5, draw=black, thin] (\umin,\vmin) rectangle (\umax,\vmax);

    \fill[blue!25,even odd rule] (0,0) circle (\ra) (0,0) circle (\rb);
    \draw[red,thick] (0,0) circle (\ra);
    \draw[red,thick] (0,0) circle (\rb);

    \fill[blue!25] (38.83:\rc) arc (38.83:26.42:\rc) -- (32.32:\rd)
                   arc (32.32:36.28:\rd) -- cycle;
    \draw[red,thick,densely dashed] (47:\rc) arc (47:18:\rc);
    \draw[red,thick,densely dashed] (44:\rd) arc (44:24:\rd);
    \draw[red,thick] (38.83:\rc) arc (38.83:26.42:\rc);
    \draw[red,thick] (36.28:\rd) arc (36.28:32.32:\rd);

    \coordinate (O)  at (0,0);
    \coordinate (uA) at (\umin,0);   \coordinate (uB) at (\umax+1.2,0);
    \coordinate (vA) at (0,\vmin);   \coordinate (vB) at (0,\vmax+1.0);
    \coordinate (Dl) at (\umin+0.8,\vmin+0.6);
    \coordinate (Cpp) at (50:2.35);  \coordinate (Cmp) at (130:2.35);
    \coordinate (Cmm) at (230:2.35); \coordinate (Cpm) at (310:2.35);
    \coordinate (Ra) at (240:\ra);   \coordinate (Rb) at (310:\rb);
    \coordinate (Tc) at (31:6.9);
  \end{scope}

  \draw[->] (uA) -- (uB) node[below] {$u$};
  \draw[->] (vA) -- (vB) node[left] {$v$};
  \node at (Dl) {$D$};

  \coordinate (Xi) at (-1.0,-5.2);
  \draw[densely dashed] (Xi) -- (Ra) node[pos=0.25,left,scale=0.9] {$c\tau_{l-1}$};
  \draw[densely dashed] (Xi) -- (Rb) node[pos=0.50,right,scale=0.9] {$c\tau_{l}$};
  \fill (O) circle (1.4pt);  \fill (Ra) circle (1.4pt);  \fill (Rb) circle (1.4pt);
  \fill (Xi) circle (1.4pt) node[below] {$\bx_i$};

  \node[scale=0.85] at (Cpp) {$D^{l}_{++}$};
  \node[scale=0.85] at (Cmp) {$D^{l}_{-+}$};
  \node[scale=0.85] at (Cmm) {$D^{l}_{--}$};
  \node[scale=0.85] at (Cpm) {$D^{l}_{+-}$};
  \draw[<-, shorten <=1.5pt] (Tc) -- ++(2.0,-0.8) node[right,scale=0.9] {$D^{l'}_{++}$};
\end{tikzpicture}
       \caption{Decomposition of $D^l$ into the subdomains $D^l_{\xi_1,\xi_2}$, drawn in the $(u,v)$ plane for a planar transducer. The voxel $\bx_i$ lies below the surface and projects onto the origin, at distances $c\tau_{l-1}$ and $c\tau_l$ from the two curves delimiting $D^l$. For the time step $l$ the four components are nonempty, whereas for $l'$ only $D^{l'}_{++}$ is nonempty and clipped by the boundary of $D$.}
       \label{fig:quadrant_decomposition}
\end{figure}

The area $|\mathcal{M}^l|$ can then be obtained via
\begin{equation} \label{eq:swept_u}
       \left|\mathcal{M}^l\right| = J_0 \sum_{\xi_1, \xi_2 \in \{+,-\}} \left| D^l_{\xi_1, \xi_2}\right|.
\end{equation}

\paragraph{Reduction to the positive quadrant}

Again, since both $\psi$ and the square function $v \mapsto v^2$ are even, the reflection $(u,v) \mapsto (\xi_1 u, \xi_2 v)$ maps $D^l_{\xi_1,\xi_2}$ onto a set of the same form contained in the quadrant $u, v \geq 0$, without changing its area.
We therefore assume $\xi_1 = \xi_2 = +$ and adapt the domain limits accordingly.  Specifically, we define
\begin{equation} \label{eq:uvminmax}
    \begin{aligned}
        u_{\min}^+ &= \max(u_1,0)   \quad &  u_{\max}^+ &= \max(u_2,0), \\
        u_{\min}^- &= \max(-u_2,0)   \quad &  u_{\max}^- &= \max(-u_1,0),
    \end{aligned}
\end{equation}
and define $v_{\min}^{\xi_2}$ and $v_{\max}^{\xi_2}$ analogously for $\xi_2 \in \{ +, - \}$, the maxima producing a degenerate interval, hence a zero area, when $D$ does not meet the corresponding quadrant.

Abusing notation, we identify from now on $D_{\xi_1, \xi_2}$ and $D^l_{\xi_1,\xi_2}$ with their images under the reflection, that is
\begin{equation*}
    \begin{aligned}
        D_{\xi_1, \xi_2} &= \left(u_{\min}^{\xi_1}, u_{\max}^{\xi_1} \right) \times \left(v_{\min}^{\xi_2}, v_{\max}^{\xi_2} \right), \\
        D^l_{\xi_1, \xi_2} &= \left\{ (u,v) \in D_{\xi_1, \xi_2} \, | \, \| \sigma(u,v) \| \in [c\tau_{l-1}, c\tau_l) \right\}.
    \end{aligned}
\end{equation*}
All four quadrants are now contained in $\{ u \geq 0, \, v \geq 0 \}$, where $\psi$ is strictly increasing and therefore invertible by \eqref{eq:separable}.

\paragraph{Characterization of $D^l_{\xi_1, \xi_2}$}

First, the separable form allows us to derive closed-form expressions for the minimum and maximum radii where $S(\bx_i,r)$ and $\sigma(D_{\xi_1, \xi_2})$ intersect.
More specifically in each quadrant $D_{\xi_1, \xi_2}$, we seek the infimum and supremum of the values of $r$ for which the equation $\psi(u) + v^2 = r^2$ admits a solution for $(u,v) \in  D_{\xi_1, \xi_2}$. Since $\psi$ and $v \mapsto v^2$ are both strictly increasing on the quadrant, these extrema are given by the corners of its closure:
\begin{equation} \label{eq:closedform_rminrmax}
    \begin{aligned}
        r_{\min}^{\xi_1, \xi_2} &= \sqrt{ \psi\left(u_{\min}^{\xi_1}\right) + \left(v_{\min}^{\xi_2}\right)^2 },\\
        r_{\max}^{\xi_1, \xi_2} &= \sqrt{ \psi\left(u_{\max}^{\xi_1}\right) + \left(v_{\max}^{\xi_2}\right)^2 }.
    \end{aligned}
\end{equation}
As a by-product, taking the minimum and the maximum over the four quadrants gives the closed-form expressions of $r_{\min}(i)$ and $r_{\max}(i)$ defined in \eqref{eq:rmin_rmax}.

Since the treatment of each $D^l_{\xi_1, \xi_2}$ will be similar, the dependence on $\xi_1,\xi_2$ will be dropped in the remainder of the derivation.
The separable form \eqref{eq:separable} allows the equation $ \| \sigma(u,v) \|^2 = r^2$ to be solved explicitly as
\begin{equation*}
    \begin{aligned}
        v &= \Psi(u,r) = \sqrt{r^2 - \psi(u)}, \quad \text{decreasing in $u$ and increasing in $r$}, \\
        u &= \Phi(v,r) = \psi^{-1}(r^2 - v^2), \quad \text{decreasing in $v$ and increasing in $r$}.
    \end{aligned}
\end{equation*}
To properly handle their domain of definition when the level sets exit $D$, they will be considered as
\begin{equation*}
    \widehat{\Psi} : [u_{\min}, u_{\max}] \times \R_+ \to [v_{\min}, v_{\max}], \qquad
    \widehat{\Phi} : [v_{\min}, v_{\max}] \times \R_+ \to [u_{\min}, u_{\max}],
\end{equation*}
defined by
\begin{equation*}
    \widehat{\Psi}(u,r) = \begin{cases}
        v_{\min} & \text{if } r^2 - \psi(u) \leq v_{\min}^2, \\
        \Psi(u,r) & \text{if } r^2-\psi(u) \in [v_{\min}^2, v_{\max}^2], \\
        v_{\max} & \text{if } r^2 - \psi(u) \geq  v_{\max}^2,
    \end{cases}
\end{equation*}
and
\begin{equation*}
    \widehat{\Phi}(v,r) = \begin{cases}
        u_{\min} & \text{if } r^2-v^2 \leq \psi(u_{\min}), \\
        \Phi(v,r) & \text{if } r^2-v^2 \in [\psi(u_{\min}),\psi(u_{\max})], \\
        u_{\max} & \text{if } r^2-v^2 \geq \psi(u_{\max}).
    \end{cases}
\end{equation*}

We define
\begin{equation} \label{eq:alphabeta}
    \alpha_l = \widehat{\Phi}(v_{\max}, c\tau_l) \quad \text{and} \quad \beta_l = \widehat{\Phi}(v_{\min}, c\tau_l).
\end{equation}
Geometrically, $\alpha_l$ and $\beta_l$ are the abscissas at which the iso-distance curve $\| \sigma \| = c\tau_l$ meets the edges $v = v_{\max}$ and $v = v_{\min}$ of the quadrant, clipped to $[u_{\min}, u_{\max}]$.
Since $\widehat{\Phi}$ is decreasing in $v$, $\alpha_l \leq \beta_l$; since it is increasing in $r$, $\alpha_{l-1} \leq \alpha_l$ and $\beta_{l-1} \leq \beta_l$.
As $\widehat{\Psi}$ is increasing in $r$ as well, the interval $[\alpha_{l-1}, \beta_l]$ is exactly the set of abscissas at which $D^l_{\xi_1,\xi_2}$ has a non-empty vertical slice, and this slice is delimited by the two iso-distance curves. It follows that, up to a set of zero measure that does not affect the areas computed below,
\begin{equation} \label{eq:D_l_xi}
    D^l_{\xi_1,\xi_2} = \left\{ (u,v) \, | \, u \in [\alpha_{l-1}, \beta_l],\; v \in \left[\widehat{\Psi}(u, c\tau_{l-1}), \widehat{\Psi}(u, c\tau_l) \right] \right\}.
\end{equation}
Note that when $r \leq r_{\min}$, both $\widehat{\Phi}(v_{\max}, r)$ and $\widehat{\Phi}(v_{\min}, r)$ are equal to $u_{\min}$, whereas for $r \geq r_{\max}$ they are both equal to $u_{\max}$.
Consequently, $D^l_{\xi_1,\xi_2}$ has zero area whenever $c\tau_l \leq r_{\min}$ or $c\tau_{l-1} \geq r_{\max}$, which restricts the time steps to be visited for that quadrant to those with $c\tau_l \in (r_{\min}, r_{\max}]$. These quantities are illustrated in \Cref{fig:swept_area}.

\begin{figure}[!htb]
       \centering
       \begin{tikzpicture}[scale=1.6,>=latex]
  \def\ri{2.343}
  \def\ro{3.413}
  \def\umn{1.0}\def\umx{4.2}
  \def\vmn{0.6}\def\vmx{1.8}

  \coordinate (O)  at (0,0);
  \coordinate (Tm) at (1.500,\vmx);
  \coordinate (Tp) at (2.900,\vmx);
  \coordinate (Bm) at (2.265,\vmn);
  \coordinate (Bp) at (3.360,\vmn);

  \fill[blue!18]
      (Tm) -- (Tp) arc (31.83:10.12:\ro) -- (Bm) arc (14.84:50.19:\ri) -- cycle;

  \draw[red,thick] (10:\ri) arc (10:62:\ri);
  \draw[red,thick] (7:\ro)  arc (7:42:\ro);
  \node[red,above]       at (62:\ri) {$\Psi(u,c\tau_{l-1})$};
  \node[red,above right] at (42:\ro) {$\Psi(u,c\tau_{l})$};

  \draw[densely dashed] (\umn,\vmn) rectangle (\umx,\vmx);
  \foreach \y/\lab in {\vmn/{v_{\min}}, \vmx/{v_{\max}}}{
     \draw[densely dotted] (\umn,\y) -- (0,\y);
     \node[left] at (0,\y) {$\lab$};
  }

  \draw[->] (-0.15,0) -- (4.75,0) node[below] {$u$};
  \draw[->] (0,-0.15) -- (0,2.7)  node[left]  {$v$};
  \fill (O) circle (0.7pt) node[below left] {$\bx_i$};

  \draw[densely dotted] (O) -- (Tm);
  \draw[densely dotted] (O) -- (Tp);
  \node[rotate=50, fill=white, inner sep=1pt] at (0.675,0.85) {$c\tau_{l-1}$};
  \node[rotate=32, fill=white, inner sep=1pt] at (1.305,0.85) {$c\tau_{l}$};

  \foreach \x/\lab in {\umn/{u_{\min}}, \umx/{u_{\max}}}{
     \draw[densely dotted] (\x,\vmn) -- (\x,0);
     \node[below] at (\x,0) {$\lab$};
  }

  \foreach \p/\lab in {Tm/{\alpha_{l-1}}, Bm/{\beta_{l-1}}, Tp/{\alpha_{l}}, Bp/{\beta_{l}}}{
     \fill (\p) circle (0.6pt);
     \draw[densely dotted] (\p) -- (\p|-O);
     \node[below] at (\p|-O) {$\lab$};
  }

  \node[fill=blue!18, inner sep=1pt] at (2.52,1.15) {$D_{\xi_1,\xi_2}^l$};
\end{tikzpicture}
       \caption{The subdomain $D^l_{\xi_1,\xi_2}$ (in blue) is delimited by the two iso-distance curves $v = \Psi(u,c\tau_{l-1})$ and $v = \Psi(u,c\tau_l)$ (in red) and by the boundary of $D_{\xi_1,\xi_2}$.}
       \label{fig:swept_area}
\end{figure}
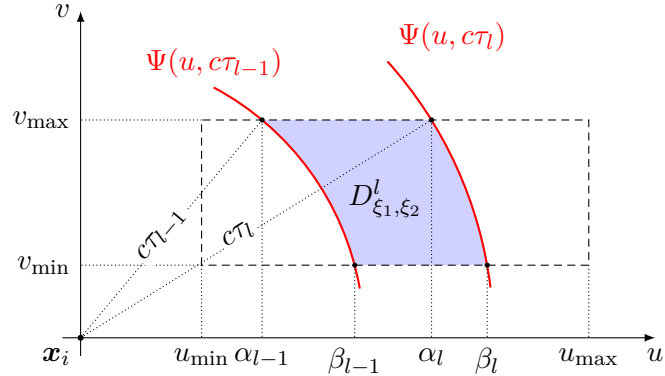

Finally, the area of $D^l_{\xi_1,\xi_2}$ is obtained by introducing the function $\mathcal{F}(r)$
\begin{equation*}
    \mathcal{F}(r) = \int_{u_{\min}}^{u_{\max}} \widehat{\Psi}(u,r) \dint u = \left| \left\{ (u,v) \in D_{\xi_1,\xi_2} \, \middle| \, \|\sigma(u,v)\| \leq r \right\} \right| + v_{\min} \left( u_{\max} - u_{\min} \right),
\end{equation*}
so that \eqref{eq:D_l_xi} reads $\left| D^l_{\xi_1, \xi_2} \right| = \mathcal{F}(c\tau_l) - \mathcal{F}(c\tau_{l-1})$.
By definition of $\alpha_l$ and $\beta_l$, the integrand is equal to $v_{\max}$ on $[u_{\min}, \alpha_l]$, to $\Psi(u, c\tau_l)$ on $[\alpha_l, \beta_l]$ and to $v_{\min}$ on $[\beta_l, u_{\max}]$, so that
\begin{equation*}
    \mathcal{F}(c\tau_l) = (\alpha_l - u_{\min}) v_{\max} + I_l + (u_{\max} - \beta_l) v_{\min}, \qquad \text{with } I_l = \int_{\alpha_l}^{\beta_l} \Psi(u, c\tau_l) \dint u.
\end{equation*}
Writing the same identity at $c\tau_{l-1}$ and subtracting, the terms in $u_{\min}$ and $u_{\max}$ cancel and
\begin{equation} \label{eq:area_D_l_xi}
    \left| D^l_{\xi_1, \xi_2} \right| = (\alpha_l - \alpha_{l-1}) v_{\max} + I_l - I_{l-1} - (\beta_l - \beta_{l-1}) v_{\min}.
\end{equation}

These derivations are summarized in \Cref{alg:separable_area}, which evaluates $\left| \mathcal{M}_i^l \right|$ for any surface admitting the separable form \eqref{eq:separable}: the only geometry-dependent inputs are the domain $D$, the function $\psi$, its inverse through $\widehat{\Phi}$, the closed form of $I_l$ and the constant $J_0$, derived for the plane in \Cref{subsec:plane_surface} and for the cylinder in \Cref{subsec:cylinder_method}.

\begin{algorithm}[ht]
\caption{Area $\big| \mathcal{M}_i^l \big|$ for a surface in separable form \eqref{eq:separable}.}
\label{alg:separable_area}
\begin{algorithmic}[1]
       \Require Voxel $\bx_i$ in the transducer frame and the transducer parameters; time step edge $c\tau_l$; triples $\left(\alpha_{l-1}, \beta_{l-1}, I_{l-1}\right)$ stored at the previous edge $c\tau_{l-1}$, initialized to $\left(u_{\min}^{\xi_1}, u_{\min}^{\xi_1}, 0\right)$ at the first edge.
       \Ensure Area $\big| \mathcal{M}_i^l \big|$; store updated to the edge $c\tau_l$.
       \State Get $D$, $\psi$, $\widehat{\Phi}$, $J_0$ and the closed form of $I_l$ \Comment{Geometry-specific, \Cref{subsec:plane_surface,subsec:cylinder_method}}
       \State $\big| \mathcal{M} \big| \gets 0$
       \For{$(\xi_1, \xi_2) \in \{+,-\}^2$} \Comment{Loop over the four quadrants of $D$}
              \State Quadrant limits $\left(u_{\min}^{\xi_1}, u_{\max}^{\xi_1}\right) \times \left(v_{\min}^{\xi_2}, v_{\max}^{\xi_2}\right)$
              \State Bounds $(\alpha_l, \beta_l)$ from \eqref{eq:alphabeta}
              \State $I_l \gets \int_{\alpha_l}^{\beta_l} \Psi(u, c\tau_l) \dint u$ \Comment{Geometry-specific closed-form \eqref{eq:plane_primitive},\eqref{eq:cyl_Il}}
              \State $\big| \mathcal{M} \big| \gets \big| \mathcal{M} \big| + J_0 \bigl[ (\alpha_l - \alpha_{l-1}) v_{\max} + I_l - I_{l-1} - (\beta_l - \beta_{l-1}) v_{\min} \bigr]$ \Comment{\eqref{eq:swept_u},\eqref{eq:area_D_l_xi}}
              \State Store $(\alpha_l, \beta_l, I_l)$ for time step $l+1$
       \EndFor
       \State \Return $\big| \mathcal{M} \big|$
\end{algorithmic}
\end{algorithm}

The following diagram shows the organization of the end of the section:
\begin{center}
       \begin{tikzpicture}[node distance=0.6cm and 0.6cm, block/.append style={text width=2.5cm}]
  \node (root)  [block] {Area $\big| \mathcal{M}_i^{l} \big|$ for surface transducer\\(\Cref{subsec:surface_method})};
  \node (junc1) [junction, right=0.9cm of root] {};

  \node (plane) [block, above right=0.6cm and 0.6cm of junc1] {Plane surface\\(\Cref{subsec:plane_surface})};
  \node (cyl)   [block, below right=0.6cm and 0.6cm of junc1] {Cylinder surface\\(\Cref{subsec:cylinder_method})};

  \node (junc2) [junction, right=0.9cm of cyl] {};
  \node (pw)    [block, above right=0.6cm and 0.6cm of junc2] {Piecewise-Planes\\(\Cref{subsec:plane_method})};
  \node (ell)   [block, below right=0.6cm and 0.6cm of junc2] {Elliptic integral\\(\Cref{subsec:elliptic_method})};

  \node (junc3) [junction, right=0.9cm of ell] {};
  \node (lut)   [block, text width=1.7cm, right=0.4cm of junc3]  {Lookup table};
  \node (exact) [block, text width=1.7cm, above=0.4cm of lut]  {Exact};
  \node (ff)    [block, text width=1.7cm, below=0.4cm of lut]  {Trapezoidal};

  \draw[arr] (root)  -- (junc1);
  \draw[arr] (junc1) |- (plane);
  \draw[arr] (junc1) |- (cyl);

  \draw[arr] (cyl)   -- (junc2);
  \draw[arr] (junc2) |- (ell);
  \draw[arr] (junc2) |- (pw);

  \draw[arr] (ell)   -- (junc3);
  \draw[arr] (junc3) |- (exact);
  \draw[arr] (junc3) -- (lut);
  \draw[arr] (junc3) |- (ff);
  \draw[arr, dashed] (plane) -| node[pos=0.25, font=\footnotesize, above] {reuse} (pw);
\end{tikzpicture}
\end{center}
The framework developed above applies to both planar and cylindrical transducers through an appropriate choice of the surface-dependent function $\psi$.
We first consider the planar transducer, for which a closed-form solution is readily obtained.
We then turn to the cylindrical transducer and derive an exact closed-form expression involving incomplete elliptic integrals. To reduce the computational cost, we introduce two approximations: a lookup table implementation and a trapezoidal approximation of the arc integral. Finally, we exploit the planar closed-form solution to derive a piecewise planar approximation of the cylindrical transducer.

\subsection{Plane surface transducer}
\label{subsec:plane_surface}
A planar transducer, depicted in \Cref{fig:planar}, is defined by a center $\bc_0 \in \R^3$, two orthonormal vectors $\be_1, \be_2$, a half-width $L_1 > 0$ and a half-height $L_2 > 0$. In the frame $(\be_1, \be_2, \be_1 \times \be_2)$ centered at $\bc_0$, the surface is the rectangle
\begin{equation} \label{eq:plane_param}
       \Sigma = \left\{ \bx = (x_1, x_2, 0) : x_1 \in [-L_1, L_1],\, x_2 \in [-L_2, L_2] \right\}, \quad \dint\sigma = \dint x_1\,\dint x_2,
\end{equation}
and a voxel is mapped from the global frame by $\bx_i \mapsto \bR_p(\bx_i - \bc_0)$ with $\bR_p^\top = (\be_1, \be_2, \be_1 \times \be_2)$.
To avoid overloading notation, in the remainder of the section, $\bx_i$ will refer to the voxel in the transducer frame.

\begin{figure}[!htb] \centering
       \begin{tikzpicture}[scale=1.5]

  \def\a{1.4}
  \def\b{0.9}

  \draw[->, thick] (0,0,0) -- (1.9,0,0)  node[anchor=north]      {$\be_1$};
  \draw[->, thick] (0,0,0) -- (0,0,1.7)  node[anchor=north west] {$\be_2$};
  \draw[->, thick] (0,0,0) -- (0,1.4,0)  node[anchor=south]      {$\be_1 \times \be_2$};
  \fill (0,0,0) circle (0.6pt) node[anchor=east] {$\bc_0$};

  \draw[fill=blue!20!white, opacity=0.4, draw=blue!60, thick]
    (-\a,0,-\b) -- (\a,0,-\b) -- (\a,0,\b) -- (-\a,0,\b) -- cycle;

  \draw[<->] (-\a,0,-\b-0.4) -- (\a,0,-\b-0.4);
  \node at (0,0,-\b-0.8) {$2L_1$};

  \draw[<->] (-\a-0.25,0,-\b) -- (-\a-0.25,0,\b);
  \node at (-\a-0.5,0,0) {$2L_2$};

\end{tikzpicture}
       \caption{Illustration of the plane on the coordinate system $(\be_1, \be_2, \be_1 \times \be_2)$.} \label{fig:planar}
\end{figure}
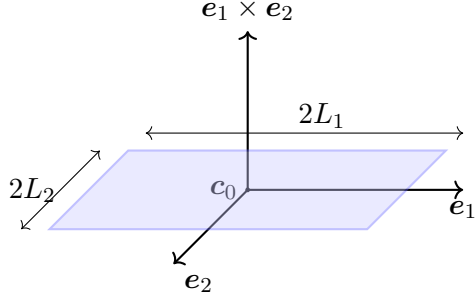

For a given $\bx_i$ the separable form \eqref{eq:separable} associated to $\Sigma$ is obtained by setting $u = x_1 - \bx_i[1]$, $v = x_2 - \bx_i[2]$ and reads
\begin{equation*}
    \left\{
    \begin{aligned}
        \sigma(u,v) &= (u, v, -\bx_i[3])\\
        D &= \left(-L_1 - \bx_i[1], L_1 - \bx_i[1]\right) \times \left(-L_2 - \bx_i[2], L_2 - \bx_i[2]\right) \\
        \psi(u) &= u^2 + \bx_i[3]^2 \\
        J_0 &= 1
    \end{aligned}
    \right.
\end{equation*}
and $\psi$ is non-negative even and strictly increasing on $u \geq 0$.

A few computations give furthermore the identities
\begin{equation*}
       \Psi(u, r) = \sqrt{\rho(r)^2 - u^2}, \qquad \Phi(v, r) = \sqrt{\rho(r)^2 - v^2}, \qquad \text{with } \rho(r) = \sqrt{r^2 - \bx_i[3]^2}.
\end{equation*}
Note that for $r \in [r_{\min}, r_{\max}]$ defined in \eqref{eq:closedform_rminrmax}, $r^2 - \bx_i[3]^2 \geq u_{\min}^2 + v_{\min}^2 \geq 0$ so that $\rho$ is properly defined.
The integral $I_l$ can therefore be obtained via the elementary primitive of $\sqrt{\rho^2 - u^2}$:
\begin{equation}  \label{eq:plane_primitive}
    I_l = \frac{1}{2}\left[ u \sqrt{\rho(c\tau_l)^2 - u^2} + \rho(c\tau_l)^2\,\arcsin \left( \frac{u}{\rho(c\tau_l)}\right) \right]_{\alpha_l}^{\beta_l}.
\end{equation}
With these computations at hand, the area of $D^l$ follows by summing the four $\left| D^l_{\xi_1, \xi_2} \right|$ given by \eqref{eq:area_D_l_xi}. \Cref{alg:separable_area} can therefore be instantiated with the planar quantities derived above.

\subsection{Cylindrical surface transducer}
\label{subsec:cylinder_method}
A cylindrical transducer, depicted in \Cref{fig:cylinder}, is defined by a center $\bc_0 \in \R^3$, two orthonormal vectors $\be_1, \be_2$, a radius $R > 0$, an angular half-aperture $\theta_{\max} \in (0, \pi/2)$ and an axial half-height $L_z > 0$. In the frame $(\be_2, \be_1 \times \be_2, \be_1)$ centered at $\bc_0$, the surface is
\begin{equation} \label{eq:cylinder_param}
       \Sigma = \left\{ \bx(\theta, z) = (R\cos\theta,\, R\sin\theta,\, z) : \theta \in [-\theta_{\max}, \theta_{\max}],\, z \in [-L_z, L_z] \right\}, \quad \dint\sigma = R\,\dint\theta\,\dint z,
\end{equation}
and a voxel is mapped from the global frame by $\bx_i \mapsto \bR_c(\bx_i - \bc_0)$, $\bR_c^\top = (\be_2,\, \be_1 \times \be_2,\, \be_1)$.
To avoid overloading notation, in the remainder of the section, $\bx_i$ will refer to the voxel in the transducer frame.

\begin{figure}[!htb] \centering
       \begin{tikzpicture}[scale=1.5]

  \def\R{2.5}
  \def\h{0.25}
  \def\thmax{45}
  \def\steps{128}

  \draw[->, thick] (0,0,0) -- (2,0,0)   node[anchor=south west]        {$\be_2$};
  \draw[->, thick] (0,0,0) -- (0,0,-3)   node[anchor=south]       {$\be_1\times \be_2$};
  \draw[->, thick] (0,0,0) -- (0,2,0)   node[anchor=south] {$\be_1$};
  \fill (0,0,0) circle (0.6pt) node[left] {$\bc_0$};

  \foreach \i in {0,...,\steps} {
    \pgfmathsetmacro\theta{-\thmax + 2*\thmax/\steps * \i}
    \pgfmathsetmacro\xA{\R*cos(\theta)}
    \pgfmathsetmacro\zA{\R*sin(\theta)}
    \pgfmathsetmacro\nextt{\theta + 2*\thmax/\steps}
    \pgfmathsetmacro\xB{\R*cos(\nextt)}
    \pgfmathsetmacro\zB{\R*sin(\nextt)}
    \pgfmathsetmacro\col{100 - (\i/\steps)*80}

    \draw[fill=blue!\col!white, opacity=0.5, draw=none]
      (\xA,-\h,\zA) -- (\xA,\h,\zA) --
      (\xB,\h,\zB) -- (\xB,-\h,\zB) -- cycle;
  }

  \draw[dashed] (0,0,0) -- ({\R*cos(\thmax)},0,{\R*sin(\thmax)});
  \draw[dashed] (0,0,0) -- ({\R*cos(-\thmax)},0,{\R*sin(-\thmax)});

  \draw[->] plot[domain=0:\thmax,variable=\t,samples=20]
    ({0.6*cos(\t)},0,{0.6*sin(\t)});
  \node at ({0.8*cos(22.5)},0,{sin(22.5)}) [right] {$\theta_{\max}$};

  \node at ({1.2*cos(\thmax)},0,{-1.2*sin(\thmax)}) [above] {$R$};

  \draw[<->] (1.5,-\h,2) -- (1.5,\h,2);
  \node at (1.3,0,2) {$2L_z$};

\end{tikzpicture}
       \caption{Illustration of the cylinder on the coordinate system $(\be_2, \be_1 \times \be_2, \be_1)$.} \label{fig:cylinder}
\end{figure}
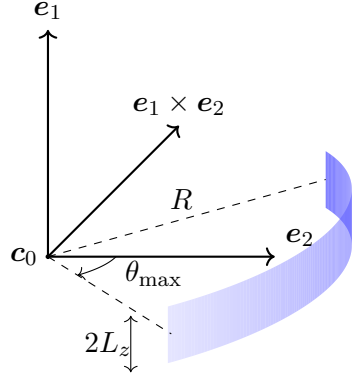

\subsubsection{Elliptic integral}
\label{subsec:elliptic_method}

For a given $\bx_i$, let
\begin{equation*}
    \begin{aligned}
        \rho &= \sqrt{\bx_i[1]^2 + \bx_i[2]^2} \\
        \theta_i &= \arctantwo(\bx_i[2], \bx_i[1]).
    \end{aligned}
\end{equation*}
and assume $\rho > 0$.
The separable form \eqref{eq:separable} associated to $\Sigma$ is obtained by setting $u = \theta - \theta_i$ and $v = z - \bx_i[3]$:
\begin{equation*}
    \left\{
    \begin{aligned}
        \sigma(u,v) &= (R\cos(u) - \rho, R\sin(u), v)\\
        D &= \left(-\theta_{\max} - \theta_i, \theta_{\max} - \theta_i\right) \times \left(-L_z - \bx_i[3], L_z - \bx_i[3]\right) \\
        \psi(u) &= \rho^2 + R^2 - 2R\rho\cos(u) \\
        J_0 &= R
    \end{aligned}
    \right.
\end{equation*}
and $\psi$ is non-negative even, $2\pi$-periodic and strictly increasing on $[0,\pi]$.

Since $\psi$ is even and $2\pi$-periodic, $u$ is regarded as an angle modulo $2\pi$, and the endpoints $-\theta_{\max}-\theta_i$ and $\theta_{\max}-\theta_i$ are wrapped into $(-\pi,\pi]$, still denoted $u_1,u_2$.
As $2\theta_{\max}<\pi$, the resulting arc contains at most one of $0$ and $\pm\pi$.
The interval is split at one of these points, that is $0$ when $u_1 < 0 < u_2$ and $\pi$ when $u_2 < u_1$. Each split interval is then folded onto $[0,\pi]$ by $u\mapsto|u|$.
This yields at most two intervals
\begin{equation}\label{eq:folded_intervals}
\begin{cases}
[0,-u_1],\ [0,u_2] & \text{if } u_1 < 0 < u_2,\\
[u_1,\pi],\ [-u_2,\pi] & \text{if } u_2 < u_1,\\
[u_1,u_2],\ \emptyset & \text{if } 0 \le u_1 < u_2, \\
[-u_2,-u_1],\ \emptyset & \text{if } u_1 < u_2 \le 0,
\end{cases}
\end{equation}
defining the $u_{\min}$, $u_{\max}$ bounds of \eqref{eq:uvminmax}.

A few computations give furthermore the identities
\begin{equation*}
       \Psi(u, r) = \sqrt{2R \rho \cos u + r^2 - \rho^2 - R^2}, \qquad \Phi(v,r) = \arccos\left( \frac{\rho^2 + R^2 - r^2 + v^2}{2R\rho} \right).
\end{equation*}
The integral $I_l$ can then be expressed, following \Cref{lem:exact_integral} in \Cref{app:elliptic_integrals}, in terms of the incomplete elliptic integral of the second kind $E$:
\begin{equation} \label{eq:cyl_Il}
    \begin{aligned}
        I_l = \int_{\alpha_l}^{\beta_l} \Psi(u,c \tau_l) \dint u & = 2 \sqrt{2R\rho + b(c \tau_l)} \left[ E\left(\frac{\beta_l}{2}, \nu(c \tau_l)\right) - E\left(\frac{\alpha_l}{2}, \nu(c \tau_l)\right) \right] \\
        \text{where, } \qquad b(r) & = r^2 - \rho^2 - R^2 \\
                             \nu(r) & = \frac{4R\rho}{2R\rho + b(r)} \\
                             E(\vartheta, \nu) & = \int_{0}^{\vartheta} \sqrt{1 - \nu \sin^2(t)} \dint t.
    \end{aligned}
\end{equation}
The elliptic integral is well-defined as long as $\nu \sin^2(t) \leq 1$ for all $t \in [0,\vartheta]$, which \Cref{lem:exact_integral} guarantees for $\vartheta = \alpha_l/2$ and $\vartheta = \beta_l/2$.

Note that the case $\rho = 0$ occurs when $\bx_i$ is on the focal line of the cylindrical transducer, that is the line passing through $\bc_0$ along $\be_1$ (see \Cref{fig:cylinder}).
The algorithm nevertheless remains valid. Since $\psi(u_{\min}) = \psi(u_{\max}) = R^2$, the clamped map $\widehat\Phi(v,r)$ takes only the values $u_{\min}$ if $r^2 - v^2 \leq R^2$ and $u_{\max}$ otherwise. This leads to $\alpha_l, \beta_l \in \{u_{\min}, u_{\max}\}$ so that the $\arccos$ is never evaluated.

With these computations at hand, the area of $D^l$ again follows from \eqref{eq:area_D_l_xi} and \Cref{alg:separable_area}.
Unlike the plane transducer, the integral $I_l$ is more complex to deal with numerically.
In the following paragraphs we investigate three ways to evaluate it: exactly, through a lookup table, and through a trapezoidal approximation.

\paragraph{Exact}

Standard routines evaluate $E$ using the arithmetic-geometric mean \cite[Section 19.22 (ii)]{NIST:DLMF} or the Carlson symmetric forms \cite{carlson1995numerical}.
These implementations generally require the parameter $\nu$ to be in the range $[0,1]$.
Although $\nu(r)$ is always positive, it exceeds $1$ whenever $r < R + \rho$ (\Cref{app:elliptic_integrals}), which occurs for a large part of the time steps and voxels.
The additional evaluation of $F$ in \eqref{eq:reciprocal_modulus} below is therefore required for a substantial part of the evaluations of $I_l$.
For $\nu(r) > 1$, and for $\vartheta \in [0,\pi/2]$ with $\nu \sin^2\vartheta \leq 1$, we leverage the classical reciprocal-modulus transformation \cite[Eqs.~17.4.15--16]{abramowitz1966handbook}
\begin{equation} \label{eq:reciprocal_modulus}
       \begin{aligned}
              E(\vartheta, \nu) & = \sqrt{\nu}\left[ E\!\left(\widetilde{\vartheta}, \tfrac{1}{\nu}\right) - \left(1 - \tfrac{1}{\nu}\right) F\!\left(\widetilde{\vartheta}, \tfrac{1}{\nu}\right) \right], \\
              \text{where, } \qquad \widetilde{\vartheta} & = \arcsin\!\left(\sqrt{\nu}\sin\vartheta\right), \\
                                   F(\vartheta, \nu) & = \int_{0}^{\vartheta} \left( 1 - \nu \sin^2 t \right)^{-\frac{1}{2}} \dint t.
       \end{aligned}
\end{equation}
This transformation maps the evaluation of $E$ to elliptic integrals with a parameter in the range supported by the numerical routines, at the additional cost of one incomplete elliptic integral of the first kind $F$.
We rely on the Carlson symmetric forms implementations of these functions provided by the \texttt{boost::math} library on GPUs.

Computing $I_l$ for a single time step requires 2 to 8 evaluations of $E$, each involving one to two calls to an elliptic integral function with a substantial amount of floating point operations.
For the system described in \Cref{subsec:lib_setup}, approximately $10^{13}$ evaluations of $I_l$ are needed, making this computation the primary bottleneck of the method.
Significant performance gains can therefore be obtained by reducing the cost of these evaluations. To this end, we consider two alternatives: a lookup table approach and a trapezoidal approximation.

\paragraph{Lookup table (LUT)}\label{subsec:LUT_method}

A standard strategy for repeatedly evaluating an expensive function is to precompute its values on a lookup table and to interpolate at query time.
When the parameter range is bounded and the function is smooth, this replaces each evaluation by a few memory accesses and a bilinear interpolation.

Following this strategy for $E$ requires two adaptations.
First, since $\nu(r)$ takes arbitrary values in $\R_+$, we use \eqref{eq:reciprocal_modulus} again to map $\nu$ to a bounded domain, at the price of an additional LUT for $F$.
Second, $F$ has a logarithmic singularity as $(\vartheta, \nu) \to \left(\frac{\pi}{2}, 1\right)$ \cite[Section 19.12]{NIST:DLMF}, so that a uniform sampling would require a prohibitively fine grid. We use a non-uniform sampling instead.

Finally, to save the $\arcsin$ computation in \eqref{eq:reciprocal_modulus}, we tabulate the composed maps
\begin{equation*} \label{eq:LUT_definition}
       \widetilde{E}(\varsigma, \nu) = E(\arcsin\varsigma, \nu), \qquad
       \widetilde{F}(\varsigma, \nu) = F(\arcsin\varsigma, \nu), \qquad
       (\varsigma, \nu) \in [0, 1]^2,
\end{equation*}
on the power-law grid
\begin{equation} \label{eq:LUT_sampling}
       \varsigma_n = (1 - \varepsilon)\left[ 1 - \left(1 - \tfrac{n-1}{N_\varsigma - 1}\right)^4 \right], \quad
       \nu_{n'} = (1 - \varepsilon)\left[ 1 - \left(1 - \tfrac{n'-1}{N_\nu - 1}\right)^4 \right], \quad (n, n') \in [N_\varsigma] \times [N_\nu],
\end{equation}
where $N_\varsigma$ and $N_\nu$ are the numbers of samples and $\varepsilon = 10^{-16}$ is a small offset that avoids the singularity.
The exponent $4$ concentrates the samples near the singularity $(1,1)$. It has been chosen empirically, together with $N_\varsigma = N_\nu = 1000$.
The cost of this non-uniform sampling is that locating the grid points adjacent to a query $(\varsigma,\nu)$ requires inverting \eqref{eq:LUT_sampling}, that is, two square-root evaluations.

Overall, the LUT reduces the running time of the Exact evaluation by one order of magnitude, with no measurable loss of accuracy (\Cref{tab:exp1_summary}).

\paragraph{Trapezoidal approximation}\label{subsec:trapezoid_method}
When $\alpha_l$ and $\beta_l$ are sufficiently close, the curve $u \mapsto \Psi(u, c\tau_l)$ can be approximated by the segment joining $(\alpha_l,\Psi(\alpha_l, c\tau_l))$ and $(\beta_l,\Psi(\beta_l, c\tau_l))$.
This leads to a trapezoidal approximation of the subdomains $D^l_{\xi_1, \xi_2}$ so that
\begin{equation} \label{eq:trapezoid_Il}
       I_l \approx \tfrac{1}{2}\left(\beta_l - \alpha_l\right)\left[ \Psi(\alpha_l, c\tau_l) + \Psi(\beta_l, c\tau_l) \right].
\end{equation}
This approximation removes the need to evaluate elliptic integrals.
The computation of $I_l$ is reduced to evaluating $\Psi(\cdot, c\tau_l)$ at the two endpoints, which requires two square-root operations.
Among the three proposed methods, this is the fastest one, but it introduces an approximation error.

\subsubsection{Piecewise-Planes}
\label{subsec:plane_method}

It is also possible to approximate the cylindrical transducer $\Sigma$ by a piecewise planar surface \cite{jensen1992calculation, wu1999spatial, baek2012spatial, ilyina2017extension}.
More precisely, the aperture $[-\theta_{\max}, \theta_{\max}]$ is divided into $N_p$ sectors of equal angular width $\Delta\theta = 2\theta_{\max}/N_p$, and the $j$-th sector is replaced by the plane $P_j$ tangent to $\Sigma$ at its central angle $\theta_j = -\theta_{\max} + (j - \tfrac{1}{2})\Delta\theta$, for $j \in [N_p]$.
Each $P_j$ is a planar transducer in the sense of \eqref{eq:plane_param}, of half-width $L_1 = \tfrac{1}{2}R\,\Delta\theta$ and half-height $L_2 = L_z$.

The approximation bears only on the geometric quantity of \eqref{eq:def:bffi}: the area swept on $\Sigma$ is replaced by the sum of the areas swept on the tangent planes,
\begin{equation} \label{eq:plane_decomp_u}
       \left| \mathcal{M}_i^l \right| \approx \sum_{j \in [N_p]} \left| \mathcal{M}_i^l(P_j) \right|,
\end{equation}
where $\left| \mathcal{M}_i^l(P_j) \right|$ is the area \eqref{eq:def:bffi} computed for the planar transducer $P_j$, evaluated by \Cref{alg:separable_area} instantiated for the planar geometry of \Cref{subsec:plane_surface}.
By linearity of (19), the coefficients of the cylinder are then $\bbf \approx \sum_{j\in[N_p]} \bbf(P_j)$, where $\bbf(P_j)$ denotes the coefficients computed for the planar transducer $P_j$.

This decomposition replaces the evaluation of the incomplete elliptic integral by a sum of the more computationally tractable expressions in \eqref{eq:plane_primitive}.
The resulting computational cost therefore depends on the number of planar transducers $N_p$, while the approximation error depends on how accurately the piecewise planar surface represents the curvature of $\Sigma$.
The choice of $N_p$ thus provides a trade-off between computational efficiency and accuracy.

\section{Numerical experiments}
\label{sec:experiments}

The goal of these numerical experiments is to compare the performance of the different forward operator approximations and to assess their scalability at representative, full-reference scale.
To this end, we consider two experimental settings.

First, in \Cref{subsec:exp1_synthetic}, we simulate a moderate-size problem involving the cylindrical transducer introduced in \Cref{subsec:lib_setup}.
A few transducers are positioned around the imaged region.
Combined with the selected phantom, this setting provides a controlled test of the ability of the approximations to account for the SIR of the transducers, while retaining the main features of a complete imaging scenario.
We compare five operators: Point (\Cref{subsec:point_method}) and the four cylindrical implementations of the surface method, Exact, LUT, Trapezoidal (\Cref{subsec:elliptic_method}) and Piecewise-Planes (\Cref{subsec:plane_method}).
Since the Exact operator remains affordable at this scale, we can directly compare the accuracy and computational cost of each approximation, addressing the practical question of which operator to use for a given computational budget.

Second, in \Cref{subsec:exp2_vessel}, we simulate a realistic reconstruction scenario with the system of \Cref{subsec:lib_setup} at the full reference scale on an anatomically realistic vascular phantom.
We compare the performance of the best model identified in the first experiment against the reconstruction method designed for the system.

Beforehand, \Cref{subsec:methodology} introduces the reconstruction machinery and metrics shared by both.

\subsection{Common methodology}
\label{subsec:methodology}

Rather than comparing the different implementations on the forward problem, for example through operator norms, we evaluate them based on the quality of the reconstructions they produce.
Starting from a phantom $\bp_0^\star$, the corresponding measurements $\bs$ will be generated using a very fine approximation of the forward operator, which is taken as the reference operator.
This reference operator will not be used in the subsequent reconstructions to avoid the inverse crime, and cannot be used in practice due to its prohibitive computational cost.

In all experiments, the radial basis function is chosen as the cone function $\phi(r) = \max(0,1 - \tfrac{r}{\kappa})$, with $\kappa$ set equal to the spacing between adjacent grid points.
This corresponds to the interpolatory stationary regime considered in \cite{ding2020model}.

\paragraph{Reconstruction problem}

Given a candidate forward operator $\bA$ and measurements $\bs$, the initial pressure is reconstructed by solving the Tikhonov-regularized non-negative least-squares problem
\begin{equation} \label{eq:nnls_tik}
       \widehat{\bp}_0 = \argmin_{\bp \geq 0}
       \frac{1}{2}\| \bA\bp - \bs \|_2^2 + \frac{1}{2}\lambda_R \|\bp\|_2^2
\end{equation}
where $\lambda_R \geq 0$ is the regularization parameter.
The non-negativity constraint is motivated by the physics of photoacoustics, as the absorbed optical energy gives rise to a non-negative initial pressure increase $p_0 \ge 0$ (see \Cref{subsec:wave_prop}).

It is important to emphasize that \eqref{eq:nnls_tik} is not intended to provide state-of-the-art reconstruction quality.
Rather, this deliberately simple reconstruction formulation is chosen to minimize the effects introduced by regularization when comparing the different operator implementations.
Obviously, the same objective and solver will be used for every operator, so that any difference in the reconstruction can be mainly attributed to the forward operator $\bA$.

The optimization \eqref{eq:nnls_tik} is solved with the limited-memory quasi-Newton scheme L-BFGS-B \cite{byrd1995limited, zhu1997algorithm}.
The iterations are stopped when either the relative decrease of the objective between two successive iterates, or the largest component of the projected gradient (the first-order optimality measure under the non-negativity constraint) falls below $5 \times 10^{-7}$.
These tight tolerances ensure that the computed solution is sufficiently close to a minimizer of the objective function, making the effect of early stopping negligible.

\paragraph{Metrics}
The reconstruction accuracy will be assessed against the ground-truth phantom $\bp_0^\star$ with two image metrics: peak signal-to-noise ratio (PSNR) and the 3D structural similarity index ($\mathrm{SSIM}$) \cite{wang2004image}.

\paragraph{Hardware and software environment}
All experiments run on one NVIDIA A100 GPU ($80\,\mathrm{GB}$) of the Jean Zay supercomputer (IDRIS, CNRS), with all arrays in double precision.

Following the analysis of \Cref{subsec:comparison}, the memory footprint of the operators is linear in $U$, independent of the reconstruction grid, and identical for the four surface operators at a given $U$.
Our implementation does not exploit the time gating of \Cref{subsec:lib_setup} and stores the full recording of $L = 2\,688$ samples ($43\,\mu\mathrm{s}$) for every transducer, so that the $4KUL$ values amount to $60\,\mathrm{GB}$ at $U = 61$ and $11\,\mathrm{GB}$ at $U = 11$.
The $80\,\mathrm{GB}$ are needed only once, to synthesize the reference measurements of Experiment 2 with the Exact operator at $U = 61$, which peaks at $61\,\mathrm{GB}$.
The reconstruction itself, repeated at every iteration, peaks at $13\,\mathrm{GB}$ with the LUT operator at $U = 11$, including the variables of the solver, and fits on a $16\,\mathrm{GB}$ GPU.

\paragraph{Cylindrical transducer} In both experiments, we use cylindrical transducers that model those of the system described in \Cref{subsec:lib_setup}.
Their main characteristics, as well as their physical implications for image formation, are summarized here.

The transducers have a radius, also referred to as focal distance, of $R = 25\,\mathrm{mm}$ and a half-aperture angle of $\theta_{\max} \approx 8.5^\circ$.
Their center frequency is $F_c = 5\,\mathrm{MHz}$, and their EIR $e$ is modeled as a third-order Butterworth band-pass filter with a bandwidth of $2$--$10\,\mathrm{MHz}$.

A cylindrical transducer has a focal zone centered around its focal point ($\bc_0$ in \Cref{fig:cylinder}), located at the focal distance $R$.
The signals generated by approximately isotropic pressure sources located within this focal region are only weakly affected by the spatial impulse response (SIR), resulting in the highest sensitivity of the detector.
The focal zone is anisotropic, with its dimensions determined by diffraction and primarily governed by the center frequency and aperture of the transducer.
Specifically, this transducer has a focal zone of width $\approx 1 \,\mathrm{mm}$ (along $\be_1 \times \be_2$) and depth of field $\ge 20 \,\mathrm{mm}$ (along $\be_2$).

These properties determine the spatial scales that can be resolved by the system and thus its achievable image resolution.
In particular, in the focal zone, the finest achievable resolution is half the central wavelength, defined as $\frac{\lambda_c}{2} = \frac{c}{2 F_c} = 150\,\mu\mathrm{m}$.

\subsection{Experiment 1: Controlled synthetic study}
\label{subsec:exp1_synthetic}

This experiment has been carefully designed to assess how the various approximations account for the SIR of the transducer and to reveal the differences between them, while remaining of moderate numerical complexity so that all operators can be evaluated with reasonable computational times.

A rotational acquisition in which the cylindrical transducer is rotated around its focal point is considered.
The isotropic sources located close to the rotation axis remain near the focal point for all transducer orientations and are only weakly affected by the SIR.
As the radial distance from the axis increases, the sources progressively move away from the focal point, resulting in a stronger SIR effect.

The experiment thus provides a range of SIR effects on which the different approximations can be compared.

\subsubsection{Setup}
\label{subsec:exp1_setup}

\paragraph{Acquisition geometry}

The acquisition geometry and the domain $\Omega$ are illustrated in \Cref{fig:exp_array}.

The cylindrical transducer is rotated around its focal point (axis of rotation $\be_1$, \Cref{fig:cylinder}) at $43$ azimuthal positions, spaced by $360^\circ/43 \approx 8.4^\circ$, on a cylinder of radius $R=25\,\mathrm{mm}$.
Since each transducer has a half-angular aperture of $\approx 8.5^\circ$, their positions provide a full azimuthal sampling.
This circular arrangement at five positions along the rotation axis ($z$-axis) spans $36.7\,\mathrm{mm}$ and results in $K = 215$ cylindrical transducers.

The domain $\Omega$ is a $6 \times 6 \times 1.05\,\mathrm{mm}^3$ volume. Its dimensions in the $x$-$y$ plane are chosen small enough to reduce computational costs while providing a sufficiently wide range of SIR effects.
The third dimension of the ROI is limited to enable the reconstruction of nearly isotropic sources, given the diffraction limited resolution in this dimension.
The $z$-coordinates of the extreme transducer positions were chosen to cover, at the center frequency, the largest angular aperture given the dimension $L_z$ of the detectors.
Intermediate positions are added to improve the reconstruction in this third dimension.

\begin{figure}[!htb]
       \centering
       \begin{subfigure}[b]{0.36\linewidth}
              \centering
              \includegraphics[width=\linewidth]{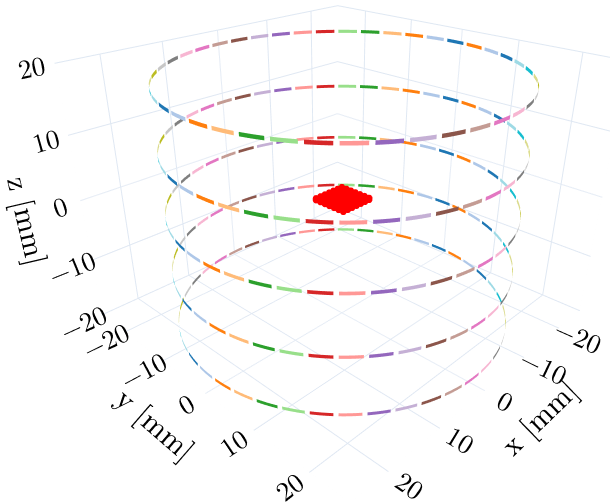}
              \caption{Acquisition geometry.}
              \label{fig:exp_array}
       \end{subfigure}
       \begin{subfigure}[b]{0.31\linewidth}
              \centering
              \includegraphics[width=\linewidth]{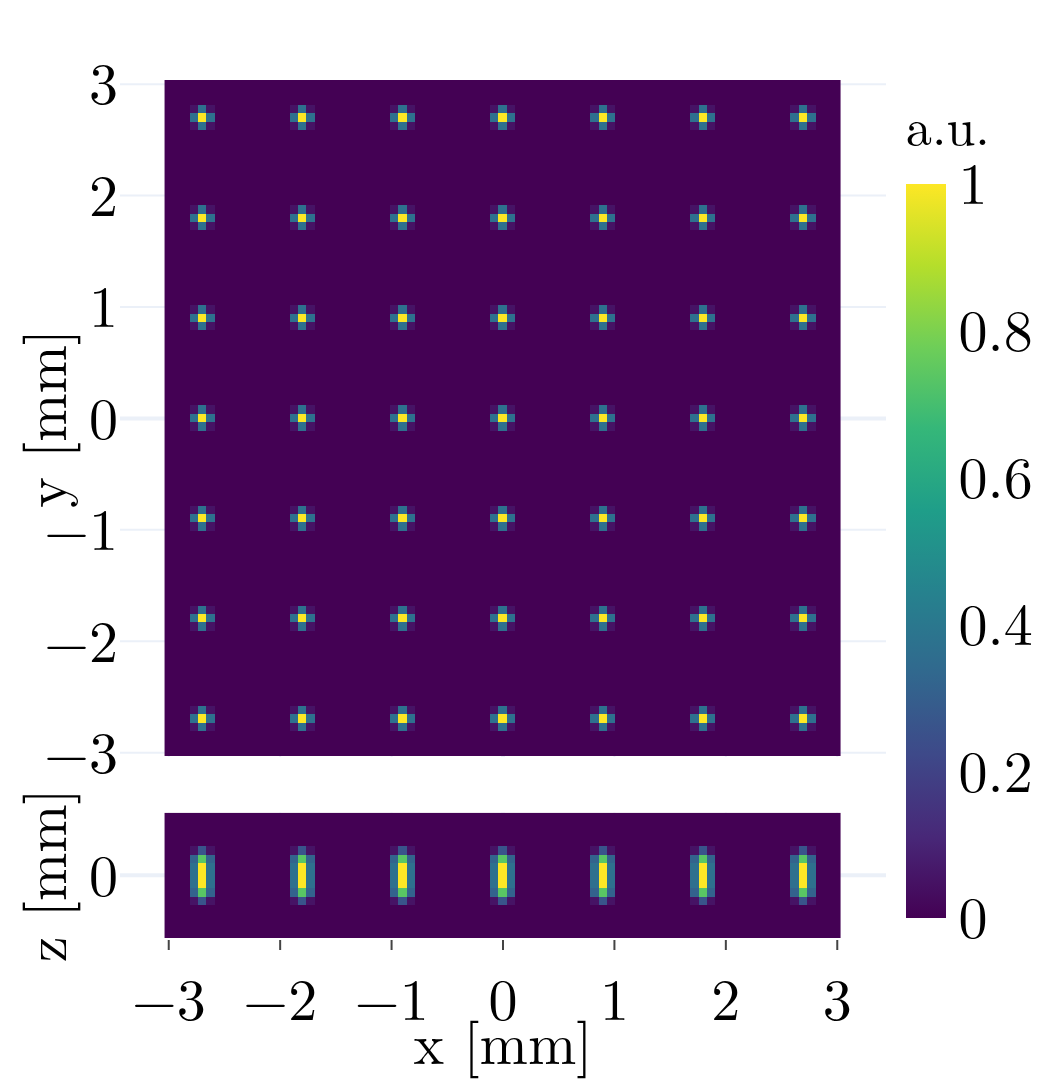}
              \caption{Phantom, fine grid.}
              \label{fig:exp1_phantom_fine}
       \end{subfigure}
       \begin{subfigure}[b]{0.31\linewidth}
              \centering
              \includegraphics[width=\linewidth]{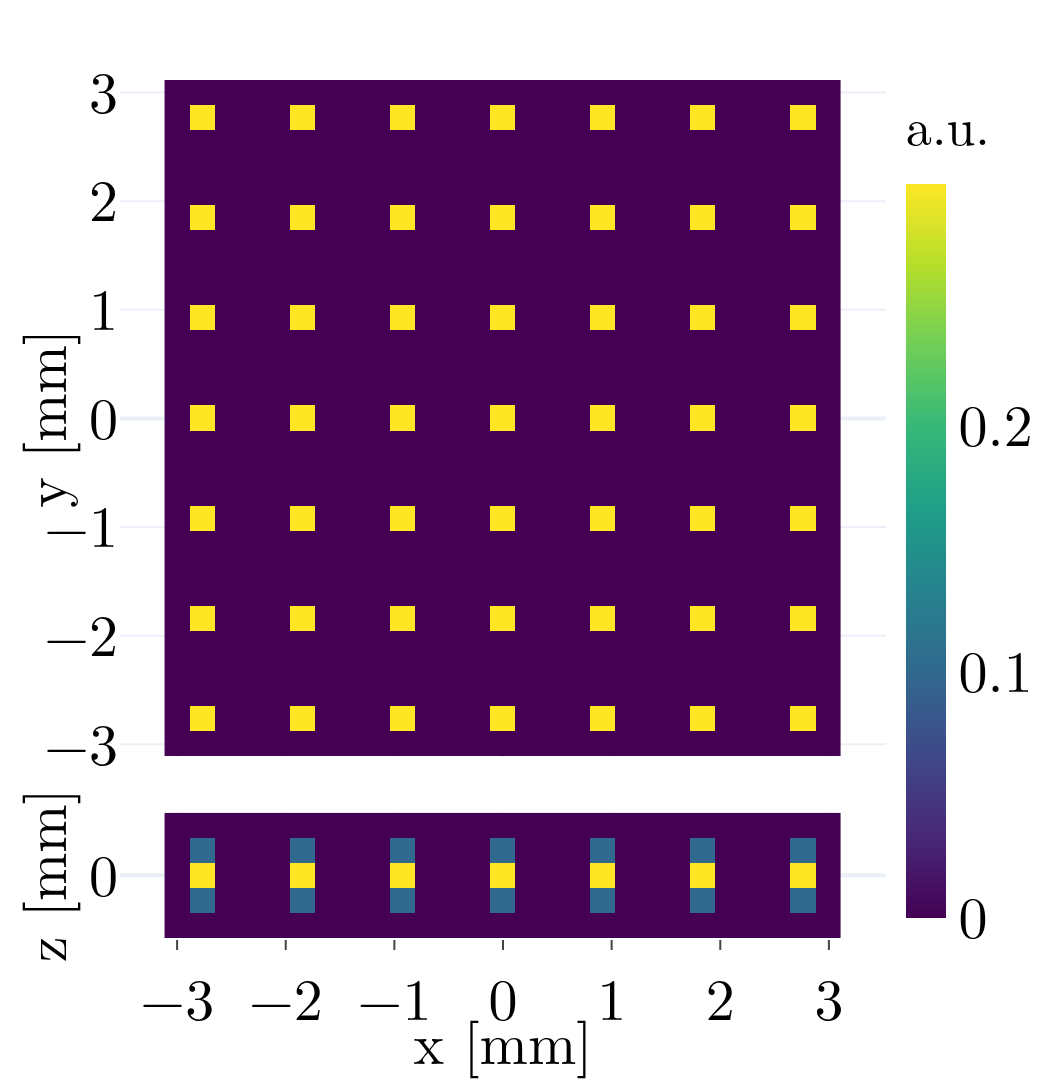}
              \caption{Phantom, coarse grid.}
              \label{fig:exp1_phantom_coarse}
       \end{subfigure}
       \caption{Experiment 1 setup. (a) Acquisition geometry: $K = 215$ elements ($5$ along $z$, $43$ rotations) on a cylinder of radius $25\,\mathrm{mm}$, with the phantom at the center (red). (b), (c) Phantom: a $7 \times 7$ grid of $49$ rods, shown on the fine synthesis grid $(81, 81, 15)$ and the coarse reconstruction grid $(27, 27, 5)$; top row: $z = 0$ slice ($x$--$y$ plane), bottom row: $y = 0$ slice ($x$--$z$ plane). The phantom is invariant under a $90^\circ$ rotation about the $z$ axis, so the $y$--$z$ slice is identical to the $x$--$z$ one and is not shown.}
       \label{fig:exp1_setup}
\end{figure}

\paragraph{Phantom}
We discretize $\Omega$, the $6 \times 6 \times 1.05\,\mathrm{mm}^3$ volume, on two grids.

The first one is an $81 \times 81 \times 15$ grid, yielding an isotropic voxel size of $75\,\mu\mathrm{m}$ in each dimension.
This corresponds to a quarter of the acoustic wavelength in water $\lambda_c = c/F_c = 0.30\,\mathrm{mm}$ at the center frequency $F_c = 5\,\mathrm{MHz}$.
The phantom consists of a $7 \times 7$ grid of absorbing rods centered on the $z=0$ plane and spaced $0.9\,\mathrm{mm}$ apart (\Cref{fig:exp1_phantom_fine}).
Each rod has a smooth paraboloidal absorption profile with an effective support of $225 \times 225 \times 525\,\mu\mathrm{m}^3$, that is $3 \times 3 \times 7$ voxels.
The rod dimensions are chosen to match the angular coverage of the detection geometry and generate photoacoustic signals whose frequency content is mostly contained within the detector EIR bandwidth.
This ensures that the rods can be reconstructed without significant information loss due to the detector response.

The second one is a three times coarser grid $(27,27,5)$ with an isotropic voxel size of $225\,\mu\mathrm{m}$ in each dimension.
On this grid, the ground truth is obtained by averaging the fine phantom over each $3 \times 3 \times 3$ block of voxels (\Cref{fig:exp1_phantom_coarse}).

\paragraph{Reference signals}

For both grids, the reference measurements are synthesized on the fine one by the Exact operator (\Cref{subsec:elliptic_method}) at temporal upsampling $U = 61$, evaluating the areas $|M_i^l|$ on every time interval ($\Delta l = 1$ in \Cref{rmk:coarser_area_grid}), fine enough that its quantization error, of order $(UF_s)^{-1}$ (\Cref{prop:approximation_surface}), is negligible.

The reconstruction on the fine grid deliberately goes beyond the resolution limit of the detectors and is necessary to observe differences in reconstruction quality between the various operators.
This setting is close to an inverse crime since the grid for generation and reconstruction is the same. However, the inverse crime is avoided since the operator generating the signals is not used for reconstruction.

The coarser reconstruction grid is closer to the resolution limit of the detectors and therefore has a voxel size more representative of what would be chosen in a real experiment.
It also introduces a discretization mismatch between the generation and reconstruction models, as would occur in practice, and allows us to compare the five operators under a non-zero model mismatch.

Two noise levels are also tested: noiseless, and additive Gaussian noise with standard deviation equal to $1\%$ of the maximum signal amplitude.

\paragraph{Operator and optimization parameters}

For all reconstruction operators we choose the upsampling factor $U = 11$. This choice is motivated by the following considerations.

The system described in \Cref{subsec:lib_setup} has $K = 11\,520$ transducers, each stored with $L = 2\,688$ time samples.
With $U = 11$, a reconstruction requires approximately $13\,\mathrm{GB}$ of memory, which fits on a medium-size graphics card.
Since we aim to draw conclusions from this experiment that extend to the full-scale problem considered in the second experiment, this is the upsampling we use here too.
Moreover, $U = 11$ makes the quantization error of the operators (\Cref{prop:approximation_points,prop:approximation_surface}) negligible.
In experiments not reported here, larger values of $U$ did not yield any noticeable improvement in reconstruction quality.

The Point operator uses $Q = 625$ quadrature points to discretize each transducer. This corresponds to five points per central wavelength $\lambda_c$ and is necessary to accurately account for acoustic interference.
Increasing $Q$ leads to prohibitive computation times without significant gains in reconstruction quality, whereas smaller values result in degraded reconstructions.

The LUT operator uses the $N_\varsigma = N_\nu = 1000$ tables of \eqref{eq:LUT_sampling} as justified in \Cref{subsec:LUT_method}. The Exact and Trapezoidal operators have no further parameters.
The LUT, Exact and Trapezoidal operators use the strategy of \Cref{rmk:coarser_area_grid} with $\Delta l = U$ which provides the best trade-off in terms of accuracy and computational cost.

The Piecewise-Planes operator uses $N_p = 30$ tangent planes. This choice is made to reduce the curvature errors as much as possible, while keeping computation times below those of the Exact operator.

The regularization parameter $\lambda_R$ of \eqref{eq:nnls_tik} is tuned for each grid-noise configuration by a manual search maximizing the PSNR, and shared by the five operators: $\lambda_R = 10^{-4}$ (fine grid, noiseless), $10^{-3}$ (fine, $1\%$ noise), $10^{-1}$ (coarse, noiseless) and $10^{-1}$ (coarse, $1\%$ noise).
The values increase with the noise level and are larger on the coarse grid where model mismatch occurs and calls for stronger regularization.

\subsubsection{Results}
\label{subsec:exp1_result}

The reconstructed images are presented in \Cref{fig:exp1_recon} and the associated metrics in \Cref{tab:exp1_summary}. \Cref{fig:exp1_SNR_vs_Time} provides a diagram illustrating the accuracy--cost trade-off for all operators.
All images feature the rotational symmetry of the detection geometry and the effects of the SIR are, as expected, larger as the radial distance to the axis increases.

On the fine grid without noise, the reconstructions of Exact, LUT and Trapezoidal operators are visually similar.
The rods are well reconstructed, especially in the center where SIR effects are negligible. A slight decrease of resolution is observed at the corners, due to the SIR effects being larger in this area.
The Exact and LUT operators achieve essentially identical PSNR values, while the Trapezoidal operator shows a slight decrease of $0.3~\mathrm{dB}$.
The Point and Piecewise-Planes reconstructions perform worse, with PSNR decreases of $1.2~\mathrm{dB}$ and $0.7~\mathrm{dB}$, respectively.
This is due to the amplitude overestimation at the central rod, despite its accurate location.

On the fine grid with noise, all reconstructions are visually similar and achieve comparable quantitative metrics.
With the added noise, the measurement signals are corrupted, and the algorithm struggles to invert the high-frequency content, which is progressively more attenuated by the SIR as the radial distance from the axis increases. This expected behavior is reflected in the smear of the rods observed, that is a progressive degradation of the resolution, with increasing radial distance from the axis.

On the coarser grid, the reconstructions given by all operators are similar and are unchanged by the addition of noise.
Similarly to the finer grid, the resolution of the rods progressively degrades with increasing radial distance from the axis.
The reconstructions are degraded compared to the fine grid because of the model mismatch which is larger than the noise level.

Regarding running times, the numbers of iterations between the algorithms are similar, and the timing differences are due to the computation times of the various operators.
The LUT operator is about 10 times faster than the Exact operator and 3 times faster than the Point operator.
The Trapezoidal operator is, however, about 20\% faster than the LUT operator on the fine grid and 40\% to 45\% faster on the coarse one.

Overall, all operators are able to take into account the SIR effects of the transducers. Except for the noiseless fine grid experiments, all operators perform similarly in terms of quality but with different computation times.
Across all experiments the LUT operator matches the accuracy of the Exact operator while being the second-fastest operator.
The Trapezoidal operator is the fastest overall but can lead to slight drops in reconstruction quality, especially visible on the fine grid.
This also suggests that the LUT operator is well-optimized despite requiring the inversion of the non-uniform indices \eqref{eq:LUT_sampling} and two bilinear interpolations, both of which are more expensive than the trapezoidal approximation.
The Piecewise-Planes operator is computationally expensive while not providing the best reconstruction quality.
Approximating the curved element by flat facets, as is customary when no closed form is available~\cite{jensen1992calculation,jensen1996field,wu1999spatial}, is therefore not an efficient strategy once a closed form is known.

\begin{figure}[!htb]
	\centering
	\setlength{\tabcolsep}{1.5pt}
	\newcommand{\rw}{0.235\linewidth}
	\newcommand{\gw}{0.45\linewidth}
	\newcommand{\rp}[1]{\includegraphics[width=\rw]{figures/exp1/#1}}
	\newcommand{\gp}[1]{\includegraphics[width=\gw]{figures/exp1/#1}}
	\newcommand{\rowlab}[1]{\raisebox{0.1\linewidth}{\rotatebox[origin=c]{90}{\small #1}}}
	\begin{tabular}{@{}c cccc@{}}
		& \multicolumn{2}{c}{\gp{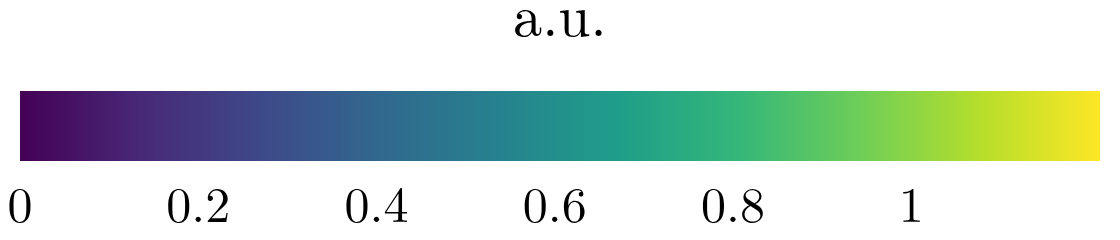}}
		& \multicolumn{2}{c}{\gp{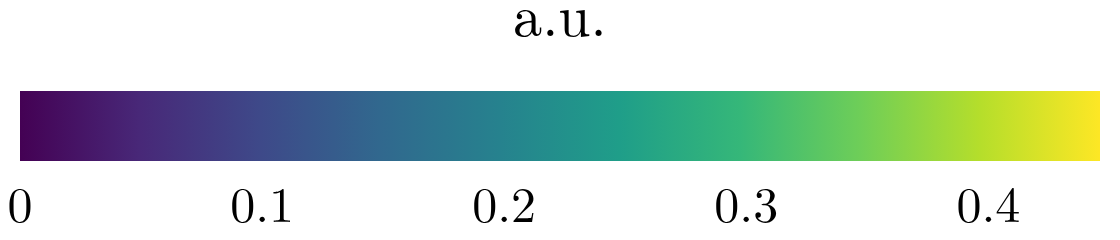}} \\[0pt]
		& \shortstack{Fine noiseless} & \shortstack{Fine $1\%$ noise}
		& \shortstack{Coarse noiseless} & \shortstack{Coarse $1\%$ noise} \\[2pt]
		\rowlab{Point}
			& \rp{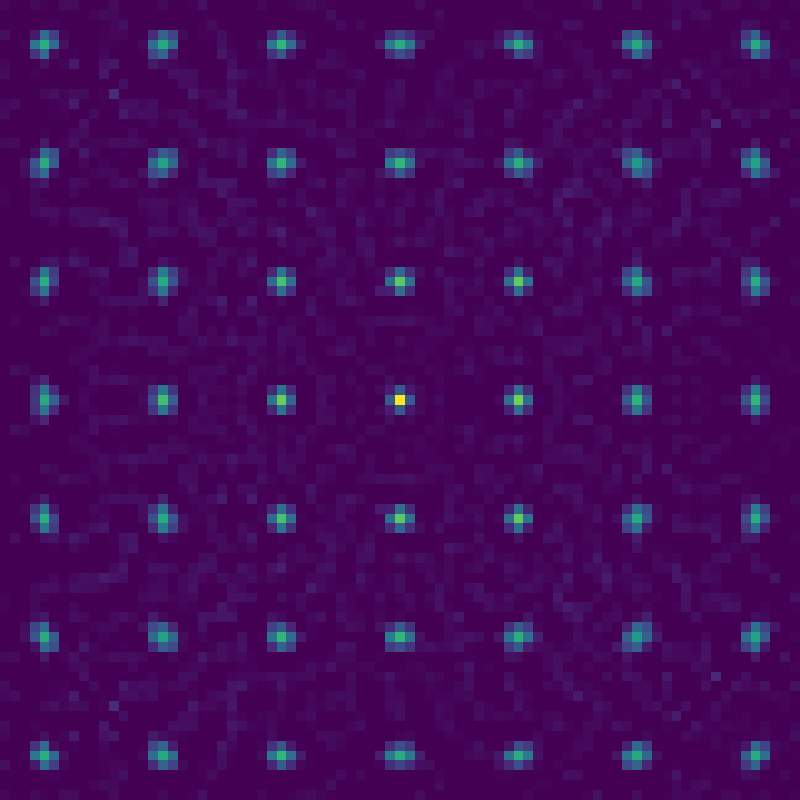} & \rp{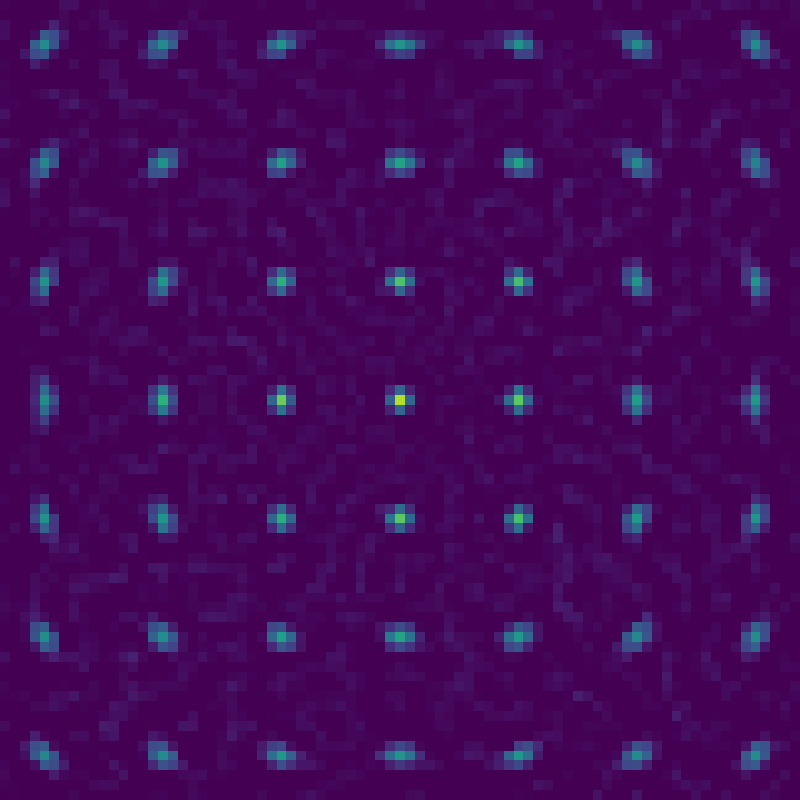}
			& \rp{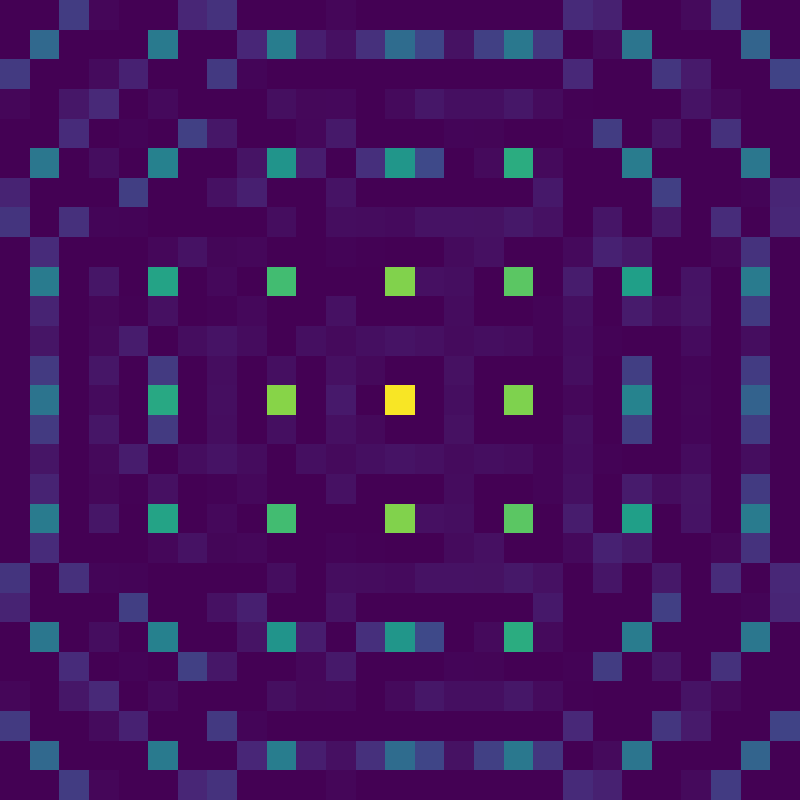} & \rp{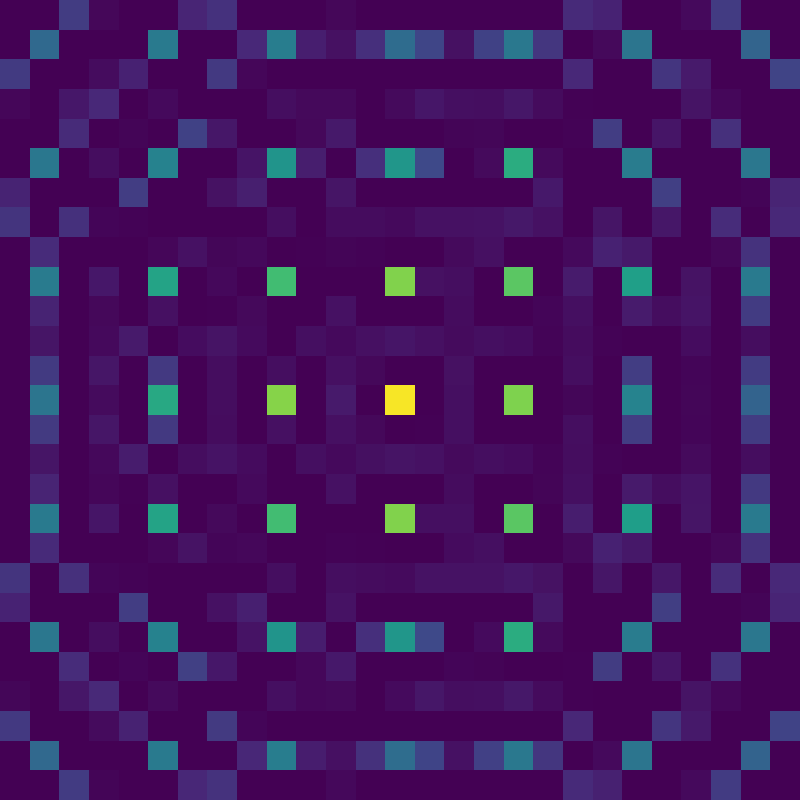} \\[2pt]
		\rowlab{Exact}
			& \rp{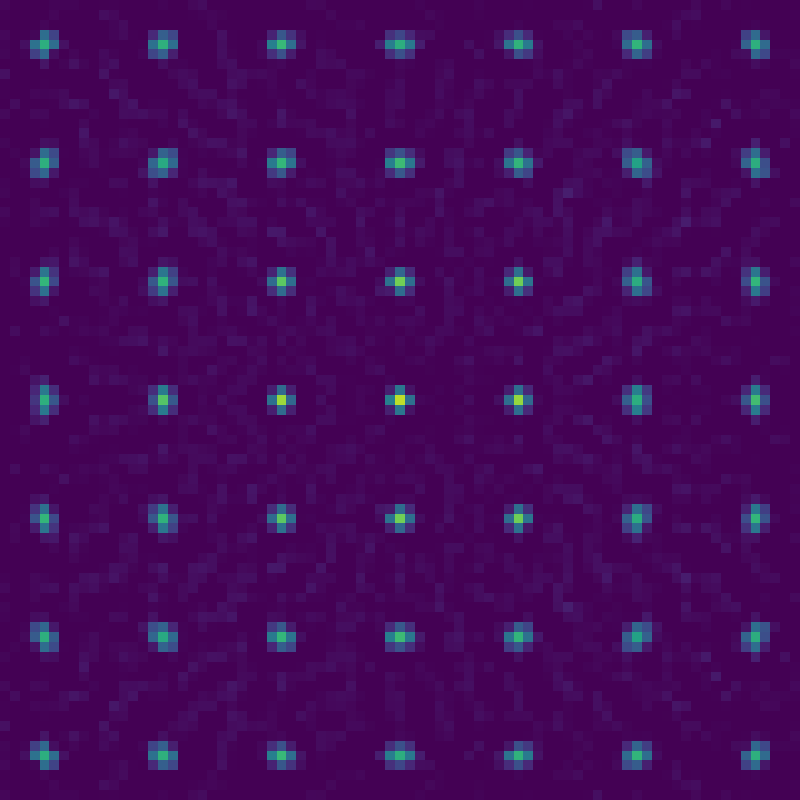} & \rp{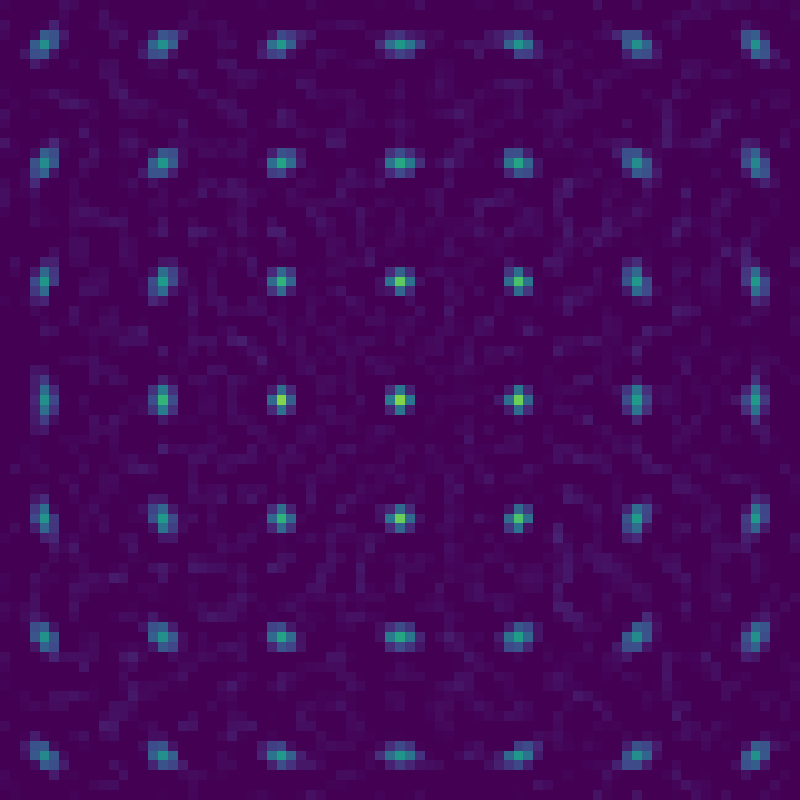}
			& \rp{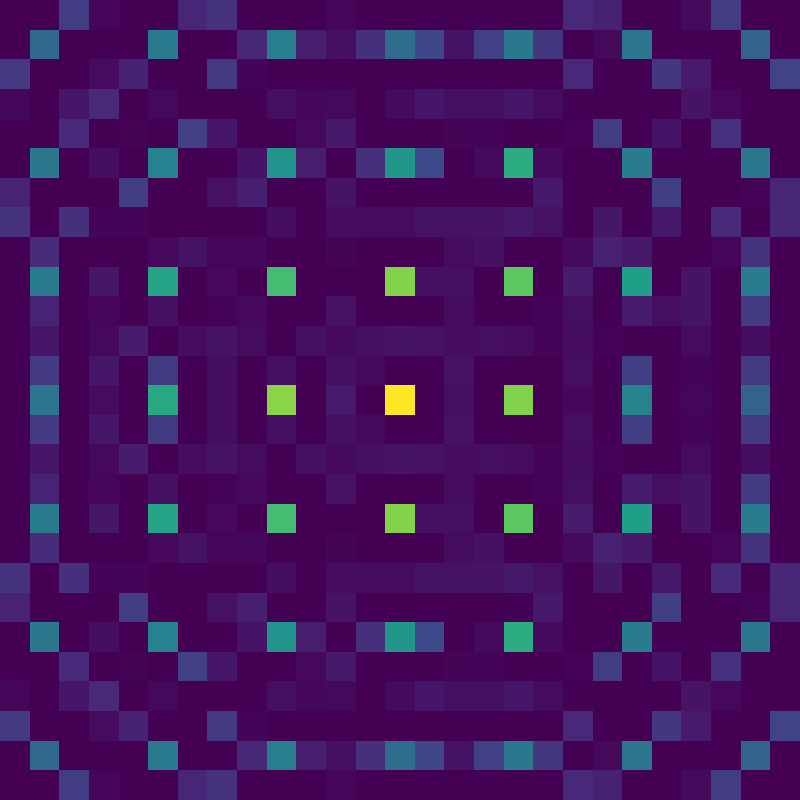} & \rp{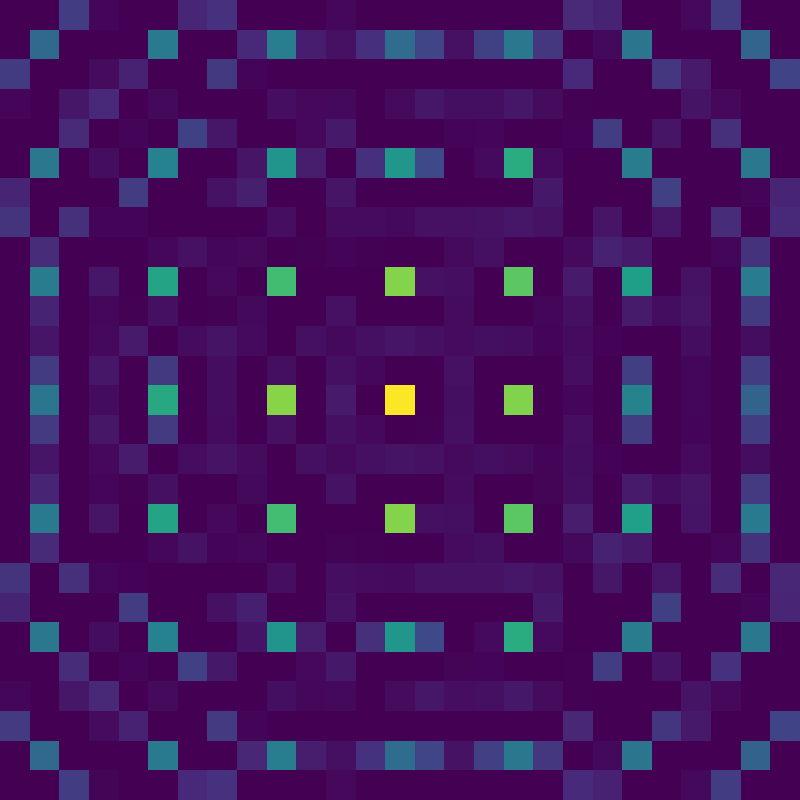} \\[2pt]
		\rowlab{LUT}
			& \rp{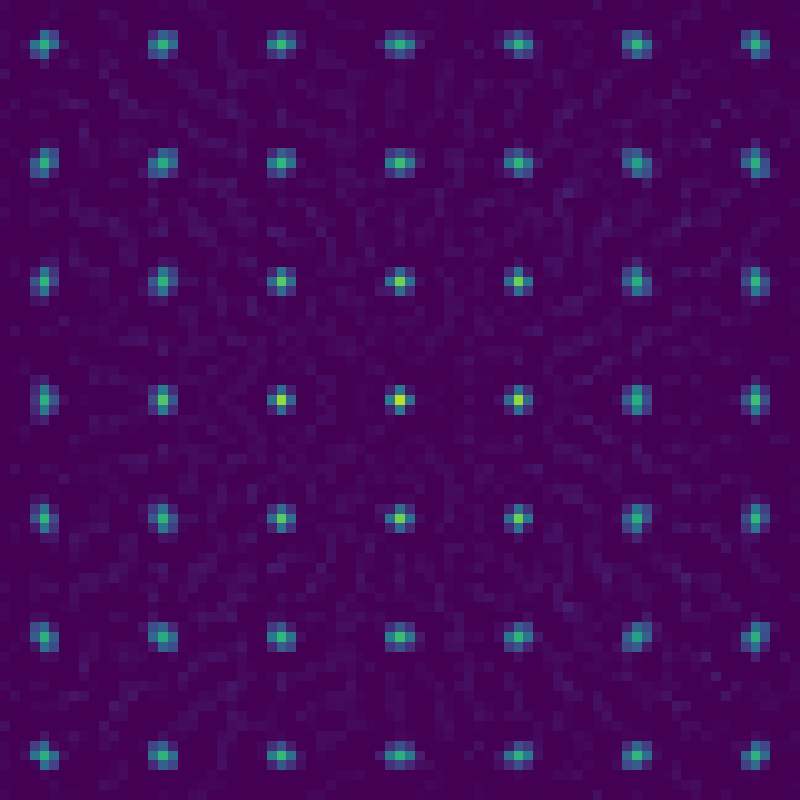} & \rp{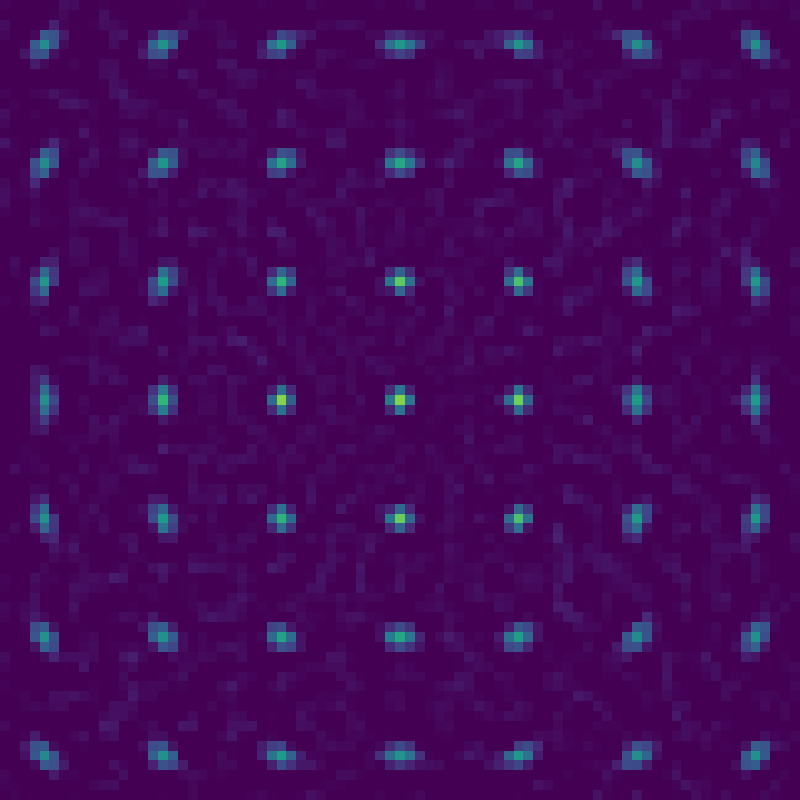}
			& \rp{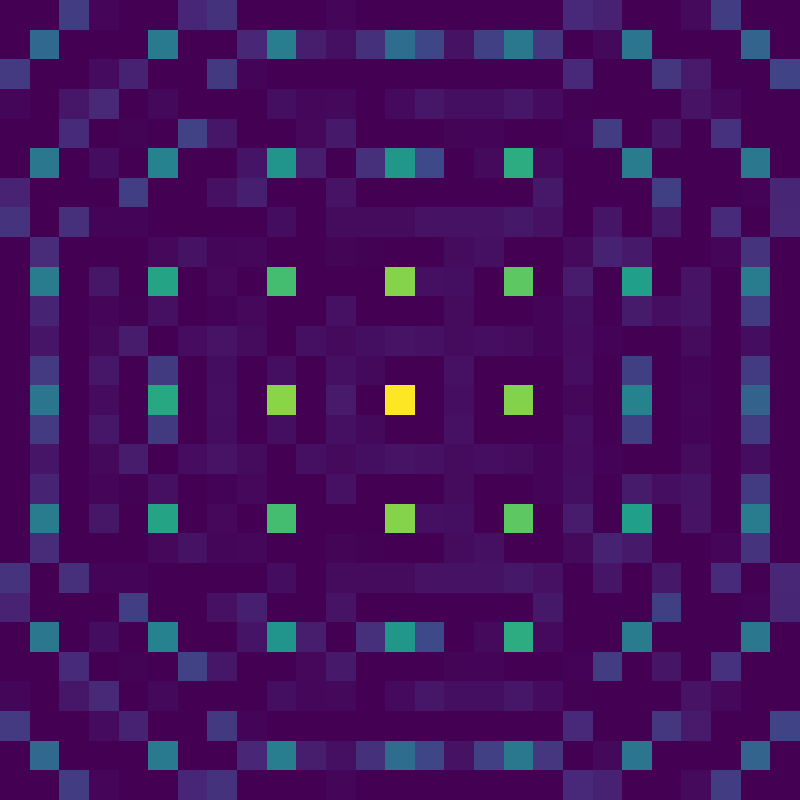} & \rp{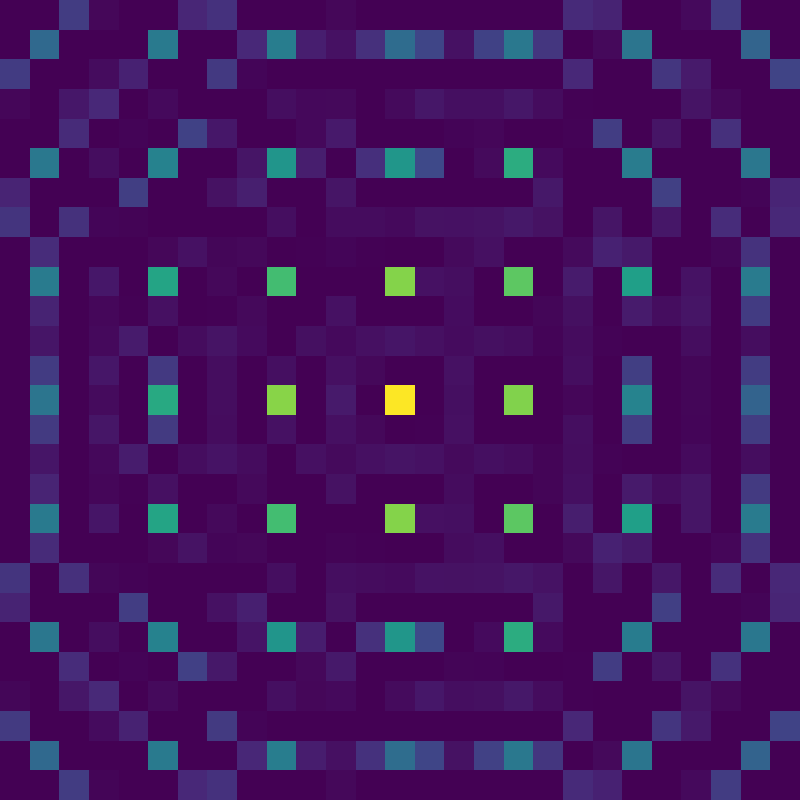} \\[2pt]
		\rowlab{Trapezoidal}
			& \rp{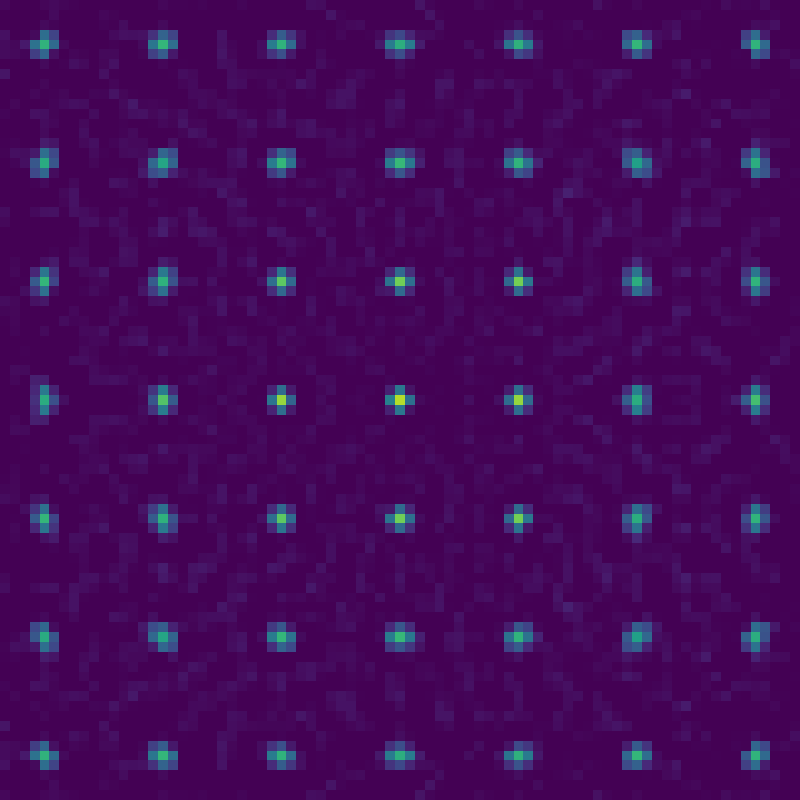} & \rp{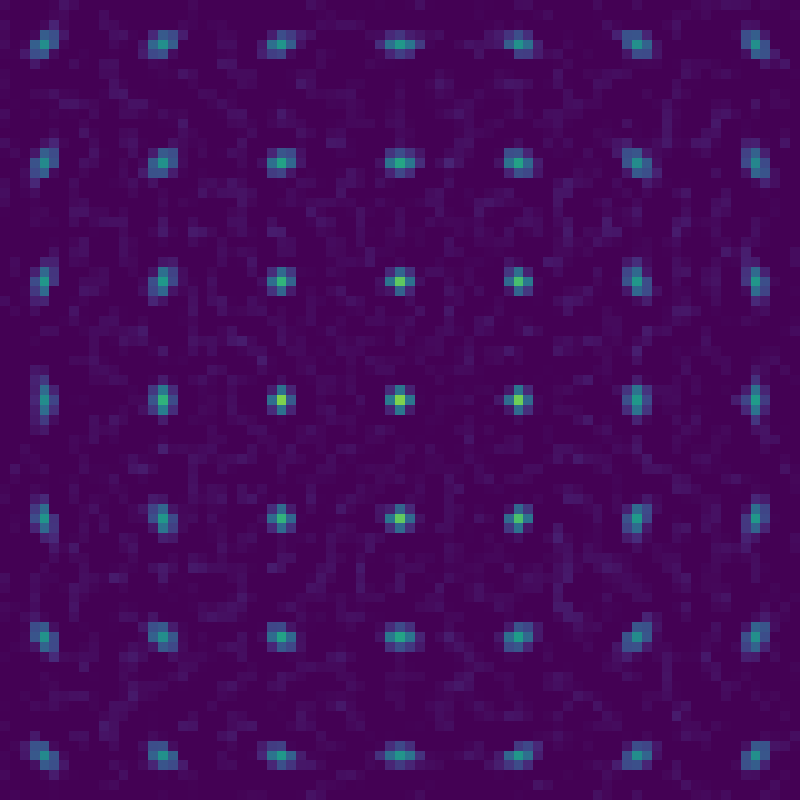}
			& \rp{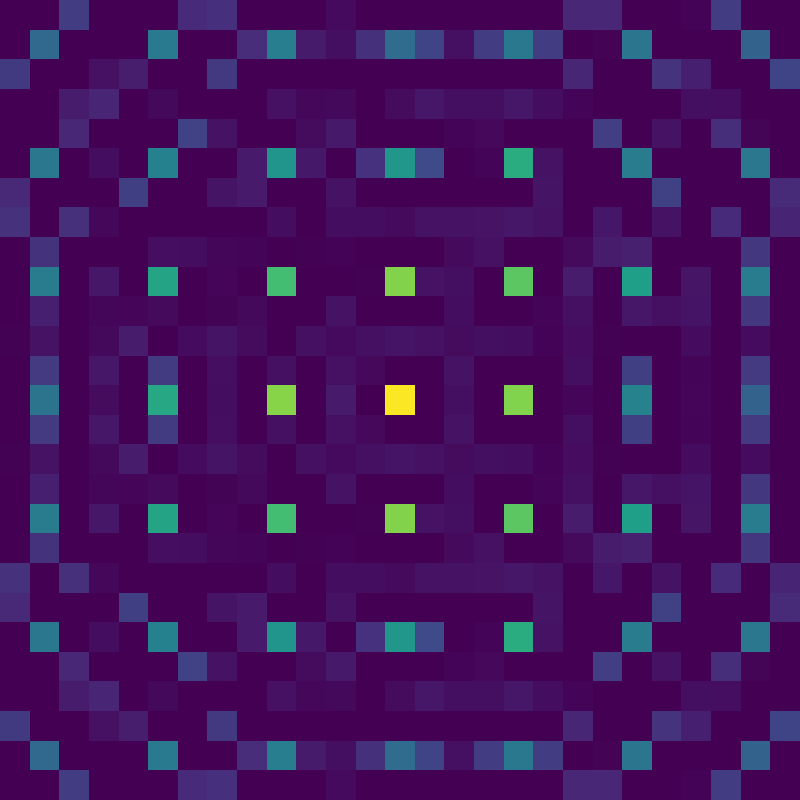} & \rp{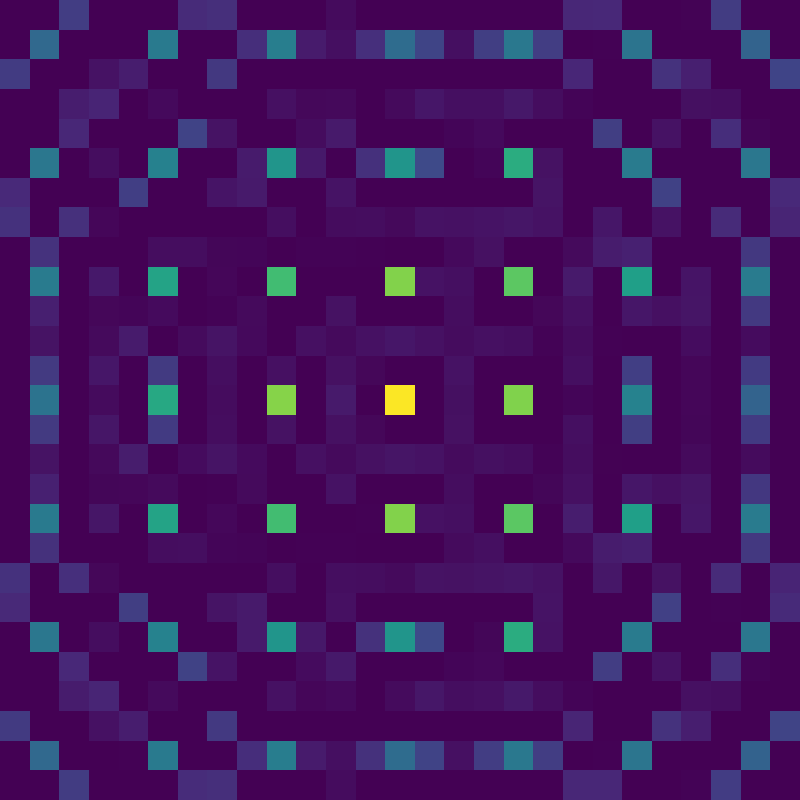} \\[2pt]
		\rowlab{Piecewise-Planes}
			& \rp{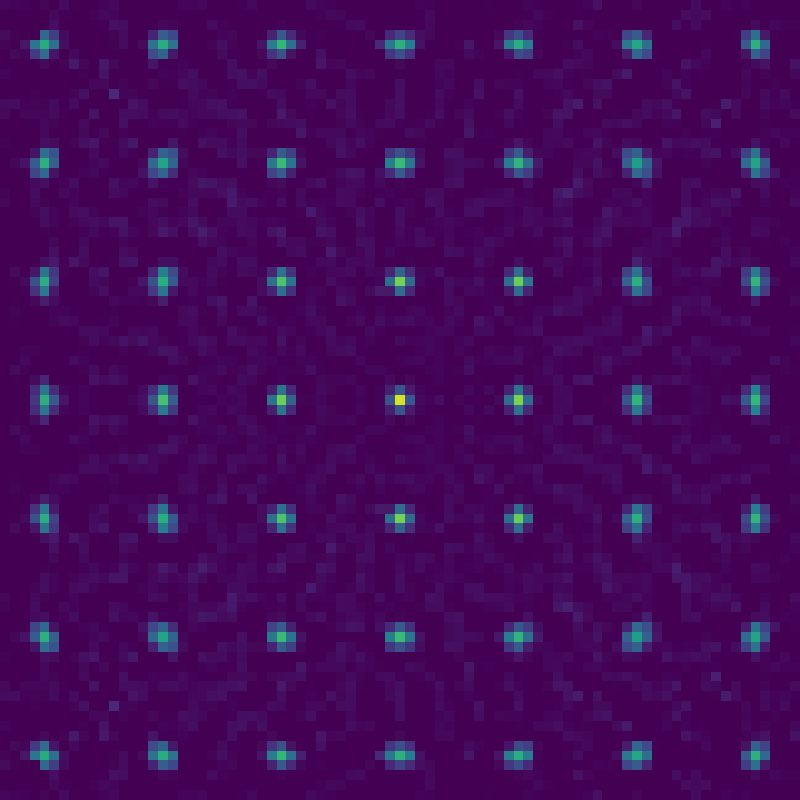} & \rp{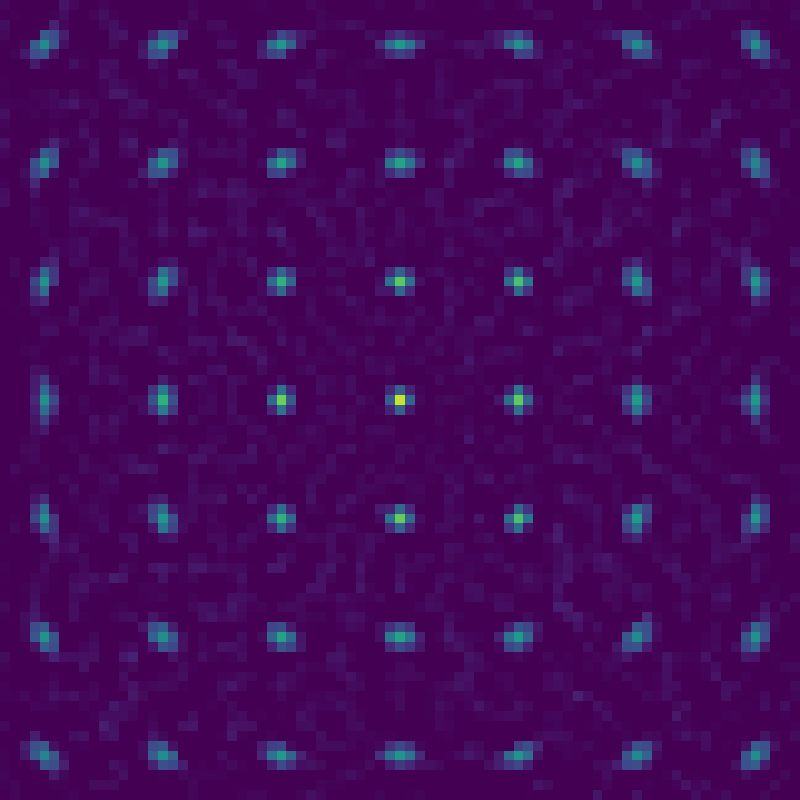}
			& \rp{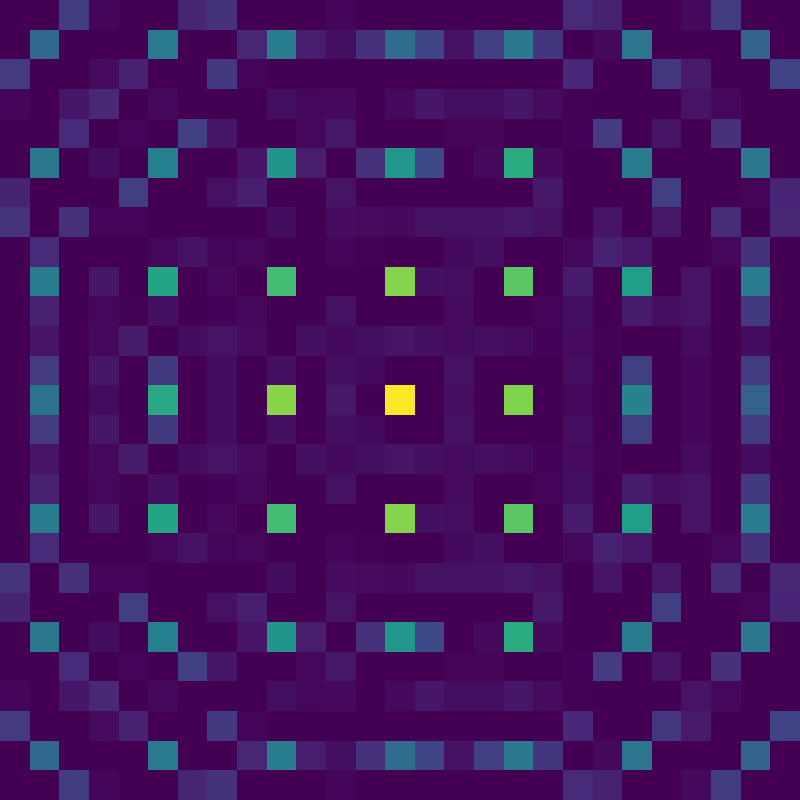} & \rp{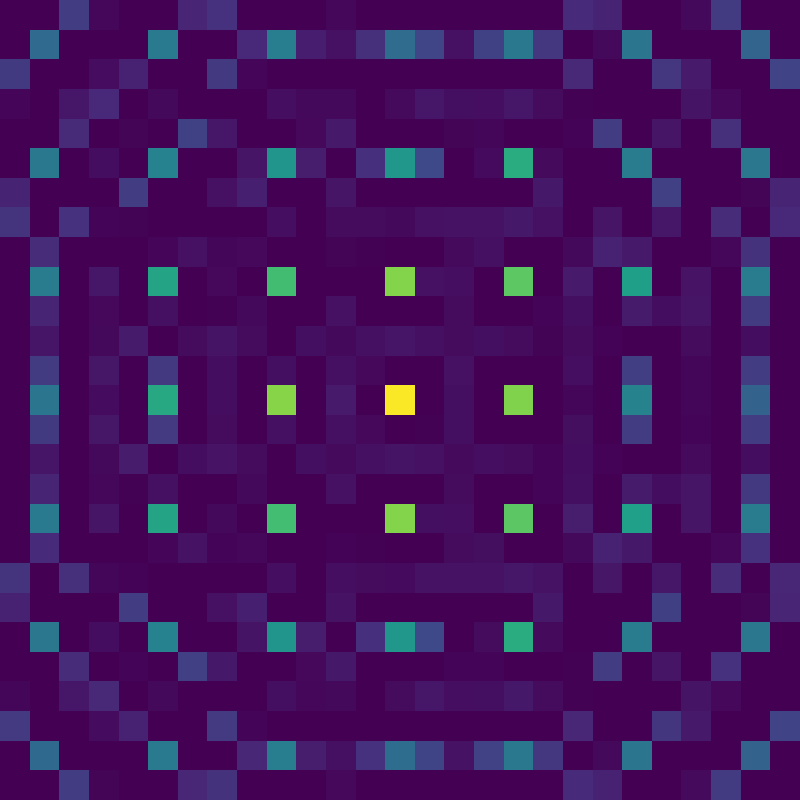} \\
	\end{tabular}
	\caption{Experiment 1 -- central $z = 0$ slice of the five operators' reconstructions (rows), for the fine and coarse grids without and with $1\%$ noise (columns) with shared colorbar for each grid (top).}
	\label{fig:exp1_recon}
\end{figure}

\begin{table}[!htb]
       \centering
       \resizebox{\linewidth}{!}{\setlength{\tabcolsep}{2pt}
\begin{tabular}{l cccc cccc cccc cccc}
    \toprule
    & \multicolumn{8}{c}{Fine grid} & \multicolumn{8}{c}{Coarse grid} \\
    \cmidrule(lr){2-9} \cmidrule(lr){10-17}
    & \multicolumn{4}{c}{noiseless} & \multicolumn{4}{c}{$1\%$ noise}
    & \multicolumn{4}{c}{noiseless} & \multicolumn{4}{c}{$1\%$ noise} \\
    \cmidrule(lr){2-5} \cmidrule(lr){6-9} \cmidrule(lr){10-13} \cmidrule(lr){14-17}
    Method & PSNR & SSIM & Time & \#Ite & PSNR & SSIM & Time & \#Ite & PSNR & SSIM & Time & \#Ite & PSNR & SSIM & Time  & \#Ite \\
    \midrule
    Point            & 31.11 & 0.533 & 101.28 & 190 & 30.55 & 0.574 & 43.11 & 83 & \textbf{24.38} & \textbf{0.718} & 3.03 & 15 & \textbf{24.39} & \textbf{0.717} & 3.03 & 15 \\
    Exact            & \underline{32.35} & \underline{0.602} & 367.88 & 182 & \underline{30.69} & \underline{0.582} & 173.06 & 84 & 24.36 & 0.716 & 8.27 & 15 & 24.36 & \underline{0.716} & 8.81 & 15 \\
    LUT              & \textbf{32.37} & \textbf{0.607} & \underline{35.21} & 183 & \textbf{30.70} & \textbf{0.583} & \underline{15.73} & 84 & 24.36 & 0.716 & \underline{0.83} & 15 & 24.36 & 0.715 & \underline{0.91} & 15 \\
    Trapezoidal      & 32.05 & 0.577 & \textbf{28.12} & 194 & 30.56 & 0.571 & \textbf{12.23} & 84 & 24.32 & 0.714 & \textbf{0.50} & 15 & 24.32 & 0.714 & \textbf{0.50} & 15 \\
    Piecewise        & 31.67 & 0.555 & 233.55 & 186 & 30.55 & 0.567 & 100.68 & 80 & \underline{24.37} & \underline{0.717} & 5.39 & 15 & \underline{24.37} & \underline{0.716} & 5.75 & 15 \\
    \bottomrule
\end{tabular}
}
       \caption{Experiment 1 -- PSNR (dB), SSIM, reconstruction time (s) and iteration count of the five forward operators, on the fine and coarse grids, without and with $1\%$ noise. In each metric column, the best value is in \textbf{bold} and the second best is \underline{underlined} (for time, best means shortest).}
	\label{tab:exp1_summary}
\end{table}

\begin{figure}[!htb]
       \centering
       \begin{tikzpicture}

\definecolor{opPoint}{HTML}{0072B2}
\definecolor{opExact}{HTML}{D55E00}
\definecolor{opLUT}{HTML}{009E73}
\definecolor{opFar}{HTML}{E69F00}
\definecolor{opPlanes}{HTML}{CC79A7}
\pgfplotsset{
  point/.style    ={only marks, color=opPoint, mark=*,        mark size=2.4pt},
  exact/.style    ={only marks, color=opExact, mark=square*,   mark size=2.3pt},
  lut/.style      ={only marks, color=opLUT,   mark=triangle*, mark size=3pt},
  farfield/.style ={only marks, color=opFar,   mark=diamond*,  mark size=3pt},
  planes/.style   ={only marks, color=opPlanes,mark=triangle*, mark options={rotate=180}, mark size=3pt},
}

\begin{groupplot}[
    group style={group size=2 by 1, horizontal sep=2.5cm, vertical sep=1cm},
    width=7cm, height=5.5cm,
    xmode=log, log basis x=10,
    xminorticks=true,
    major tick length=4pt, minor tick length=2.4pt,
    xlabel={time [s]}, ylabel={PSNR [dB]},
    grid=both, major grid style={gray!25}, minor grid style={gray!12},
    tick align=outside,
    title style={font=\normalsize, yshift=-1pt},
]

\nextgroupplot[title={Fine grid --- noiseless},
    xmin=10, xmax=600, xtick={10,100},
    minor xtick={20,30,40,50,60,70,80,90,200,300,400,500,600},
    legend to name=snrlegend, legend columns=5,
    legend style={/tikz/every even column/.append style={column sep=6pt}},
    legend image post style={mark size=2.6pt}]
    \addplot[point]    coordinates {(101.28, 31.11)}; \addlegendentry{Point}
    \addplot[exact]    coordinates {(367.88, 32.35)}; \addlegendentry{Exact}
    \addplot[lut]      coordinates {( 35.21, 32.37)}; \addlegendentry{LUT}
    \addplot[farfield] coordinates {( 28.12, 32.05)}; \addlegendentry{Trapezoidal}
    \addplot[planes]   coordinates {(233.55, 31.67)}; \addlegendentry{Piecewise-Planes}

\nextgroupplot[title={Fine grid --- $1\%$ noise},
    xmin=10, xmax=300, xtick={10,100},
    minor xtick={20,30,40,50,60,70,80,90,200}]
    \addplot[point]    coordinates {( 43.11, 30.55)};
    \addplot[exact]    coordinates {(173.06, 30.69)};
    \addplot[lut]      coordinates {( 15.73, 30.7)};
    \addplot[farfield] coordinates {( 12.23, 30.56)};
    \addplot[planes]   coordinates {(100.68, 30.55)};

\end{groupplot}

\node[anchor=south] at ($(group c1r1.north)!0.5!(group c2r1.north) + (0,0.9cm)$)
     {\pgfplotslegendfromname{snrlegend}};

\end{tikzpicture}
       \caption{Experiment 1 -- accuracy--cost trade-off. Reconstruction PSNR against computational time (log scale) for the five operators, on the fine grid, without and with $1\%$ noise. Up and to the left is better. The coarse grid diagrams are not reported since the reconstruction qualities are comparable.}
       \label{fig:exp1_SNR_vs_Time}
\end{figure}
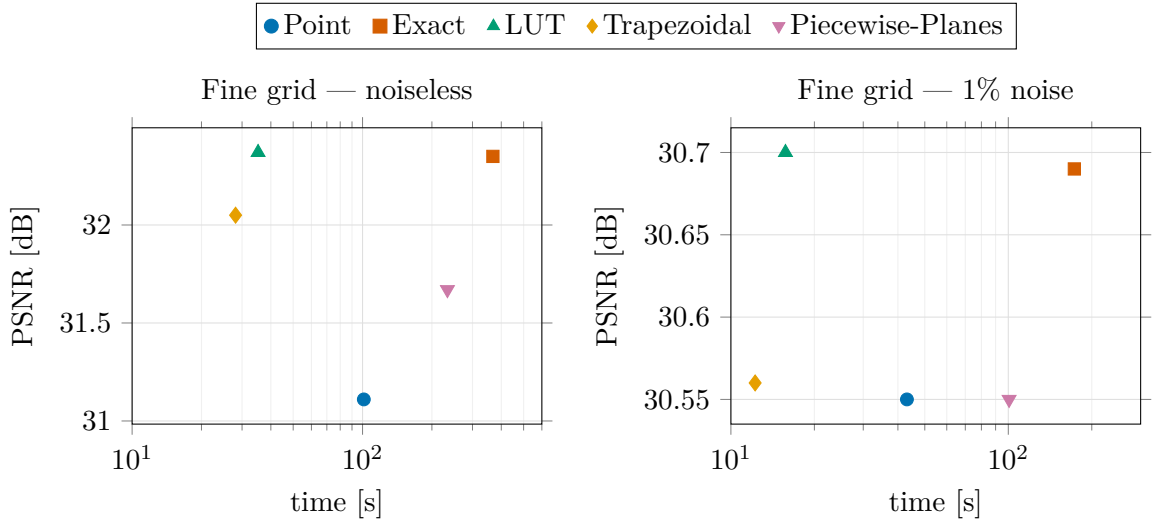

\subsection{Experiment 2: Mimicking a vascular imaging experiment with a linear ultrasound array}
\label{subsec:exp2_vessel}

The goal of this experiment is to assess the performance of the reconstructions based on the LUT operator in a realistic full-scale imaging scenario.
The reconstruction is compared to a modified back-projection reconstruction that has been tailored to the acquisition system \cite{linger2023volumetric}.

\subsubsection{Setup}
\label{subsec:exp2_setup}

\paragraph{Acquisition geometry}
We simulate the full translate-rotate acquisition of the reference system described in \Cref{subsec:lib_setup}.
The imaging domain $\Omega$ is a cube of side length $10\,\mathrm{mm}$.
The system consists of a single ultrasound probe of $64$ elements, which is rotated over $12$ angles and translated over $15$ locations \cite{linger2023volumetric}.
This results in $K = 11\,520$ cylindrical transducers described in \Cref{subsec:methodology}.
Their locations and orientations are illustrated in \Cref{fig:exp2_setup}.

The transducers are located approximately at $25\,\mathrm{mm}$ above the center of the reconstruction grid $\Omega$, such that their focal zone lies within $\Omega$.
The transducer positions were extracted from an experimental acquisition.
Unlike the full azimuthal sampling of Experiment 1, all probe positions are on the same side of $\Omega$, while the rotation angles span a range of $70^\circ$. This acquisition is therefore considered limited-view.

\begin{figure}[!htb]
	\centering
	\begin{subfigure}[b]{0.3\linewidth}
		\centering
		\includegraphics[width=\linewidth]{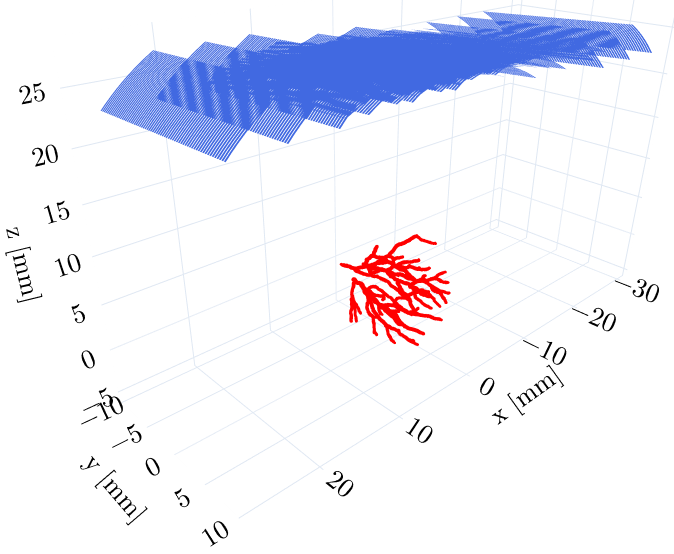}
		\caption{Acquisition geometry}
		\label{fig:exp2_setup1}
	\end{subfigure}
	\begin{subfigure}[b]{0.4\linewidth}
		\centering
		\includegraphics[width=\linewidth]{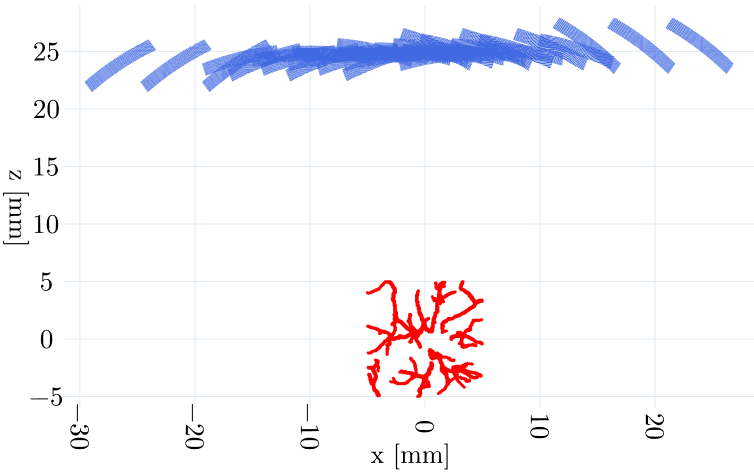}
		\caption{Acquisition geometry (front view)}
		\label{fig:exp2_setup2}
	\end{subfigure}
	\begin{subfigure}[b]{0.28\linewidth}
		\centering
		\includegraphics[width=\linewidth]{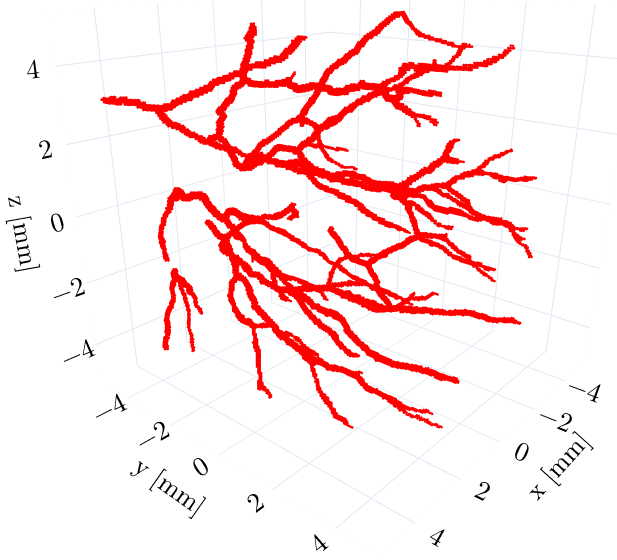}
		\caption{Phantom}
		\label{fig:exp2_phantom}
	\end{subfigure}
	\caption{Experiment 2 setup. (a) The $64$-element array surfaces (blue) at the $12 \times 15$ translate-rotate positions, $25\,\mathrm{mm}$ above the phantom (red). (b) Front view ($x$--$z$ plane): the cylindrically focused elements are oriented toward the imaged ROI and cover a limited angular sector above the phantom. (c) Phantom: vessel tree of a HiP-CT human kidney in a $10\,\mathrm{mm}$ cube.}
	\label{fig:exp2_setup}
\end{figure}

\paragraph{Phantom}
The domain $\Omega$ is discretized on a grid of $200^3$ voxels. This results in $N = 8$ million voxels of side length $50\,\mu\mathrm{m}$.
The domain is filled with a region extracted from the Hierarchical Phase-Contrast Tomography (HiP-CT) human kidney from the blood-vessel segmentation dataset \cite{blood-vessel-segmentation}.
The resulting phantom is illustrated in \Cref{fig:exp2_phantom}.

We selected a region in which the vessels are predominantly parallel to the $y$ axis.
This choice mitigates the poor reconstruction of structures oriented along the $z$ axis that can arise from the limited-view acquisition.
More specifically, the directions of those vessels are outside the angular sector covered by the detection geometry. The corresponding information is therefore lost in the measurements and cannot be used for the reconstruction.

The binary vessel labels are used as the initial pressure distribution, with the vessel diameter capped at three voxels ($\approx 150\,\mu\mathrm{m}$) to ensure that the photoacoustic signals fall within the EIR bandwidth.

\paragraph{Reference signals}
As in Experiment 1, reference measurements are synthesized by the Exact operator at temporal upsampling $U = 61$.
At this scale, the single application takes $\approx 2\,\mathrm{h}$ on one NVIDIA A100 GPU, and this figure already exploits the sparsity of the phantom by restricting the sum over voxels in \Cref{alg:surface_based} to the nonzero entries of $\bp_0$ ($0.30\%$ of the volume).
This is prohibitive inside an iterative solver that applies the operator tens of times, on iterates that are moreover dense.
Additive Gaussian noise with standard deviation equal to $1\%$ of the maximum signal amplitude is added to the measurements.

\paragraph{Operator and optimization parameters}
The reconstruction algorithm is run with the LUT operator with the same parameters as in Experiment 1 ($U = 11$, $N_\varsigma = N_\nu = 1000$ tables).
The resulting reconstruction will be compared to a back-projection (BP) algorithm specifically tailored for this imaging configuration, as detailed in \cite{linger2023volumetric}.
The BP is considered the standard non-iterative baseline.
The regularization parameter is $\lambda_R = 5 \times 10^{-3}$ and the solver uses the stopping tolerances defined in \Cref{subsec:methodology}.

\subsubsection{Results}
\label{subsec:exp2_result}
\Cref{fig:exp2_recon} displays side-by-side the Maximum Intensity Projection (MIP) images of the reconstructions obtained with the BP algorithm and the LUT operator, together with the ground truth phantom.
An interactive 3D rendering of the ground truth and both reconstructions is available at \url{https://trungthai1771-patminton-vessel-itk.static.hf.space}.

\paragraph{Qualitative analysis}

A qualitative analysis already reveals several artifacts that highlight the improved reconstruction quality of the model-based approach.

First, the BP reconstruction is blurred, with a visible loss of resolution around each vessel.
This may result from the low-pass filtering induced by the SIR, which is not accounted for by the BP algorithm.
In contrast, the vessels reconstructed with the model-based approach appear sharper.

Second, the model-based reconstruction is quantitative, meaning that its intensities reach the ground-truth level (maximum $\approx 1.1$ for a ground truth of $1$).
The BP amplitudes are not directly comparable since BP is designed to restore only relative, not absolute, amplitudes.

Third, BP fragments some vessels into disconnected blobs, whereas the model-based reconstruction recovers them as continuous structures. This can be observed, for example, for the vessels indicated by arrows labeled A in \Cref{fig:exp2_recon}.

Finally, the limited-view acquisition particularly affects vessels with a significant $z$ component. The model-based reconstruction partially mitigates this effect: the vessels indicated by arrows B, with a small $z$ component, remain visible, whereas those marked by arrows C, with a larger $z$ component, are only barely recovered. In the BP reconstruction, neither set of vessels is recovered.

\begin{figure}[!htb]
	\centering
	\setlength{\tabcolsep}{1.5pt}
	\newlength{\panelw}\setlength{\panelw}{0.29\linewidth}
	\newcommand{\rp}[1]{\includegraphics[width=\panelw]{figures/exp2/#1}}
	\newcommand{\cb}[1]{\includegraphics[height=\panelw]{figures/exp2/#1}}
	\newcommand{\rpgt}[3]{%
		\begin{tikzpicture}
			\node[anchor=south west, inner sep=0] at (0,0) {\rp{#1}};
			\begin{scope}[font=\scriptsize]
				\draw[black, -{Stealth[length=5.2pt]}, line width=1.6pt] (0.025\panelw,0.025\panelw) -- ++(0.1\panelw,0);
				\draw[black, -{Stealth[length=5.2pt]}, line width=1.6pt] (0.025\panelw,0.025\panelw) -- ++(0,0.1\panelw);
				\draw[cyan, -{Stealth[length=4pt]}, line width=0.7pt] (0.025\panelw,0.025\panelw) -- ++(0.1\panelw,0) node[anchor=west, inner sep=1pt] {\contour{black}{#2}};
				\draw[cyan, -{Stealth[length=4pt]}, line width=0.7pt] (0.025\panelw,0.025\panelw) -- ++(0,0.1\panelw) node[anchor=south, inner sep=1pt] {\contour{black}{#3}};
				\draw[black, line width=2.4pt] (0.79\panelw,0.02\panelw) -- ++(0.2\panelw,0);
				\draw[cyan, line width=1.2pt] (0.79\panelw,0.02\panelw) -- ++(0.2\panelw,0) node[midway, anchor=south, inner sep=1pt] {\color{cyan}\contour{black}{2\,mm}};
			\end{scope}
		\end{tikzpicture}}
	\newcommand{\rowlab}[1]{\raisebox{0.14\linewidth}{\rotatebox[origin=c]{90}{\small #1}}}
	\begin{tabular}{@{}c ccc c@{}}
		& \small MIP $z$ ($x$--$y$) & \small MIP $y$ ($x$--$z$) & \small MIP $x$ ($y$--$z$) & \\[2pt]
		\rowlab{Ground truth} & \rpgt{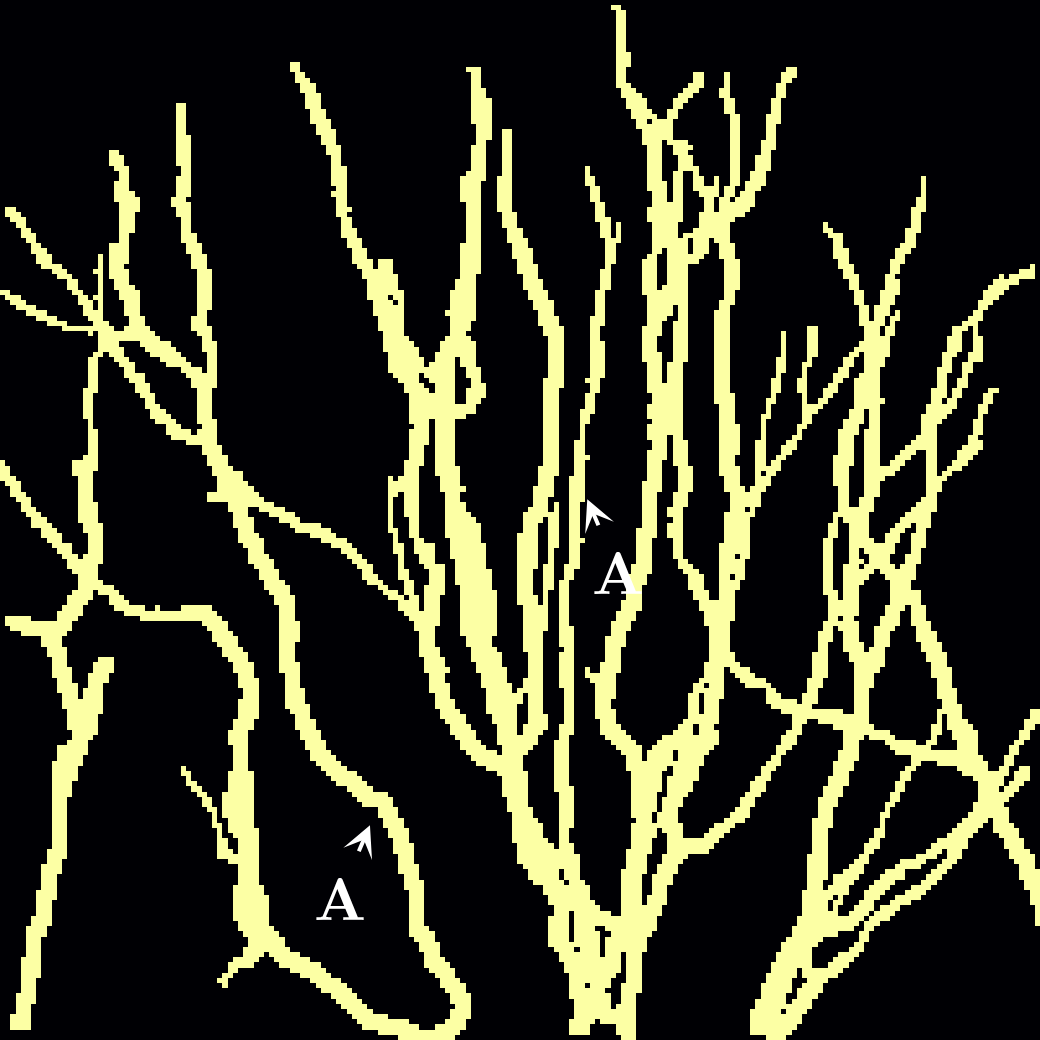}{$x$}{$y$} & \rpgt{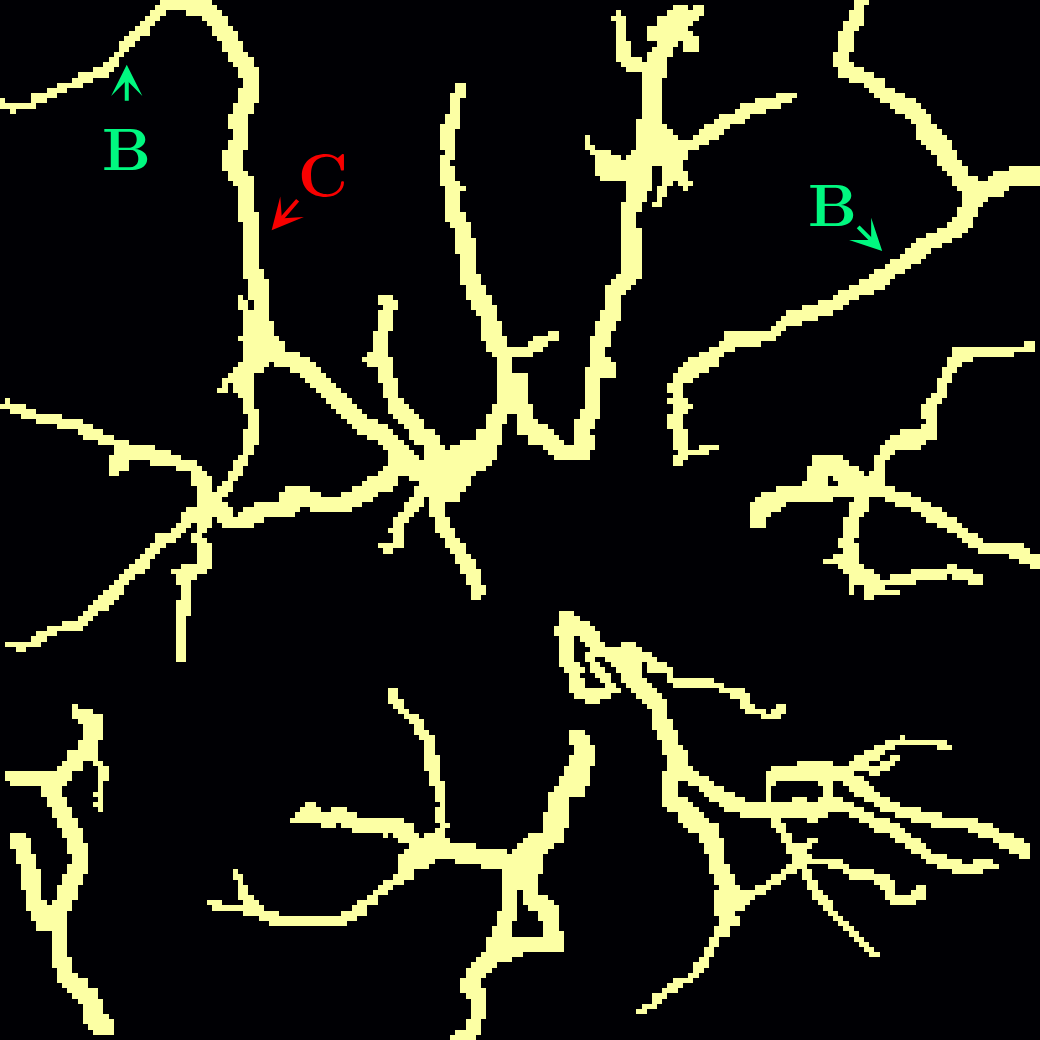}{$x$}{$z$} & \rpgt{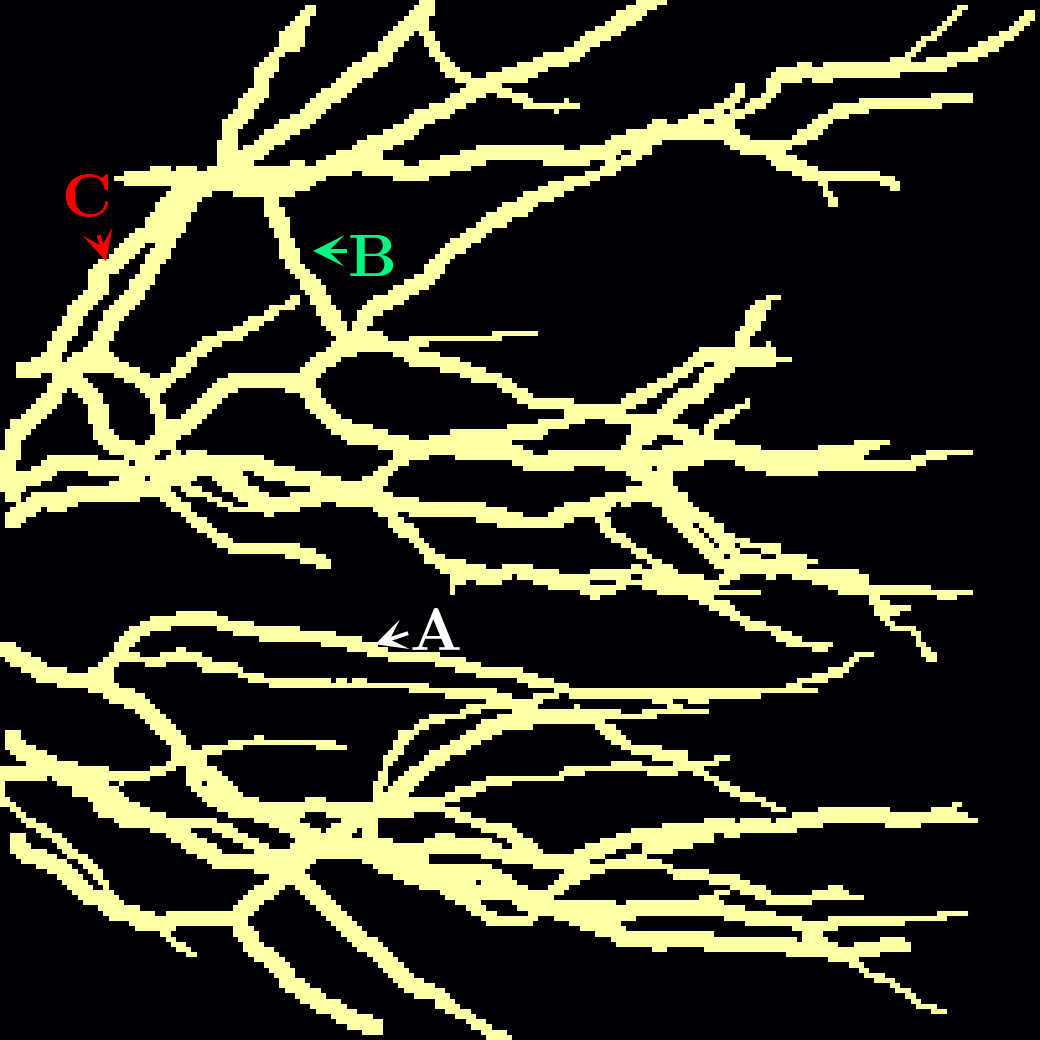}{$y$}{$z$} & \cb{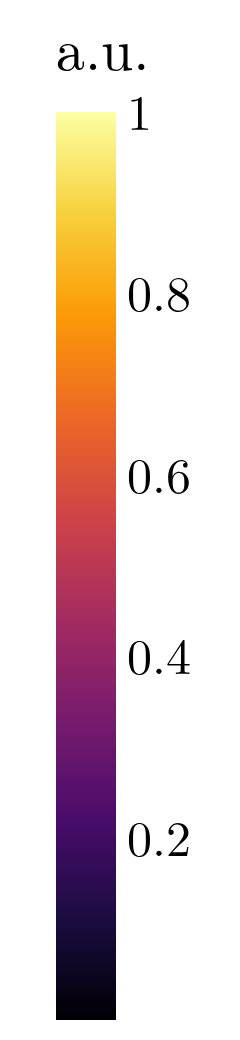} \\
		\rowlab{Model-based (LUT)} & \rp{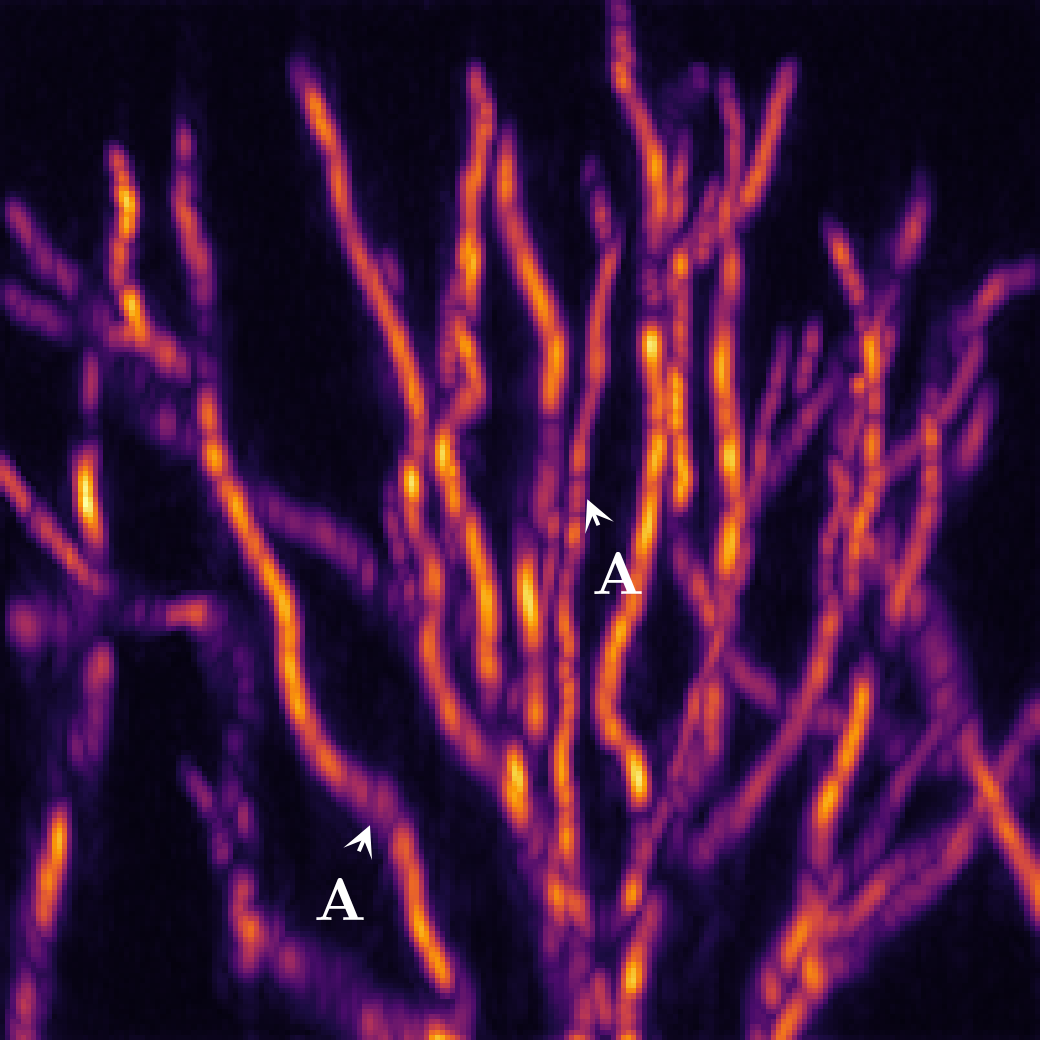} & \rp{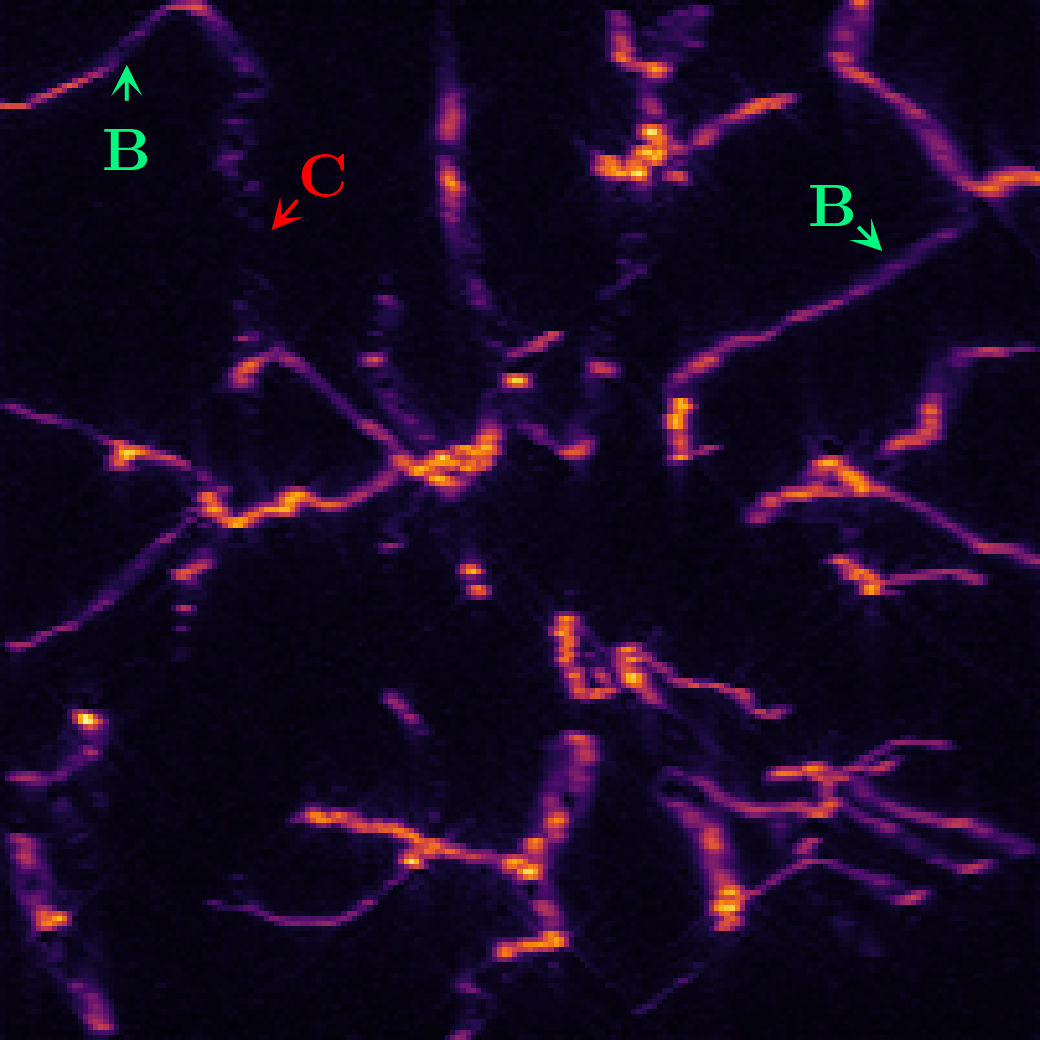} & \rp{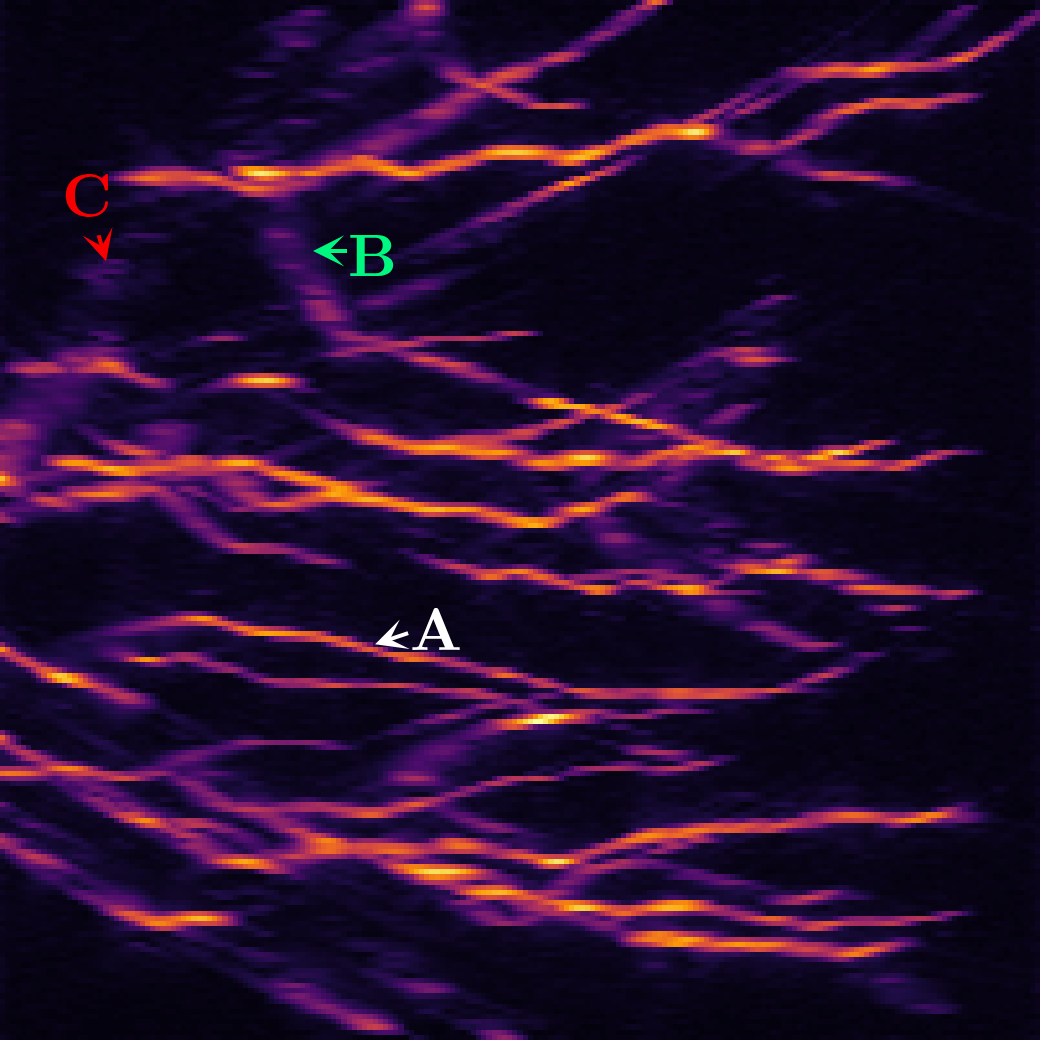} & \cb{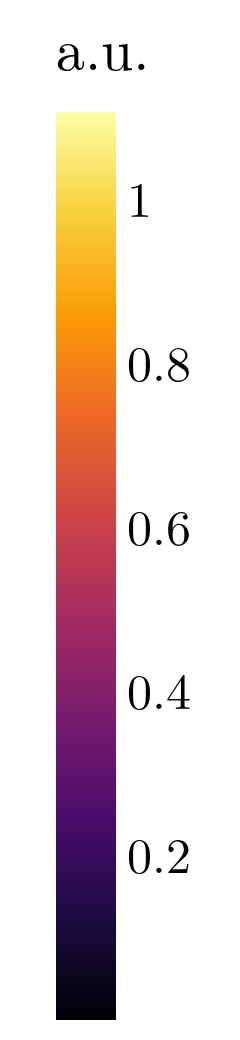} \\
		\rowlab{Back-projection (BP)} & \rp{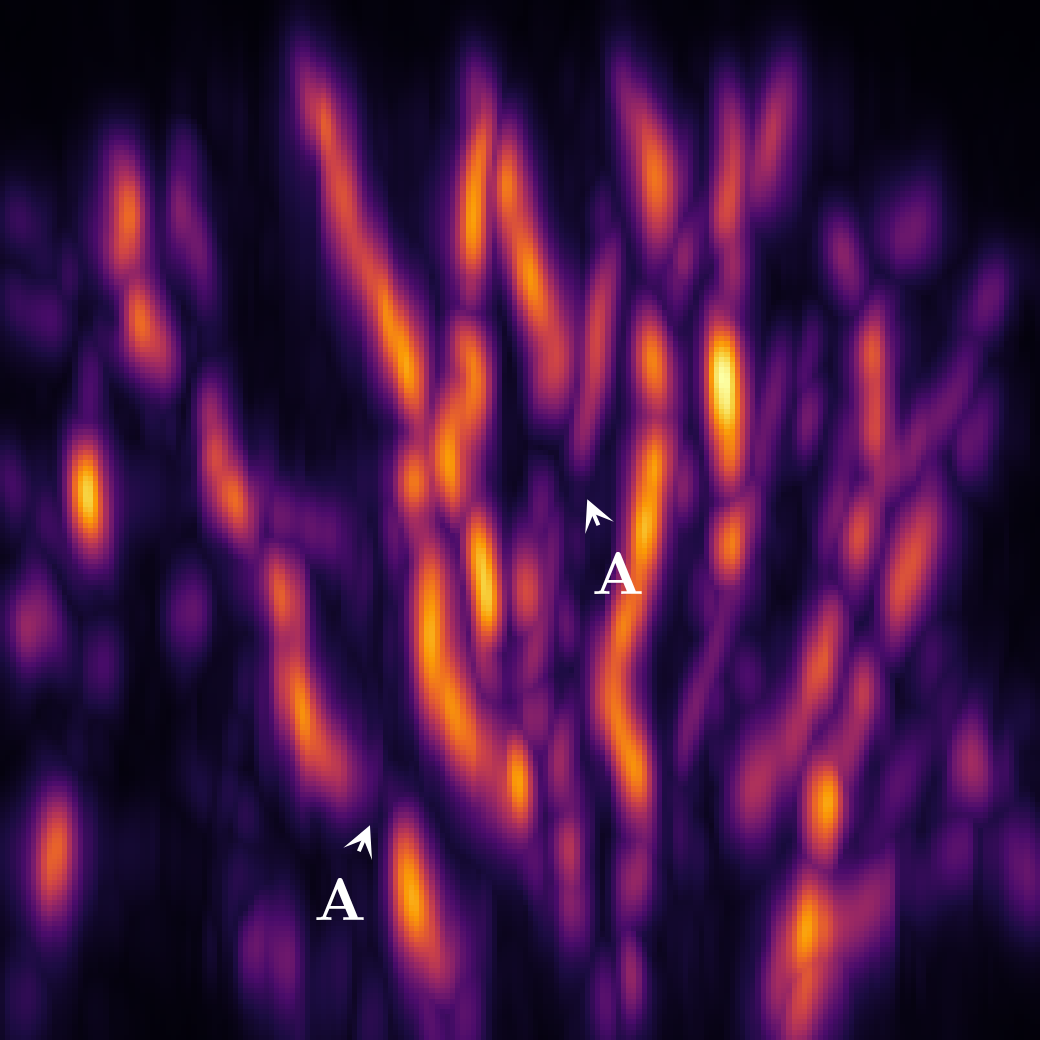} & \rp{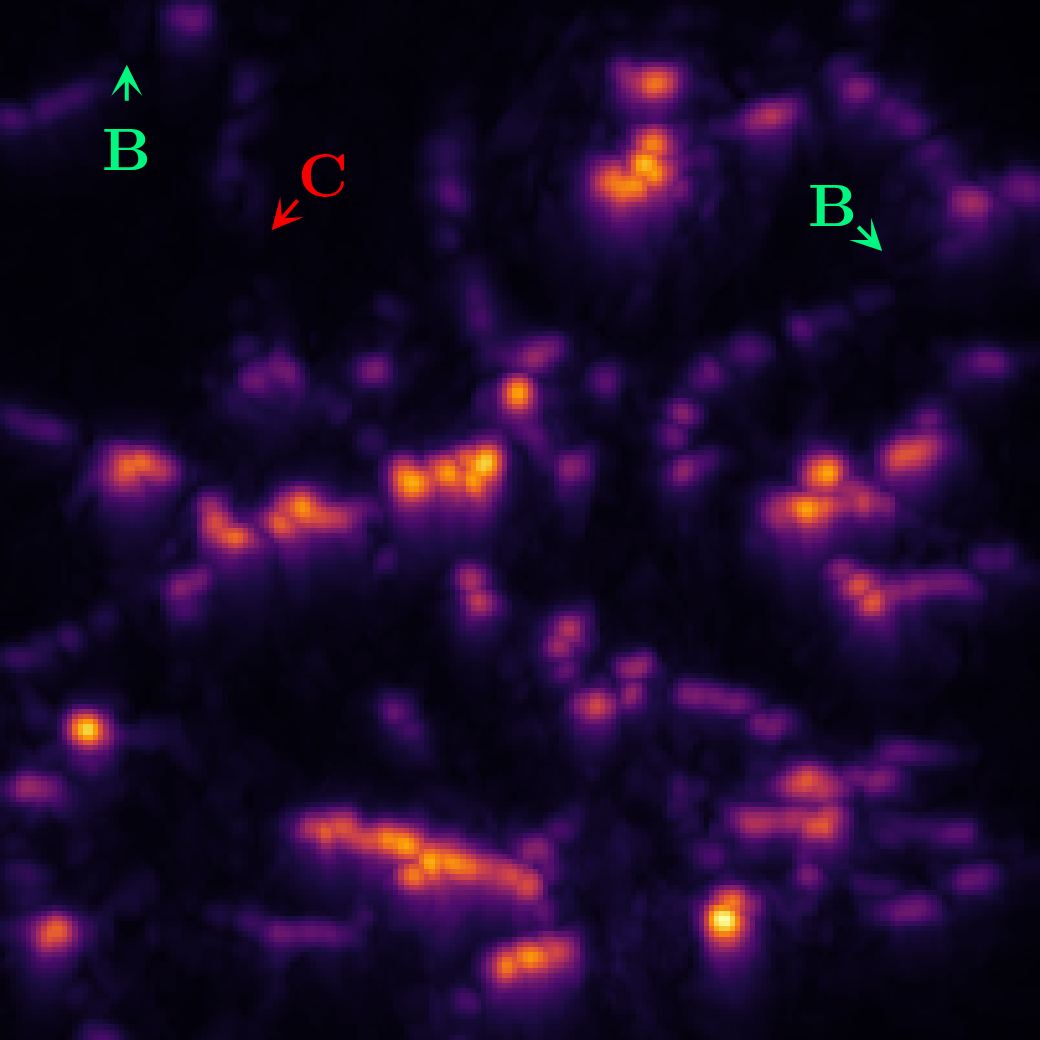} & \rp{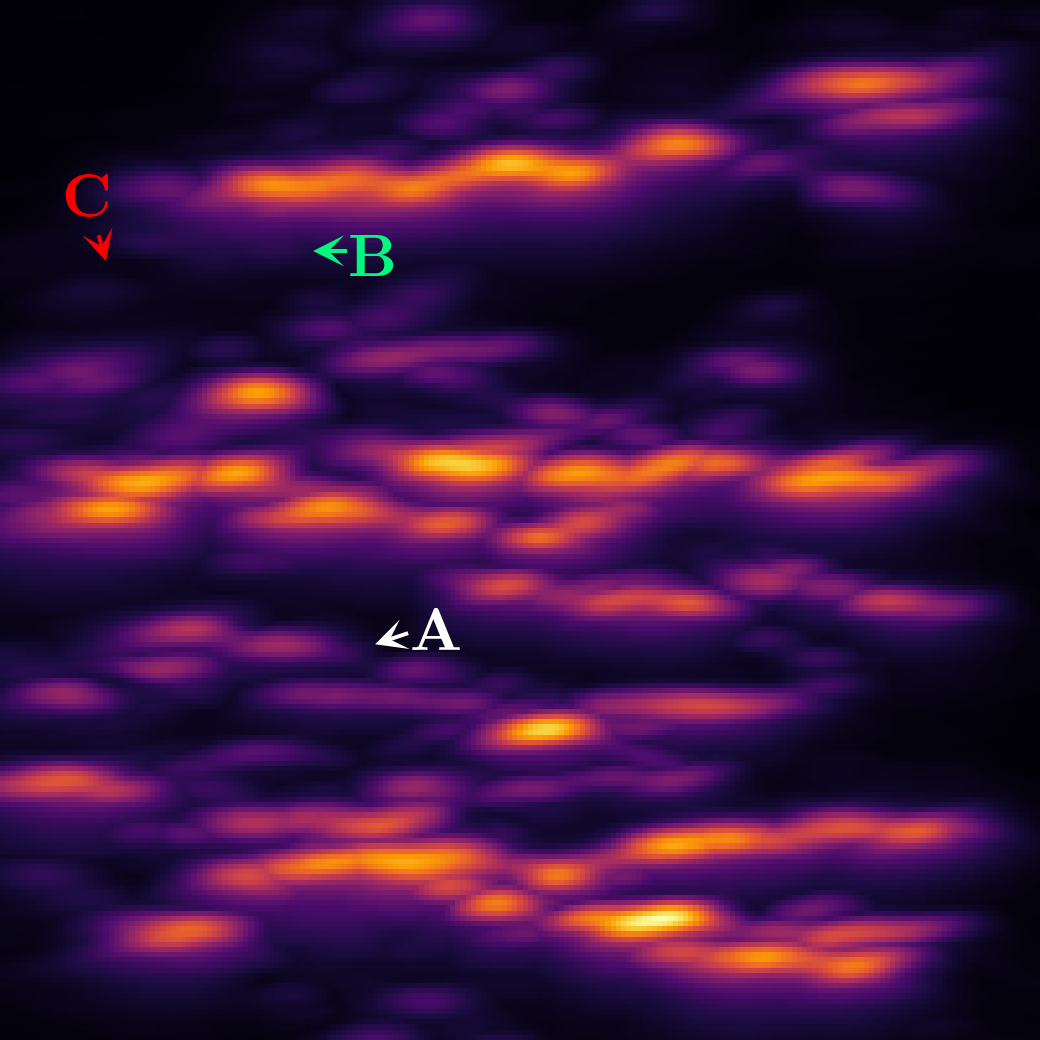} & \cb{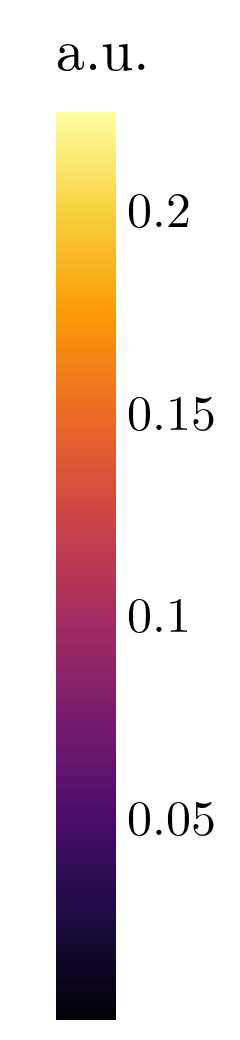} \\
	\end{tabular}
	\caption{Experiment 2 -- maximum-intensity projections (MIPs) along $z$, $y$ and $x$ (columns) of the ground truth, the model-based reconstruction with the LUT operator, and the back-projection reconstruction (BP) (rows). Arrows: (A) vessel fragmented into disconnected blobs by BP but recovered as a continuous structure by the model-based reconstruction; (B, C) vessels whose direction has a significant $z$ component, attenuated by the limited view: B, with a small $z$ component, remains partly visible in the model-based reconstruction, C, with a larger $z$ component, is barely recovered, and BP shows neither.}
	\label{fig:exp2_recon}
\end{figure}
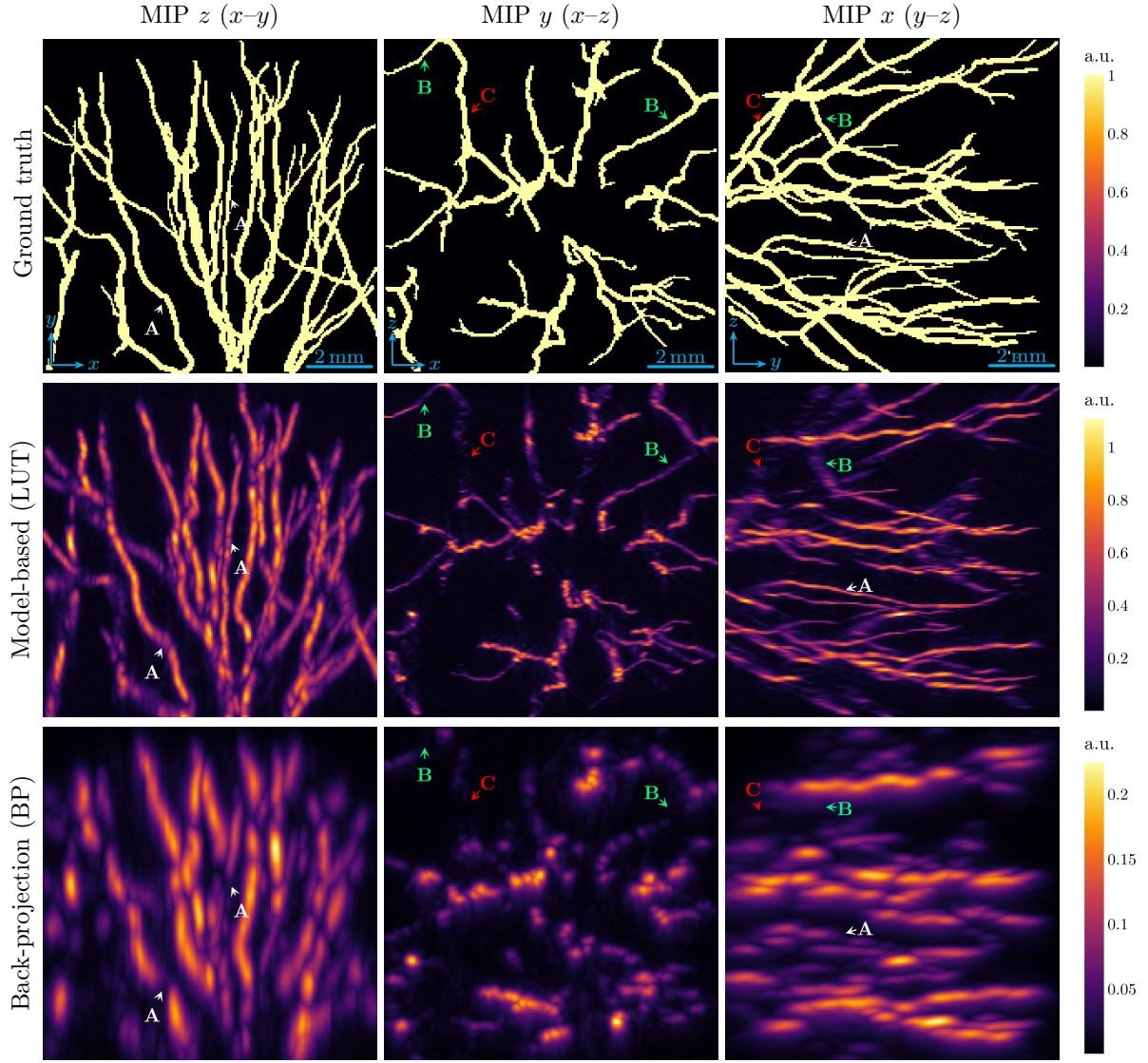

\paragraph{Quantitative analysis}

Quantitatively, since BP recovers only relative amplitudes, the image metrics are made scale-invariant by first applying a linear transformation that minimizes the least-squares error with respect to the ground truth.
All metrics for BP are then computed on the transformed reconstruction, while those for the model-based reconstruction are computed without any rescaling.

The model-based reconstruction achieves a higher PSNR than BP ($27.18$ vs. $25.46\, \mathrm{dB}$), whereas BP achieves a higher SSIM ($0.849$ vs. $0.657$).
This apparent discrepancy is explained by the extreme sparsity of the phantom, with vessels occupying only $0.30\%$ of the voxels.
Both the mean squared error underlying PSNR and the local statistics underlying SSIM are therefore dominated by the $99.7\%$ background voxels.
Consequently, a reconstruction can achieve a high score by accurately reproducing the background, even if its recovery of the vessels is poor.
This issue has also been highlighted in the photoacoustic imaging literature, where PSNR and SSIM are often evaluated on masks restricted to informative regions \cite{van2026assessing}.

To assess reconstruction quality around the vessels, we recompute both metrics on a dedicated mask.
It is defined as the union of the ground-truth support, which is binary, and the voxels where the reconstruction exceeds a threshold of $0.1$. It is then dilated by a ball of radius $3$ voxels.
Within this mask, the model-based reconstruction outperforms BP on both metrics, achieving a PSNR of $14.45$ vs. $10.65\, \mathrm{dB}$ ($+3.8\, \mathrm{dB}$) and an SSIM of $0.321$ vs. $0.029$. These masked metrics are consistent with the qualitative observations above.

Regarding computational time, the model-based reconstruction requires $15.8\,\mathrm{h}$ for $30$ iterations on a single NVIDIA A100 GPU, corresponding to approximately $32\,\mathrm{min}$ per iteration, compared with only $3\,\mathrm{s}$ for BP.

\paragraph{Resolution via Fourier shell correlation}

The qualitative observations above indicate that the model-based reconstruction produces sharper vessels than BP.
This is complemented, in this paragraph, with a quantitative resolution metric that provides information unavailable with the masked PSNR and SSIM.
In particular, we ask a natural question in PAT: what is the finest vessel size that can be resolved by a given (system, reconstruction method) pair?

We assess the resolution using the Fourier shell correlation ($\mathrm{FSC}$), the three-dimensional analogue of the Fourier ring correlation.
This metric was originally introduced for resolution assessment in cryo-EM for investigations in structural biology \cite{van2005fourier}.
The $\mathrm{FSC}$ between two volumes is a one-dimensional curve obtained by computing the normalized cross-correlation of their Fourier coefficients over spherical shells of constant spatial frequency $\zeta$.
At a given frequency, $\mathrm{FSC}(\zeta)=1$ indicates identical Fourier content, while values approaching zero indicate that the two shells become uncorrelated.

\Cref{fig:exp2_fsc} shows the $\mathrm{FSC}$ curves between the ground truth, as a reference image, and the BP and model-based reconstructions respectively.
The spatial frequency $\zeta$ can be converted to an acoustic frequency $F$ through $F = c\,\zeta/w$ with $w = 50\,\mu\mathrm{m}$ the voxel size, $c$ the speed of sound.
Therefore, the Nyquist frequency of the grid $\zeta = 0.5$ corresponds to $15\,\mathrm{MHz}$ which is above the $10\,\mathrm{MHz}$ upper cut-off of the EIR.
The resolution is consequently limited by the measurement and not by the grid.
Both $\mathrm{FSC}$ curves are globally decreasing which shows that lower frequencies are better reconstructed than higher ones.
The model-based FSC decays substantially more slowly than that of BP. This indicates better preservation of high-frequency information.

The standard practice \cite{van2005fourier} is then to set the achievable resolution from the highest spatial frequency at which the FSC remains above a prescribed threshold.
Above this frequency, the reconstructed content is considered insufficiently correlated with the ground truth or dominated by noise.
Several thresholds exist, but we consider the $1$-bit threshold, see \Cref{fig:exp2_fsc}.
The model-based reconstruction remains above the threshold up to $\zeta \approx 0.34$, whereas BP falls below it at only $\zeta \approx 0.085$.
This leads to a resolution gain of a factor of about $4$ for the model-based over the BP reconstruction.
In acoustic terms these frequencies correspond to $\approx 10\,\mathrm{MHz}$ for the model-based reconstruction and $\approx 2.6\,\mathrm{MHz}$ for BP.

To conclude, the model-based reconstruction remains correlated with the ground truth down to a spatial period of approximately $150\,\mu\mathrm{m}$, that is three voxels, the finest scale present in the phantom.
In contrast, BP recovers only the low-frequency content of the vessels, consistent with the reduced sharpness observed in \Cref{fig:exp2_recon}.

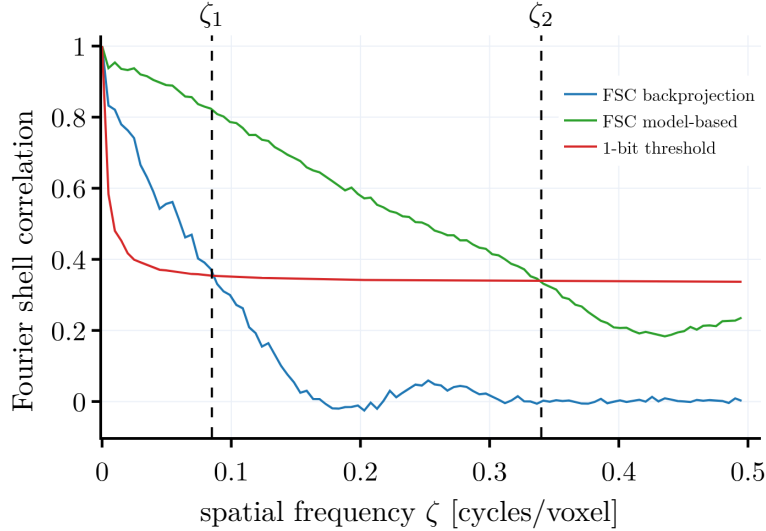
\begin{figure}[!htb]
       \centering
       \begin{tikzpicture}
              \node[anchor=south west, inner sep=0] (img) at (0,0) {\includegraphics[width=0.6\linewidth]{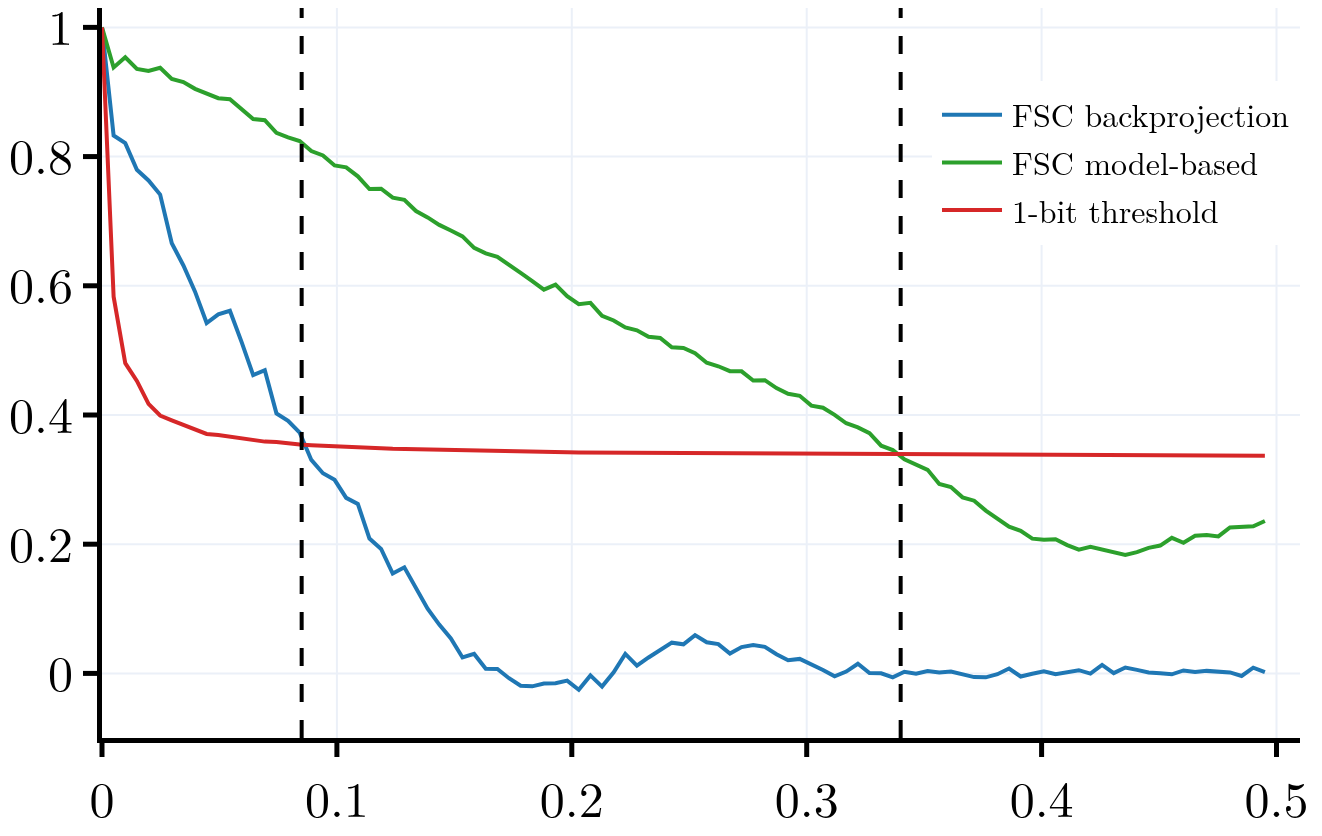}};
              \begin{scope}[x={(img.south east)}, y={(img.north west)}]
                     \node[anchor=south, inner sep=1pt] at (0.2277,1) {$\zeta_1$};
                     \node[anchor=south, inner sep=1pt] at (0.6808,1) {$\zeta_2$};
              \end{scope}
              \node[rotate=90, anchor=south, inner sep=5pt] at (img.west) {Fourier shell correlation};
              \node[anchor=north, inner sep=3pt] at (img.south) {spatial frequency $\zeta$ [cycles/voxel]};
       \end{tikzpicture}
       \caption{Experiment 2 -- Fourier shell correlation (FSC) of the two reconstructions against the ground truth, as a function of normalized spatial frequency $\zeta$ (in cycles per voxel; $\zeta = 0.5$ is the grid Nyquist, $15\,\mathrm{MHz}$; $\zeta_1 = 0.085$, $\zeta_2 = 0.34$).}
       \label{fig:exp2_fsc}
\end{figure}

Overall, the model-based reconstruction provides a significant improvement in reconstruction quality by explicitly accounting for the detector SIR.
This improvement comes at the cost of substantially higher computational requirements.
To the best of our knowledge, this is the first model-based reconstruction demonstrated on such a large-scale system, enabled by the efficient GPU implementations developed in \Cref{sec:implementations,sec:usual_surfaces}.

\section{Conclusion}
\label{sec:conclusion}

This paper provides fast, matrix-free implementations of the forward operator $\bA$ and its adjoint $\bA^*$ for photoacoustic tomography, accounting for the spatial impulse response (SIR) of transducers, at a memory cost that scales to full-resolution 3D systems.

We derived two such implementations, both based on a Kansa-type radial discretization of the initial pressure $p_0$ and a decoupling of interpolation and propagation that is valid for any compactly supported radial basis function.
The two methods differ only in their treatment of the surface integral.
The point method replaces the surface by a quadrature rule with $Q$ points, as commonly done in the photoacoustic literature, whereas the surface method treats the surface analytically. The latter was instantiated for planar and cylindrical elements, together with a methodology for extending it to other surface geometries.

On the theoretical side, both implementations are analyzed within a common mathematical framework. We establish error bounds with respect to the continuous wave-equation model, together with the computational complexity of each implementation, therefore making the trade-off between accuracy and computational cost explicit.

On the numerical side, the analytical treatment of the transducer surface outperforms the point-based quadrature approach in both accuracy and computational speed.
On the large-scale reference system, the surface method enables model-based reconstruction and outperforms the back-projection method currently used in practice, which is optimized for the same system.
To the best of our knowledge, this is the first model-based reconstruction demonstrated at this scale that takes into account the SIR of the transducers.

The surface of practical transducers is often considered a computational obstacle that requires fine discretization. The proposed analytical surface method avoids this discretization while providing improved accuracy and computational efficiency. The released implementation in \url{https://patminton.readthedocs.io/en/latest/} brings model-based three-dimensional photoacoustic reconstruction at the scale of a real acquisition system within the capabilities of a single GPU.

\ack{We sincerely thank Emmanuel Soubies for introducing us to the Fourier shell correlation method as a metric for image reconstruction.}

\funding{This project has received financial support from the CNRS through the MITI interdisciplinary programs.
This work was granted access to the HPC resources of IDRIS under the allocation 2025-AD011014692R2 made by GENCI.
This project received funding from the PHC Merlion, the Ministry for Europe and Foreign Affairs, the Ministry of Higher Education, Research and Space, and the Agency for Science, Technology and Research (A*STAR).
}

\data{No new data were created or analyzed in this study. The code used to produce the numerical results reported in this paper is available at \url{https://patminton.readthedocs.io/en/latest/}}

\appendix
\section{Table of main notations.}
\label{sec:notations_table}

{\footnotesize
\setlength{\LTpre}{\medskipamount}
\renewcommand{\arraystretch}{1.15}
\newcommand{\ntgroup}[1]{\multicolumn{2}{@{}l}{\emph{#1}}\\[2pt]}

\begin{longtable}{@{}l@{\hspace{1.5em}}p{0.75\textwidth}@{}}
\caption{Main notations used throughout the paper.}
\label{tab:notations}\\
\toprule
Symbol & Definition \\
\midrule
\endfirsthead
\multicolumn{2}{@{}l}{\itshape \Cref{tab:notations}, continued}\\
\toprule
Symbol & Definition \\
\midrule
\endhead
\midrule
\multicolumn{2}{r@{}}{\itshape continued on next page}\\
\endfoot
\bottomrule
\endlastfoot

\ntgroup{General conventions}
$[n]$                                   & the set $\{1, \ldots, n\}$, also the flattened index set of a grid of $n$ points \\
$\bv[i]$, $\bm{M}[i,j]$                 & $i$-th entry of a vector, $(i,j)$-th entry of a matrix \\
$\bs[k,l]$, $\bbf_k$                    & an index or subscript $k$ refers to the transducer $\Sigma_k$ \\
$\bx$, $\bx'$                           & points of $\R^3$; $\bx'$ is the source when $\bx$ is the detection point \\
$\|\cdot\|_p$, $\|\cdot\|$              & $\ell_p$ norm, Euclidean norm $\|\cdot\| = \|\cdot\|_2$ \\
$\|\cdot\|_{L^q}$, $\|\cdot\|_{W^{k,q}}$ & Lebesgue and Sobolev norms (\Cref{sec:notations}) \\
$|\cdot|$                               & Lebesgue measure of a set, absolute value of a number \\
$a \lesssim b$                          & $a \leq C b$ with $C > 0$ independent of the parameters \\
$C(\beta_1,\beta_2)$                    & constant depending only on $\beta_1$ and $\beta_2$ \\
$\star$                                 & convolution in time, continuous \eqref{eq:measurement} or discrete \eqref{eq:discrete_conv_2} \\
$\indic_A$, $\delta_a$                  & indicator function of the set $A$, Dirac mass at $a$ \\
$\widetilde{\cdot}$                     & approximation by quadrature on $\Sigma$, e.g.\ $\widetilde{m}$, $\widetilde{f}_\Sigma$ \\
$\overline{\cdot}$                      & exact average of $\,\widetilde{\cdot}\,$ over the cells of the upsampled time grid \\
$\widehat{\cdot}$                       & computed approximation \\
\midrule

\ntgroup{Acquisition model (\Cref{sec:forward_model})}
$\Omega \subset \R^3$                   & sample region, containing the support of $p_0$ \\
$c$                                     & speed of sound, assumed constant \\
$p_0$                                   & initial pressure, the quantity to reconstruct \\
$p(\bx,t)$                              & acoustic pressure, solution of \eqref{eq:wave_equation} \\
$S(\bx,r)$                              & sphere of center $\bx$ and radius $r$ \\
$K$, $\Sigma_k$                         & number of transducers, surface of the $k$-th transducer \\
$\Sigma$, $|\Sigma|$                    & generic transducer (index $k$ dropped), its area \\
$m$                                     & pressure integrated over the transducer surface \eqref{eq:measurement} \\
$e$, $T_e$                              & electrical impulse response (EIR), length of its support \\
$s$                                     & signal measured by a transducer, $s = e \star m$ \eqref{eq:measurement} \\
$F_s$, $t_0$, $t_l$                     & sampling frequency, window start, sampling instants $t_l = t_0 + l/F_s$ \\
$L$                                     & number of time samples per transducer \\
$\mathcal{S}$                           & sampling operator \eqref{eq:def_operatorS} \\
$\bs$                                   & discrete measurements in $\R^{KL}$, or $\R^{L}$ for one transducer \eqref{eq:measurement_discrete} \\
$\bs^{U}$                               & coefficients of the measured signal on the upsampled grid, $\dsop \bs^{U} = \bs$ (\Cref{sec:discrete_conv}) \\
\midrule

\ntgroup{Discretization (\Cref{subsec:discretization})}
$N$, $\bx_i$                            & number of grid points, grid points covering $\Omega$ \\
$\phi$, $\kappa$                        & radial function, supported in $[0,\kappa]$ \\
$\bp_0 \in \R^N$                        & coefficients of $p_0$, unknown of the inverse problem \eqref{eq:continuous_discretized_field} \\
$\dist_{\min}$                          & minimal grid-to-transducer distance \eqref{eq:min_dist} \\
\midrule

\ntgroup{Model and convolutional structure (\Cref{subsec:overview})}
$g$                                     & kernel depending only on $\phi$, $g(t) = -\frac{1}{2} c t\, \phi(|ct|)$ (\Cref{prop:point_value_pressure_field}) \\
$h$                                     & system kernel $h = e \star g$ \eqref{eq:def_h}, assumed in $W^{2,\infty}(\R)$ (see \Cref{prop:approximation_points}) \\
$f(\bx,\cdot)$, $f_\Sigma$              & arrival-time distributions, $p(\bx,\cdot) = g \star f(\bx,\cdot)$ and $m = g \star f_\Sigma$ \eqref{eq:def_f_sigma} \\
$U$, $\tau_l$                           & upsampling factor, upsampled grid $\tau_l = t_0 + l/(UF_s)$ with midpoints $\tau_{l-\frac{1}{2}}$ \\
$\widehat{f}_\Sigma$, $\bbf$            & piecewise constant approximation of $f_\Sigma$, its coefficients in $\R^{UL}$ \eqref{eq:approx_u_sigma} \\
$\widehat{m}$                           & computed signal $g \star \widehat{f}_\Sigma$ \\
$\bh$                                   & discretization of $h$ on the upsampled grid \eqref{eq:def_bh} \\
$\dsop$                                 & downsampling by the factor $U$ \eqref{eq:def_downsampling} \\
$\widehat{\bs}$                         & computed measurements $\dsop(\bh \star \bbf)$ \eqref{eq:discrete_conv} \\
$\bA$, $\bA^{*}$                        & forward operator $\bp_0 \mapsto \widehat{\bs}$, its adjoint \\
\midrule

\ntgroup{Point detector approximation (\Cref{subsec:point_method})}
$Q$                                     & number of quadrature points on $\Sigma$ \\
$(\bq_j, \Delta \bq_j)_{j \in [Q]}$     & quadrature points and areas \\
$\widetilde{m}$, $\widetilde{f}_\Sigma$ & quadrature approximations of $m$ and $f_\Sigma$ \eqref{eq:def_f_point} \\
\midrule

\ntgroup{Surface method (\Cref{subsec:surface_method})}
$\mathcal{M}_i^l$                       & part of $\Sigma$ reached by the wave from $\bx_i$ during $[\tau_{l-1},\tau_l)$ \eqref{eq:def:bffi} \\
$\overline{\bbf}$, $\overline{\bbf}_i$  & exact cell averages of $f_\Sigma$ and part due to $\bx_i$, $\overline{\bbf} = \sum_i \bp_0[i]\, \overline{\bbf}_i$ \\
$\bbf_i$                                & midpoint approximation $\bbf_i[l] = U F_s \left| \mathcal{M}_i^l \right| / (c\tau_{l-\frac{1}{2}})$ of $\overline{\bbf}_i[l]$ \eqref{eq:ui_area} \\
$r_{\min}(i)$, $r_{\max}(i)$            & minimal and maximal distances from $\bx_i$ to $\Sigma$ \eqref{eq:rmin_rmax} \\
$l_{\min}(i)$, $l_{\max}(i)$            & first and last time intervals where $\mathcal{M}_i^l$ can be nonempty \eqref{eq:rmin_rmax} \\
$\overline{B}$, $C_{\mathcal{M}}$       & mean number of time intervals per voxel, cost of one area $\left| \mathcal{M}_i^l \right|$ \\
$\Delta l$                              & step of the coarser area grid (\Cref{rmk:coarser_area_grid}) \\
\midrule

\ntgroup{Surface parametrization (\Cref{sec:usual_surfaces})}
$\sigma$, $D$                           & parametrization of $\Sigma$, its domain; recentered on $\bx_i$ in \Cref{sec:usual_surfaces} \\
$(u,v)$                                 & coordinates in $D$ \\
$\psi$, $J_0$                           & separable form $\|\sigma(u,v)\|^2 = \psi(u) + v^2$, $\dint\sigma = J_0 \dint u \dint v$ \eqref{eq:separable} \\
$D_{\xi_1,\xi_2}$                       & quadrant $[u_{\min}, u_{\max}] \times [v_{\min}, v_{\max}]$ of $D$, reflected to $u, v \geq 0$ \\
$D^l$, $D^l_{\xi_1,\xi_2}$              & preimage of $\mathcal{M}^l$ in $D$, its part in the quadrant $D_{\xi_1,\xi_2}$ \eqref{eq:D_l_xi} \\
$\Psi$, $\Phi$                          & solutions in $v$ and in $u$ of $\|\sigma(u,v)\| = r$; $\widehat{\Psi}$, $\widehat{\Phi}$ clamped to the quadrant \\
$\alpha_l$, $\beta_l$                   & abscissas where $\|\sigma\| = c\tau_l$ meets $v = v_{\max}$ and $v = v_{\min}$ \eqref{eq:alphabeta} \\
$I_l$                                   & arc integral $\int_{\alpha_l}^{\beta_l} \Psi(u, c\tau_l) \dint u$ \eqref{eq:area_D_l_xi} \\
\midrule

\ntgroup{Transducer geometries (\Cref{subsec:plane_surface,subsec:cylinder_method})}
$\bc_0$, $\be_1$, $\be_2$               & center and orthonormal frame of a transducer \\
$L_1$, $L_2$                            & half-width and half-height of a planar transducer \eqref{eq:plane_param} \\
$R$, $\theta_{\max}$, $L_z$             & radius, angular half-aperture, axial half-height of a cylinder \eqref{eq:cylinder_param} \\
$\rho(r)$                               & plane: $\rho(r) = \sqrt{r^2 - \bx_i[3]^2}$ \eqref{eq:plane_primitive} \\
$\rho$, $\theta_i$                      & cylinder: polar coordinates of $\bx_i$ \\
$E(\vartheta,\nu)$, $F(\vartheta,\nu)$  & incomplete elliptic integrals of the second and first kind \\
$b(r)$, $\nu(r)$                        & coefficient and parameter of the elliptic form of $I_l$ \eqref{eq:cyl_Il} \\
$N_\varsigma$, $N_\nu$                  & numbers of samples along the two axes of the lookup tables \eqref{eq:LUT_sampling} \\
$N_p$                                   & number of planes of the Piecewise-Planes operator \eqref{eq:plane_decomp_u} \\
\midrule

\ntgroup{Numerical experiments (\Cref{sec:experiments})}
$\bp_0^{\star}$, $\widehat{\bp}_0$      & ground-truth phantom, its reconstruction \eqref{eq:nnls_tik} \\
$\lambda_R$                             & regularization parameter \eqref{eq:nnls_tik} \\
$F_c$, $\lambda_c$                      & center frequency, acoustic wavelength $\lambda_c = c/F_c$
\end{longtable}
}

\section{Proof of \Cref{prop:point_value_pressure_field}} \label{sec:prop1}

Let $\bx \in \Sigma_k$ for some $k \in [K]$ and $t > 0$, substituting \eqref{eq:continuous_discretized_field} in \eqref{eq:forward} gives
\begin{equation*}
       p(\bx,t) = \sum_{i \in [N]} \bp_0[i] \frac{1}{4\pi c} \partial_t a(i,\bx,ct) \quad \text{with} \quad
       a(i, \bx, r) = \frac{1}{r} \int_{S(\bx,r)} \phi( \| \bx' - \bx_i\| ) \dint S_r(\bx')
\end{equation*}
The remainder of the proof consists in expressing the value of each $a(i,\bx,r)$. It will be computed for an arbitrary $i$ and the same result will hold for all $i$.
To simplify notation, we will abbreviate $a(i, \bx, r)$ as
\begin{equation*}
       a(r) = \frac{1}{r} I_S(r) \quad \text{with} \quad  I_S(r) = \int_{S} \phi(\| \bx' - \bx_i \|) \dint S_r(\bx'),
\end{equation*}
where $S = S(\bx,r)$ with surface element $\dint S_r$. We furthermore let $d = \| \bx - \bx_i\|$.

\begin{figure}[!htb] \centering
\begin{tikzpicture}[scale=1.5, >=Stealth]

  \def\ax{213.69}
  \def\rr{3.1}
  \def\th{21}
  \def\kap{1.65}
  \def\spread{32}

  \coordinate (xi) at (1,2);
  \coordinate (X)  at (4,4);
  \coordinate (y)  at ($(X)+({\ax+\th}:\rr)$);
  \coordinate (ym) at ($(X)+({\ax-\th}:\rr)$);
  \coordinate (C)  at ($(X)+(\ax:{\rr*cos(\th)})$);
  \coordinate (p0) at ($(X)+(\ax:\rr)$);

  \draw[step=1cm, gray!55, dashed] (0,0) grid (2,2);
  \foreach \gx in {0,1,2}{\foreach \gy in {0,1,2}{\fill[black] (\gx,\gy) circle (1pt);}}

  \fill[ForestGreen, opacity=0.10] (xi) circle (\kap);
  \draw[ForestGreen, opacity=0.55] (xi) circle (\kap);

  \draw[red, thick] ($(X)+({\ax-\spread}:\rr)$) arc ({\ax-\spread}:{\ax+\spread}:\rr);

  \draw[dashed] (xi) -- (X)
        node[pos=0.42, inner sep=1.5pt, below=1pt] {$d$};
  \draw (X) -- (y)
        node[pos=0.68, inner sep=1.5pt, below=1pt] {$r$};
  \draw[blue, thick] (xi) -- (y)
        node[pos=0.52, below=1pt, blue, inner sep=1.5pt] {$\lambda(\theta)$};

  \draw[thin] ($(X)+(\ax:0.95)$) arc (\ax:{\ax+\th}:0.95);
  \node[font=\small] at ($(X)+({\ax+\th/2}:1.22)$) {$\theta$};

  \fill       (xi) circle (1.3pt); \node[anchor=south east]      at (xi) {$\bx_i$};
  \fill[red]  (X)  circle (1.3pt); \node[anchor=south west, red] at (X)  {$\bx$};
  \fill[blue] (y)  circle (1.3pt); \node[anchor=west, blue]      at ($(y)+(0.09,-0.03)$) {$\bx'$};

  \node[anchor=north west, red] at ($(X)+({\ax+\spread}:\rr)$) {$S(\bx,r)$};
  \node[anchor=north, ForestGreen, font=\small] at (1,-0.15)
        {$\supp\phi(\|\cdot-\bx_i\|)$};

\end{tikzpicture}
\caption{Illustration of the geometrical objects at stake in the integration.} \label{fig:change_variable}
\end{figure}
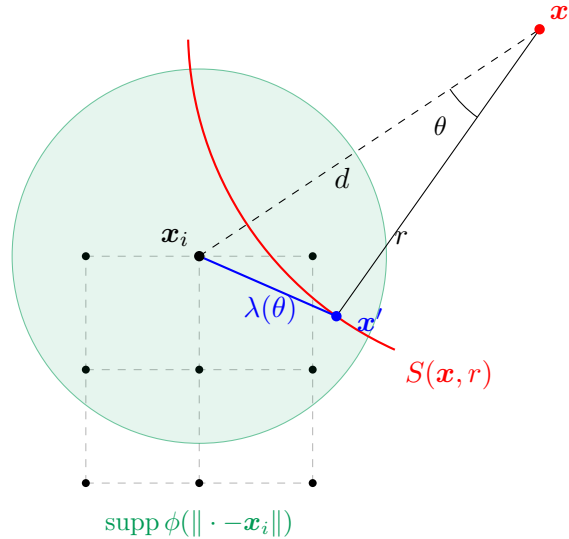

\paragraph{Reduction to a 1d integral} Let $r > 0$, the sphere $S(\bx, r)$ is parametrized by spherical coordinates of polar axis $d^{-1} (\bx_i - \bx)$.

A point $\bx'(\omega, \theta) \in S(\bx,r)$ is determined by the azimuth
$\omega \in [0,2\pi)$ and the polar angle $\theta \in [0,\pi]$ between $\bx' - \bx$
and the axis $\bx_i - \bx$. The surface element reads
$\dint S_r = r^2 \sin \theta \dint \theta \dint \omega$.
Let $\lambda(\theta) = \| \bx'(\omega, \theta) - \bx_i \|$. By the law of cosines,
this distance is independent of $\omega$ and is defined as
\begin{equation*}
    \lambda(\theta) = \sqrt{r^2 + d^2 - 2rd\cos(\theta)}.
\end{equation*}
\Cref{fig:change_variable} illustrates this geometry.

The integration therefore leads to
\begin{equation*}
    \begin{aligned}
       I_S(r) &= 2\pi r^2 \int_{0}^\pi \phi(\lambda(\theta)) \sin \theta \dint \theta.
    \end{aligned}
\end{equation*}

The map $\lambda$ defines a change of variable from $[0,\pi]$ onto $[|r-d|, r+d]$
since it is $C^1$ and strictly increasing on $(0, \pi)$. Consequently,
$r^2 \sin(\theta) \dint \theta = \frac{r}{d} \lambda(\theta) \lambda'(\theta) \dint \theta$.
Hence,
\begin{equation*}
    \begin{aligned}
       I_S(r) &= \frac{2\pi r}{d} \int_{|r-d|}^{r+d} \lambda \phi(\lambda) \dint \lambda.
    \end{aligned}
\end{equation*}
Thus,
\begin{equation*}
    \begin{aligned}
       a(r) &= \frac{2\pi}{d} \int_{|r-d|}^{r+d} \lambda \phi(\lambda) \dint \lambda.
    \end{aligned}
\end{equation*}

\paragraph{Differentiating $a$.}
The kernel $g(t') = -\tfrac12 ct' \phi(|ct'|)$ of \Cref{prop:point_value_pressure_field} is odd, integrable and supported in $[-\kappa/c,\kappa/c]$, since $\phi$ is bounded and supported in $[0,\kappa]$.
Moreover, $\lambda\phi(\lambda)=-2g(\lambda/c)$ for $\lambda\ge0$.
Since $g$ is odd,
\begin{equation*}
    \int_{\frac{r-d}{c}}^{\frac{d-r}{c}} g(t') \dint t' =0.
\end{equation*}
Hence in both cases $r\le d$ and $r\ge d$,
\begin{equation*}
  a(r)=\frac{2\pi}{d}\int_{|r-d|}^{r+d}\lambda\phi(\lambda)\,\dint\lambda
      = -\frac{4\pi c}{d} \int_{\frac{r-d}{c}}^{\frac{r+d}{c}} g(t') \dint t'.
\end{equation*}
The radius $r$ therefore enters only through the integration limits, and not through the integrand.
The function $\Lambda(\tau) = \int_{0}^{\tau} g(t') \dint t'$ is absolutely continuous with $\Lambda' = g$ almost everywhere.
Since
\begin{equation*}
  a(r) = -\frac{4\pi c}{d} \left( \Lambda\left(\frac{r+d}{c}\right) - \Lambda\left(\frac{r-d}{c}\right) \right),
\end{equation*}
we obtain, for almost every $r > 0$
\begin{equation} \label{eq:def:diff_a}
        \frac{1}{4 \pi} a'(r) = -\frac{1}{d} \left( g\left(\frac{r+d}{c}\right) - g\left(\frac{r-d}{c}\right) \right) =  \frac{1}{d} g\left(\frac{r-d}{c}\right)
\end{equation}
where $g(c^{-1}(r+d)) = 0$ since $r+d > d \geq \dist_{\min} \geq \kappa$ under \Cref{ass:setting}\ref{ass:setting:separation}.

\paragraph{Gathering}
Since $\frac{1}{4\pi c} \partial_t a(i,\bx,ct) = \frac{1}{4\pi} a'(ct)$, restoring the index $i$ and plugging \eqref{eq:def:diff_a} into the definition of $p$ gives
\begin{equation*}
    p(\bx, t) = \sum_{i \in [N]} g\left( t - \frac{\|\bx - \bx_i\|}{c}\right) \frac{\bp_0[i]}{\|\bx - \bx_i\|}.
\end{equation*}
This is the announced identity.

\section{Mathematical justification of \Cref{sec:implementations}} \label{sec:conv_form}

This appendix provides the justification of the arguments developed in \Cref{sec:implementations}.
We will use the framework of distributions, as we believe it is the most widespread in the community. The proposed arguments can also be framed with Radon measures.

\subsection{Quick reminders on distributions}

This section gathers some facts on distributions that will be used in the next section. They are borrowed from textbooks such as \cite{schwartz1951theorie,friedlander1998introduction} in which more details can be found.

Let $\D(\R)$ be the set of infinitely differentiable functions with compact support.
A distribution acts on $\D(\R)$ with a pairing denoted $T(\varphi) = \langle T, \varphi \rangle_{\D',\D}$, and is a continuous linear form on $\D(\R)$ in the sense that for all compact $\mathcal{K} \subset \R$, there exists a constant $C_{\mathcal{K}} \geq 0$ and an integer $m_{\mathcal{K}} \geq 0$ such that
\begin{equation} \label{eq:local_boundedness}
       \left| \langle T, \varphi \rangle_{\D',\D} \right| \leq C_{\mathcal{K}} \sup_{0 \leq k \leq m_{\mathcal{K}}} \|\varphi^{(k)}\|_{L^\infty}, \qquad \forall \varphi \in \D(\R) \text{ such that } \supp \varphi \subseteq \mathcal{K}.
\end{equation}
The space $\D'(\R)$ contains all distributions.

Every $F \in L^1_{\mathrm{loc}}(\R)$ defines a distribution $T_F \in \D'(\R)$ by $\langle T_F, \varphi \rangle_{\D',\D} = \int_\R F \varphi \dint t$, and we systematically identify $F$ with $T_F$. The Dirac mass $\delta_{a} \in \D'(\R)$ at $a \in \R$ is defined by $\langle \delta_a, \varphi \rangle_{\D',\D} = \varphi(a)$ for all $\varphi \in \D(\R)$.

Let $g$ be a compactly supported $L^1_{\mathrm{loc}}(\R)$ function and define $\overline g := g(-\cdot)$.
For $\varphi \in \D(\R)$, the function $\overline g \star \varphi \in \D(\R)$.
Given $T \in \D'(\R)$, the convolution $g \star T$ is also in $\D'(\R)$ and is defined by
\begin{equation*}
       \langle g \star T, \varphi \rangle_{\D',\D} := \langle T, \overline g \star \varphi \rangle_{\D',\D}, \qquad \forall \varphi \in \D(\R).
\end{equation*}

For $g \in L^1_{\mathrm{loc}}(\R)$, we have the identity $g \star \delta_{a} = g(\cdot - a)$ since
\begin{equation*}
 \langle g\star\delta_{a},\varphi\rangle_{\D',\D} = \langle \delta_{a}, \overline g\star\varphi\rangle_{\D',\D} = (\overline g \star \varphi)(a) = \int_\R g(t - a)\varphi(t)\dint t.
\end{equation*}
This implies that for $T = \sum_{n \in [J]} \gamma_n \delta_{a_n}$, a finite sum of Dirac masses,
\begin{equation} \label{eq:conv_sum_dirac}
       g \star T = \sum_{n} \gamma_n g(\cdot - a_n).
\end{equation}

Let $\mu$ be a measure on $X$ and let $T : X \to \D'(\R)$ be a family of distributions such that, for every $\varphi \in \D(\R)$, the map $\bx \mapsto \langle T(\bx), \varphi\rangle_{\D',\D}$ is measurable and integrable with respect to $\mu$.
Then $\int_X T(\bx) \dint \mu(\bx)$ is a linear form on $\D(\R)$ defined by
\begin{equation*}
       \left\langle \int_X T(\bx) \dint \mu(\bx), \varphi \right\rangle_{\D',\D} = \int_X \langle T(\bx), \varphi \rangle_{\D',\D} \dint \mu(\bx) \quad \forall \varphi \in \D(\R).
\end{equation*}
It is a distribution in $\D'(\R)$ as soon as it is continuous in the sense of \eqref{eq:local_boundedness}. This weak integration of distributions depending on parameters can be thought of as Pettis integrals.

\subsection{Justification of the derivations}

Recall that for all $t > 0$ and $\bx \in \Sigma$
\begin{equation*}
       \begin{aligned}
              p(\bx,\cdot) &= \sum_{i \in [N]} \frac{\bp_0[i]}{\|\bx - \bx_i\|} \, g\left( \cdot - \frac{\|\bx - \bx_i\|}{c} \right) \\
                           &= \sum_{i \in [N]} \frac{\bp_0[i]}{\|\bx - \bx_i\|} \, g \star \delta_{\frac{\|\bx - \bx_i\|}{c}} \\
                           &= g \star \underbrace{\left(\sum_{i \in [N]} \frac{\bp_0[i]}{\|\bx - \bx_i\|} \delta_{\frac{\|\bx - \bx_i\|}{c}} \right)}_{f(\bx, \cdot)}
       \end{aligned}
\end{equation*}
where we used \eqref{eq:conv_sum_dirac}, valid since $g$ is integrable (\Cref{lem:smoothness_g}). As a finite weighted sum of Dirac masses, $f(\bx, \cdot)$ defines a distribution in $\D'(\R)$.

In the following $\dint \mu$ refers to any positive measure on $\Sigma$ with total mass $|\Sigma|$, allowing us to study the two cases where $\dint \mu = \sum_{j \in [Q]} \Delta \bq_j \delta_{\bq_j}$, obtained with the quadrature rule for \Cref{subsec:point_method} or $\dint \mu = \dint \sigma$ for \Cref{subsec:surface_method}.

\begin{lemma} \label{lem:conv_measure}
       Let $\mu$ be a measure supported on $\Sigma$ and total mass $|\Sigma|$. Then
       \begin{equation*}
              \int_\Sigma p (\bx, \cdot) \dint\mu = g \star f_{\Sigma}, \quad \text{with } \quad f_{\Sigma} = \int_{\Sigma} f(\bx, \cdot) \dint \mu \in \D'(\R).
       \end{equation*}
\end{lemma}

\begin{proof}
       Recall that
       \begin{equation*}
              \begin{aligned}
                     f(\bx,\tau) &= \sum_{i \in [N]} \alpha_{i}(\bx) \, \delta_{\tau = \tau_{i}(\bx)} \\
                     \tau_{i}(\bx) &= c^{-1}\|\bx - \bx_i\| \\
                     \alpha_{i}(\bx) &= \bp_0[i](c\tau_{i}(\bx))^{-1} \\
                     p(\bx,t) &= (g \star f(\bx,\cdot))(t) = \sum_{i \in [N]} \alpha_{i}(\bx) g(t - \tau_{i}(\bx))
              \end{aligned}
       \end{equation*}
       The function $m$ defined by
       \begin{equation*}
              m(t) = \int_\Sigma p(\bx, t) \dint\mu(\bx)
       \end{equation*}
       belongs to $L^1(\R)$ since
       \begin{equation*} \label{eq:integrability_check_m}
              \begin{aligned}
                     \int_{\R} \int_\Sigma \left| p(\bx, t) \right| \dint\mu(\bx) \dint t & \leq \sum_{i \in [N]} \int_{\R} \int_{\Sigma}  \left| \alpha_{i}(\bx)  g(t - \tau_{i}(\bx))\right| \dint \mu(\bx) \dint t \\
                            & \leq \sum_{i \in [N]} |\Sigma| \frac{|\bp_0[i]|}{\dist_{\min}} \|g \|_{L^1} \\
                            & \leq |\Sigma| \frac{\|\bp_0\|_1}{\dist_{\min}} \|g \|_{L^1} \\
                            & < + \infty
              \end{aligned}
       \end{equation*}
       since under \Cref{ass:setting}, $\tau_{i}(\bx) \geq c^{-1} \dist_{\min}$ and $\|g \|_{L^1} < +\infty$ due to \Cref{lem:smoothness_g}.
       Moreover, $p(\bx, \cdot)$ is bounded by
       \begin{equation*}
              |p(\bx, t)| \leq \frac{\|\bp_0\|_1}{\dist_{\min}} \|g \|_{L^\infty},
       \end{equation*}
       which is integrable against $\mu$, so that $m$ is continuous by dominated convergence whenever $g$ is, that is when $\phi$ is continuous and vanishes at $\kappa$.

       \begin{description}
              \item[Checking that $f_{\Sigma} \in \D'(\R)$]
              Let $\varphi \in \D(\R)$, since
              \begin{equation*}
                     \langle f(\bx, \cdot) , \varphi \rangle_{\D',\D} = \sum_{i \in [N]} \alpha_{i}(\bx) \varphi(\tau_{i}(\bx))
              \end{equation*}
              for $\bx \in \Sigma$, the function $\bx \in \Sigma \mapsto \langle f(\bx, \cdot) , \varphi \rangle_{\D',\D}$ is continuous and therefore measurable. It is moreover integrable since
              \begin{equation*}
                     \begin{aligned}
                            \int_{\Sigma} |\langle f(\bx,\cdot), \varphi \rangle_{\D',\D}| \dint \mu(\bx) & \leq \| \bp_0\|_1 \| \varphi\|_{L^\infty} |\Sigma| \dist_{\min}^{-1}
                     \end{aligned}
              \end{equation*}
              using that $\| \alpha_i \|_{L^{\infty}(\Sigma)} \leq |\bp_0[i]| \dist_{\min}^{-1}$ under \Cref{ass:setting}. This inequality also allows one to prove that \eqref{eq:local_boundedness} holds for all compact supports, with $m_K = 0$ and $C_K = \| \bp_0\|_1 |\Sigma| \dist_{\min}^{-1}$, therefore proving that $f_{\Sigma} \in \D'(\R)$.

              \item[Showing the equality] Since $f_\Sigma \in \D'(\R)$, we have for $\varphi \in \D(\R)$
              \begin{equation*}
                     \begin{aligned}
                            \langle g \star f_\Sigma, \varphi \rangle_{\D',\D} & = \langle f_\Sigma,  \overline{g} \star \varphi \rangle_{\D',\D} \\
                             & = \int_{\Sigma }\langle f(\bx, \cdot),  \overline{g} \star \varphi \rangle_{\D',\D} \dint \mu(\bx) \\
                             & = \int_{\Sigma }\langle g \star f(\bx, \cdot),  \varphi \rangle_{\D',\D} \dint \mu(\bx) \\
                             & = \int_{\Sigma }\langle p(\bx, \cdot),  \varphi \rangle_{\D',\D} \dint \mu(\bx) \\
                             & = \left \langle \int_{\Sigma} p(\bx, \cdot) \dint \mu(\bx),  \varphi \right \rangle_{\D',\D}
                     \end{aligned}
              \end{equation*}
              showing the equality in $\D'$. The equality therefore holds almost everywhere, and everywhere when $g$ is continuous.
       \end{description}

\end{proof}

\begin{lemma} \label{lem:dirac_indicator}
       Let $a,b \in \R$ such that $a \leq b$. Then
       \begin{equation*}
              \int_{a}^b \delta_{\tau = t} \dint \tau = \indic_{[a,b)}(t).
       \end{equation*}
\end{lemma}

\begin{proof}
       Let $T = \int_{a}^b \delta_{\tau = t} \dint \tau$. We first show that $T \in \D'(\R)$ and then show the equality.
       \begin{description}
              \item[Checking that $T \in \D'(\R)$]
              Let $\varphi \in \D(\R)$, since
              \begin{equation*}
                     \langle \delta_{\tau = t} , \varphi \rangle_{\D',\D} = \varphi(\tau)
              \end{equation*}
              for all $\tau \in \R$, the function $\tau \in \R \mapsto \langle \delta_{\tau = t} , \varphi \rangle_{\D',\D}$ is continuous and therefore measurable. It is moreover integrable since
              \begin{equation*}
                     \begin{aligned}
                            \int_{a}^b |\langle \delta_{\tau = t}, \varphi \rangle_{\D',\D}| \dint \tau & \leq (b-a) \| \varphi\|_{L^\infty}
                     \end{aligned}
              \end{equation*}
              This inequality also allows one to prove that \eqref{eq:local_boundedness} holds for all compact supports, with $m_K = 0$ and $C_K = (b-a)$, therefore proving that $T \in \D'(\R)$.

              \item[Showing the equality] Since $T \in \D'(\R)$, we have for $\varphi \in \D(\R)$
              \begin{equation*}
                     \begin{aligned}
                            \langle T, \varphi \rangle_{\D',\D} & = \int_{a}^b \langle \delta_{\tau = t}, \varphi  \rangle_{\D',\D} \dint \tau \\
                             & = \int_a^b \varphi(\tau) \dint \tau \\
                             & = \int_{\R} \indic_{[a,b)}(t) \varphi(t) \dint t \\
                            & = \left\langle \indic_{[a,b)}, \varphi \right\rangle_{\D',\D}
                     \end{aligned}
              \end{equation*}
              showing the equality in $\D'$. Moreover since $\indic_{[a,b)}$ is in $L^1(\R)$, the equality stands almost everywhere.
       \end{description}
\end{proof}

\section{Discrete convolutions} \label{sec:discrete_conv}

This appendix establishes \eqref{eq:discrete_conv}: convolving a piecewise constant function with a kernel and sampling the result following \eqref{eq:def_operatorS} is computed exactly by a discrete convolution followed by a downsampling.
The derivation is performed for arbitrary functions. For the sake of clarity and ease of understanding, the same notation as in \Cref{sec:implementations} is used.

The kernel $h = e \star g$ belongs to $L^1(\R)$ by Young's inequality, since $e \in L^1(\R)$ by \Cref{ass:setting}\ref{ass:setting:e} and $g \in L^1(\R)$ by \Cref{lem:smoothness_g}. Moreover, it is compactly supported since both $e, g$ are.

\paragraph{Setting}
Let the upsampled grid be defined by $\tau_0 = t_0$ and $\tau_l = t_0 + \frac{l}{UF_s}$ for $l \in \Z$, let $\varphi = \indic_{[t_0,t_0 + \frac{1}{UF_s})}$ and define
\begin{equation*}
       \varphi_l = \varphi\left( \cdot - \frac{l-1}{UF_s} \right), \qquad \text{so that } \; \varphi_l = \indic_{[\tau_{l-1}, \tau_l)}, \quad \forall l \in \Z.
\end{equation*}
Given $\bbf \in \R^{UL}$, let $f = \sum_{l \in [UL]} \bbf[l] \varphi_l$ and $s = h \star f$.
The $\varphi_l$ having disjoint supports of length $(UF_s)^{-1}$, the coefficients $\bbf$ are recovered from $f$ by
\begin{equation*}
       \begin{aligned}
              \bbf[l_1] &= UF_s \langle f, \varphi_{l_1} \rangle = UF_s \int_{\tau_{l_1-1}}^{\tau_{l_1}} f(t) \dint t, \quad \forall l_1 \in [UL],
       \end{aligned}
\end{equation*}
where $\langle \cdot, \cdot \rangle$ denotes the scalar product of $L^2(\R)$.
We define likewise the coefficients of $s$ on the upsampled grid, denoted $\bs^{U}$ to distinguish them from the coefficients $\bs$ on the original grid,
\begin{equation*}
       \begin{aligned}
              \bs^{U}[l_2] &= UF_s \langle s, \varphi_{l_2} \rangle = UF_s \int_{\tau_{l_2-1}}^{\tau_{l_2}} s(t) \dint t, \quad \forall l_2 \in \Z,
       \end{aligned}
\end{equation*}
indexed by the whole of $\Z$ since $s$ need not be supported in $[\tau_0, \tau_{UL}]$.

\paragraph{Discrete convolution}

By linearity of the convolution and finiteness of the sum,
\begin{equation*}
       \begin{aligned}
              s &= \sum_{l_1 \in [UL]} \bbf[l_1] (h \star \varphi_{l_1}),
       \end{aligned}
\end{equation*}
so that
\begin{equation*}
       \begin{aligned}
              \bs^{U}[l_2] &= UF_s  \sum_{l_1 \in [UL]} \bbf[l_1] \langle h \star \varphi_{l_1}, \varphi_{l_2}  \rangle \\
              & = UF_s  \sum_{l_1 \in [UL]} \bbf[l_1] \int_{\R \times \R} h(t-\tau) \varphi_{l_1}(\tau) \varphi_{l_2}(t) \dint \tau \dint t
       \end{aligned}
\end{equation*}
using Fubini's theorem, licensed by Tonelli's theorem since a successive application of the Cauchy--Schwarz inequality and of Young's inequality for convolution gives
\begin{equation*}
       \int_{\R \times \R} \left| h(t-\tau) \varphi_{l_1}(\tau) \varphi_{l_2}(t) \right| \dint \tau \dint t \leq \| h \|_{L^1} \| \varphi_{l_2} \|_{L^2} \| \varphi_{l_1} \|_{L^2} = \frac{\| h \|_{L^1}}{UF_s} < + \infty.
\end{equation*}
The change of variable $\tau \gets \tau + \frac{l_1-1}{UF_s}$, $t \gets t + \frac{l_2-1}{UF_s}$ then removes the dependence on $l_1$ and $l_2$ except through their difference,
\begin{equation*}
       \begin{aligned}
              \bs^{U}[l_2]& = UF_s  \sum_{l_1 \in [UL]} \bbf[l_1] \int_{\R \times \R} h \left(t-\tau + \frac{l_2 - l_1}{UF_s}\right) \varphi(\tau) \varphi(t) \dint \tau \dint t \\
              & = \sum_{l_1 \in [UL]} \bh[l_2 - l_1] \bbf[l_1], \quad \forall l_2 \in \Z,
       \end{aligned}
\end{equation*}
with $\bh : \Z \to \R$ defined by
\begin{equation*}
       \begin{aligned}
              \bh[n] & = UF_s \int_{\R \times \R} h\left(t-\tau + \frac{n}{UF_s}\right) \varphi(\tau) \varphi(t) \dint \tau \dint t, \quad \forall n \in \Z.
       \end{aligned}
\end{equation*}
Since $\varphi$ is supported in an interval of length $(UF_s)^{-1}$, the integrand vanishes unless $\frac{n}{UF_s}$ is within $(UF_s)^{-1}$ of the support of $h$, which gives the support of $\bh$ announced in \Cref{subsec:overview}.
In other words, $\bs^{U} = \bh \star \bbf$.

\paragraph{Downsampling}

Consider now the original time grid, $t_l = t_0 + l/F_s$ for $l \in [L]$, and the coefficients
\begin{equation*}
       \widehat{\bs}[l] = F_s \int_{t_{l-1}}^{t_l} s(t) \dint t .
\end{equation*}
Since $t_{l-1} = \tau_{U(l-1)}$ and $t_l = \tau_{Ul}$, the cell $[t_{l-1}, t_l)$ is partitioned by the $U$ cells of the upsampled grid it contains, so that
\begin{equation*}
       \begin{aligned}
              \widehat{\bs}[l] &= F_s \sum_{l' \in [U]} \int_{\tau_{U(l-1) + (l'-1)}}^{\tau_{U(l-1) + l'}} s(t) \dint t \\
              & = F_s \sum_{l' \in [U]} \frac{1}{UF_s} UF_s\int_{\tau_{U(l-1) + (l'-1)}}^{\tau_{U(l-1) + l'}} s(t) \dint t \\
              & = \frac{1}{U} \sum_{l' \in [U]} \bs^{U}[U(l-1) + l'] = \left( \dsop \bs^{U} \right)[l].
       \end{aligned}
\end{equation*}

\paragraph{Application}

Applying the above with $f = \widehat{f}_\Sigma$, whose coefficients are $\bbf$, gives $\widehat{s} = h \star \widehat{f}_\Sigma = e \star \widehat{m}$, $\bh$ as defined in \eqref{eq:def_bh} and $\bs^{U} = \bh \star \bbf$. The coefficients of $\widehat{s}$ on the original grid are then $\mathcal{S}(e \star \widehat{m}) = \widehat{\bs}$.
Hence, $\widehat{\bs} = \dsop \bs^{U} = \dsop (\bh \star \bbf)$, that is \eqref{eq:discrete_conv}.

\section{Proofs of \Cref{prop:approximation_points,prop:approximation_surface}} \label{app:proof_prop_23}

\subsection{Overview of the proofs}

The approximation strategy of both methods can be summarized as:
\begin{center}
\resizebox{\linewidth}{!}{%
       \begin{tikzpicture}

  \node (FF)   [block]                        {\Cref{prop:point_value_pressure_field}};
  \node (junc) [junction,right=2cm of FF]  {};

  \draw[arr] ++(-2.2cm, 0) -- node[above]{$m$} (FF);
  \draw[arr] (FF) -- node[above]{$\displaystyle m = g \star f_{\Sigma}$} (junc);

  \node (Q)    [block,  above right=0.8cm and 0.5cm of junc] {Quadrature};
  \node (PA)   [block,  right=1cm of Q]                    {Piecewise Approx. \Cref{lem:quantization}};

  \draw[arr] (junc) |- node[above, near end]{} (Q);
  \draw[arr] (Q)    -- node[above]{$\widetilde{m}$} (PA);

  \coordinate (PO_y) at ($(junc.south) + (0,-1.2cm)$);
  \node (PO) [block] at (PA |- PO_y) {Piecewise Approx. \Cref{lem:quantization}};

  \node (CA)   [block,  right=1cm of PO]                   {Coeff. Approx};
  \node (out2) [right=1cm of CA]                           {};

    \draw[arr] (junc) |- node[above, near end]{$\widetilde m$} (PO);
  \draw[arr] (PO)  -- node[above]{$\overline{m}$} (CA);
  \draw[arr] (CA)  -- node[above]{} (out2) node[right]{$\widehat{m}$};

\node (out1) at (out2 |- PA) {};
\draw[arr] (PA)   -- node[above]{$\overline{m}$} (out1) node[right]{$\widehat{m}$};

  \node (box_top) [dashbox, fit=(Q)(PA)(out1),
        label={[font=\small\itshape]above:Point detector approx.\ \Cref{prop:approximation_points}}] {};

  \node (box_bot) [dashbox, fit=(PO)(CA)(out2),
        label={[font=\small\itshape]below:Surface integral \Cref{prop:approximation_surface}}] {};

\end{tikzpicture}
}
\end{center}
The diagram reads as a chain of successive approximations of $f_\Sigma$ leading to different approximations of $m = g \star f_\Sigma$.
First, $m$ is written as $m = g \star f_{\Sigma}$ following \Cref{prop:point_value_pressure_field}.
The point detector approximation consists in approximating the surface integral over $\Sigma$ defining $f_\Sigma$, by quadrature giving $\widetilde{f}_\Sigma$. Then $\widetilde{f}_{\Sigma}$ is further approximated by the piecewise constant function $\overline{f}_{\Sigma}$. No further approximations are used, hence $\widehat{f}_\Sigma = \overline{f}_\Sigma$.
The surface method treats the integral over $\Sigma$ analytically so that $\widetilde{f}_\Sigma = f_\Sigma$ and the piecewise constant approximation is directly applied to $f_\Sigma$. The coefficients $\overline{\bbf}$ of $\overline{f}_\Sigma$ are then further approximated by a midpoint rule, giving $\bbf$ and $\widehat{f}_\Sigma$.
Each bound of \Cref{prop:approximation_points,prop:approximation_surface} is obtained by summing these contributions along the chain.

We recall that
\begin{equation*}
       \begin{aligned}
              m(t) &=  g \star f_{\Sigma} (t)  = \int_{\Sigma} \sum_{i \in [N]} g\left( t - \frac{\|\bx - \bx_i\|}{c} \right) \frac{\bp_0[i]}{\|\bx - \bx_i\|}  \dint\sigma(\bx) \\
              \widetilde{m}(t) &= g \star \widetilde{f}_\Sigma(t), \quad \text{with } \; \widetilde{f}_\Sigma =  \sum_{j \in [Q], i \in [N]} \frac{\bp_0[i] \Delta \bq_j}{\| \bq_j - \bx_i\|} \delta_{\frac{\| \bq_j - \bx_i\|}{c}} \\
              \overline{m}(t) &= g \star \overline{f}_\Sigma(t), \quad \text{with } \; \overline{f}_\Sigma = \sum_{l \in [UL]} \overline{\bbf}[l] \indic_{[\tau_{l-1}, \tau_l)}, \quad \text{and } \; \overline{\bbf}[l] = U F_s \int_{\tau_{l-1}}^{\tau_l} \widetilde{f}_\Sigma(t) \dint t \\
              \widehat{m}(t) &= g \star \widehat{f}_\Sigma(t), \quad \text{with } \; \widehat{f}_\Sigma = \sum_{l \in [UL]} \bbf[l] \indic_{[\tau_{l-1}, \tau_l)}.
       \end{aligned}
\end{equation*}

Recall that $s = e \star m$  and define its approximations
\begin{equation*}
       s = e \star m = h \star f_{\Sigma}, \qquad
       \widetilde{s} = e \star \widetilde{m} = h \star \widetilde{f}_\Sigma, \qquad
       \overline{s} = e \star \overline{m} = h \star \overline{f}_\Sigma, \qquad
       \widehat{s} = e \star \widehat{m} = h \star \widehat{f}_\Sigma,
\end{equation*}
where each identity follows from the associativity of convolution.

Due to discrete convolution being exact, the coefficients $\widehat{\bs} = \dsop(\bh \star \bbf)$ are given by
       \begin{equation*}
              \widehat{\bs}[l] = \frac{1}{U} \sum_{l' \in [U]} U F_s \int_{\tau_{U(l-1)+(l'-1)}}^{\tau_{U(l-1)+l'}} \widehat{s} (t) \dint t.
       \end{equation*}

       Moreover, the value of $\bs$ can be decomposed by the integration on the finer grid as
       \begin{equation*}
              \begin{aligned}
                     \bs[l] &= F_s \int_{t_{l-1}}^{t_l} s(t) \dint t \\
                     & = F_s \sum_{l' \in [U]} \int_{\tau_{U(l-1)+(l'-1)}}^{\tau_{U(l-1)+l'}} s(t) \dint t
              \end{aligned}
       \end{equation*}

       Therefore,
       \begin{equation*}
              \begin{aligned}
                     \left|\bs[l] - \widehat{\bs}[l]\right| &\leq  F_s \sum_{l' \in [U]} \int_{\tau_{U(l-1)+(l'-1)}}^{\tau_{U(l-1)+l'}} | s - \widehat{s} | \dint t \\
                     & \leq UF_s (UF_s)^{-1} \| s - \widehat{s} \|_{L^\infty}
              \end{aligned}
       \end{equation*}
       The term $\| s - \widehat{s}\|_{L^\infty}$ can be bounded leveraging the following decomposition for $t > 0$,
       \begin{equation} \label{eq:error_decomp}
              \begin{aligned}
                     |s(t) - \widehat{s}(t)| &\leq \left| s(t) - \widetilde{s}(t) \right| + \left| \widetilde{s}(t) - \overline{s}(t) \right| + \left| \overline{s}(t) - \widehat{s}(t)  \right| \\
                     &= \err_1 + \err_2 + \err_3.
              \end{aligned}
       \end{equation}

The errors are bounded independently for each method.

\subsection{Proofs}

\begin{proof}[{Proof of \Cref{prop:approximation_points}}]
       The error decomposition \eqref{eq:error_decomp} for the point detector method gives $\err_3 = 0$ since in this case $\overline{\bbf} = \bbf$, so that $\widehat{s} = \overline{s}$. It remains to bound $\err_1$ and $\err_2$.
       \begin{description}
              \item[Bounding $\err_1$:] this term is due to the quadrature rule used to approximate the integral on $\Sigma$.
              Let $q(\bx, \cdot) := e \star p(\bx, \cdot) = h \star f(\bx,\cdot)$, so that $s = \int_\Sigma q(\bx,\cdot) \dint\sigma(\bx)$, and let $q_t : (u, v) \in D \mapsto q(\sigma(u, v), t)$.

              Note that $(u_j,v_j)$ satisfies
              \begin{equation*}
                \int_{D_j} \bn \dint \sigma(u,v) = 0
              \end{equation*}
              where $\bn = (u - u_j, v - v_j)$ and $(u_j,v_j) \in D_j$ since $D_j$ is convex.

              Using the above notation, the error $\err_1$ can be decomposed as
              \begin{equation*}
                     \begin{aligned}
                            \err_1 &= \left|\int_{D} q_t(u, v) \dint \sigma(u, v) - \sum_{j \in [Q]} \Delta \bq_j q_t(u_j,v_j)\right| \\
                            & = \left| \sum_{j \in [Q]} \int_{D_j} q_t(u, v) - q_t(u_j,v_j) \dint \sigma(u, v)  \right|
                     \end{aligned}
              \end{equation*}
since $\int_{D_j} \dint \sigma(u, v) = \Delta \bq_j$.

              First, observe that for all $t> 0$ the map $q_t$ has a Hessian $H[q_t]$ satisfying
              \begin{equation*}
                     \esssup_{t > 0, (u, v) \in D} \| H[q_t](u, v)\| \lesssim \| \bp_0\|_1 C(h,\dist_{\min}, \sigma),
              \end{equation*}
              using \Cref{lem:smoothness_f} and where $C(h,\dist_{\min}, \sigma)$ is a constant depending on $h,\dist_{\min}$ and $\sigma$ defined in the lemma.

              Given a cell $j \in [Q]$, this smoothness allows writing the Taylor expansion around its barycenter $(u_j, v_j)$
              \begin{equation*}
                     q_t(u, v) = q_t(u_j,v_j) + \langle \nabla q_t(u_j,v_j), \bn \rangle + \int_{0}^{1} (1-r) \langle H[q_t]((u_j,v_j)+r\bn) \bn, \bn \rangle \dint r
              \end{equation*}
              where $\bn = (u - u_j, v - v_j)$. Note that since $(u_j,v_j)$ is the barycenter of the cell $j$,
              \begin{equation*}
                     \begin{aligned}
                            &\int_{(u,v) \in D_j} q_t(u, v) - q_t(u_j,v_j) \dint\sigma(u, v) \\
                            &\qquad = \int_{(u,v) \in D_j} \int_{0}^{1} (1-r) \langle H[q_t]((u_j,v_j)+r\bn) \bn, \bn \rangle \dint r \dint\sigma(u, v).
                     \end{aligned}
              \end{equation*}
              This allows one to bound
              \begin{equation*}
                     \begin{aligned}
                            \left| \int_{(u,v) \in D_j} q_t(u, v) - q_t(u_j,v_j) \dint\sigma(u, v) \right| &\leq  \int_{(u,v) \in D_j} \frac{1}{2} \esssup_{ (u, v) \in D} \| H[q_t](u, v)\| \| \bn \|^2 \dint\sigma(u, v) \\
                            & \lesssim C(h,\dist_{\min}, \sigma) \|\bp_0\|_1 \diam(D_j)^2 \Delta \bq_j.
                     \end{aligned}
              \end{equation*}
              Summing over all the cells gives
              \begin{equation*}
                     \begin{aligned}
                            \err_1    & \lesssim C(h,\dist_{\min}, \sigma) \|\bp_0\|_1 \sum_{j \in [Q]} \diam(D_j)^2 \Delta \bq_j\\
                                   & \lesssim C(h,\dist_{\min}, \sigma) \|\bp_0\|_1 \sum_{j \in [Q]} \Delta \bq_j Q^{-1}  \\
                                   & \lesssim C(h,\dist_{\min}, \sigma) \|\bp_0\|_1 |\Sigma| Q^{-1}
                     \end{aligned}
              \end{equation*}
              using the fact that for all $j$, $\diam(D_j) \lesssim Q^{-1/2}$.
              \item[Bounding $\err_2$:] this term is due to the quantization of the arrival times $c^{-1}\|\bq_j - \bx_i\|$ on the grid $(\tau_l)_{l=0}^{UL}$.
              Recall that $\err_2 = | \widetilde{s}(t) - \overline{s}(t)|$ with
              \begin{equation*}
                     \begin{aligned}
                            \widetilde{s}(t) &= \sum_{j \in [Q]} \sum_{i \in [N]} h \left( t - \frac{\| \bq_j - \bx_i\|}{c}\right) \frac{\bp_0[i] \Delta \bq_j}{\| \bq_j - \bx_i\|}\\
                            \overline{s}(t) &= \left( h \star \sum_{l \in [UL]} \overline{\bbf}[l] \indic_{[\tau_{l-1}, \tau_l)} \right)(t) \\
                            \overline{\bbf}[l] & = U F_s \sum_{j \in [Q]} \sum_{i \in [N]} \frac{\bp_0[i] \Delta \bq_j}{\| \bq_j - \bx_i\|} \indic_{[\tau_{l-1}, \tau_l)} \left( \frac{\| \bq_j - \bx_i\|}{c} \right).
                     \end{aligned}
              \end{equation*}
              Using \Cref{lem:quantization} with the measure $\mu = \sum_{j \in [Q]} \Delta \bq_j \delta_{\bq_j}$ gives
              \begin{equation*}
                     \err_2 \leq \dist_{\min}^{-1} |\Sigma| \|h'\|_{L^\infty} \|\bp_0\|_1 (UF_s)^{-1}.
              \end{equation*}
       \end{description}
       The final bound is obtained by combining $\err_1$, $\err_2$.

       When $h \in W^{1,\infty}(\R)$ only, the mean value inequality on each convex cell $D_j$ replaces the Taylor expansion.
       Since $\|\nabla q_t(u,v)\| \lesssim \|\bp_0\|_1 \bigl( \|h'\|_{L^\infty}/(c\,\dist_{\min}) + \|h\|_{L^\infty}/\dist_{\min}^2 \bigr) |\sigma|_{W^{1,\infty}}$ by \Cref{lem:smoothness_f}, it gives $\err_1 \lesssim \|\bp_0\|_1 \sum_{j \in [Q]} \diam(D_j) \Delta \bq_j \lesssim |\Sigma| \, \|\bp_0\|_1 \, Q^{-\frac{1}{2}}$, with a constant depending only on $\|h\|_{L^\infty}$, $\|h'\|_{L^\infty}$, $\dist_{\min}$ and $|\sigma|_{W^{1,\infty}}$.

\end{proof}

\begin{proof}[{Proof of \Cref{prop:approximation_surface}}]
        The error decomposition \eqref{eq:error_decomp} for the surface method gives $\err_1 = 0$ since in this case $\widetilde{s} = s$. It remains to bound $\err_2$ and $\err_3$.

       \begin{description}
              \item[Bounding $\err_2$:] this term is due to the quantization of the arrival times $c^{-1}\|\bx - \bx_i\|$, for $\bx \in \Sigma$, on the grid $(\tau_l)_{l=0}^{UL}$.
              Recall that $\err_2 = |\widetilde{s}(t) - \overline{s}(t)|$ with
              \begin{equation*}
                     \begin{aligned}
                            \widetilde{s}(t) &= \sum_{i \in [N]} \int_{\Sigma} h \left( t - \frac{\| \bx - \bx_i\|}{c}\right) \frac{\bp_0[i]}{\| \bx - \bx_i\|} \dint\sigma(\bx)\\
                            \overline{s}(t) &= \left( h \star \sum_{l \in [UL]} \overline{\bbf}[l] \indic_{[\tau_{l-1}, \tau_l)} \right)(t) \\
                            \overline{\bbf}[l] & = U F_s \sum_{i \in [N]} \int_\Sigma \frac{\bp_0[i]}{\| \bx - \bx_i\|} \indic_{[\tau_{l-1}, \tau_l)} \left( \frac{\| \bx - \bx_i\|}{c} \right) \dint \sigma(\bx).
                     \end{aligned}
              \end{equation*}
              Using \Cref{lem:quantization} with the measure $\dint \mu = \dint \sigma$ gives
              \begin{equation*}
                     \err_2 \leq \dist_{\min}^{-1} |\Sigma| \|h'\|_{L^\infty} \|\bp_0\|_1 (UF_s)^{-1}.
              \end{equation*}
              \item[Bounding $\err_3$:] This last error term is due to the approximation of $\| \bx - \bx_i\|^{-1}$ by $(c \tau_{l-\frac{1}{2}})^{-1}$.
              Recalling that
              \begin{equation*}
                     \begin{aligned}
                           \overline{s}(t) - \widehat{s}(t) &= \int_{\R} h(t - \tau) \sum_{l \in [UL]} \left( \overline{\bbf}[l] - \bbf[l] \right) \indic_{[\tau_{l-1}, \tau_l)} (\tau) \dint \tau \\
                           &= \sum_{l \in [UL]} \left( \overline{\bbf}[l] - \bbf[l] \right) \int_{\R} h(t - \tau)  \indic_{[\tau_{l-1}, \tau_l)} (\tau) \dint \tau
                     \end{aligned}
              \end{equation*}
              so that
              \begin{equation*}
                     \begin{aligned}
                           \|\overline{s} - \widehat{s}\|_{L^\infty} & \leq \frac{\| h\|_{L^\infty}}{UF_s} \left\| \overline{\bbf} - \bbf \right\|_1.
                     \end{aligned}
              \end{equation*}

              First observe that for all $l \in [UL]$,
              \begin{equation*}
                     \begin{aligned}
                           |\overline{\bbf}[l] - \bbf[l]| & \leq UF_s \sum_{i} | \bp_0[i]| \int_{\mathcal{M}_i^l} \left| \frac{1}{\|\bx - \bx_i\|} - \frac{1}{c \tau_{l-\frac{1}{2}}} \right| \dint\sigma(\bx) \\
                           & \leq UF_s \sum_{i} | \bp_0[i]| \left| \mathcal{M}_i^l \right| \sup_{\bx \in \mathcal{M}_i^l} \left| \frac{1}{\| \bx - \bx_i\|} - \frac{1}{c\tau_{l-\frac{1}{2}}} \right| \\
                    \end{aligned}
              \end{equation*}
              Due to \eqref{eq:min_dist}, $\|\bx - \bx_i\| \geq \dist_{\min}$.
              Note that \Cref{alg:surface_based} only computes $\bbf[l]$ for $l_{\min} \leq l \leq l_{\max}$ so that the midpoint satisfies $c\tau_{l-\frac12} \geq r_{\min}(i) - \frac{c}{2UF_s} \geq \frac{\dist_{\min}}{2}$ as soon as $c/(UF_s) \leq \dist_{\min}$.
              Since on $\mathcal{M}_i^l$ both $\| \bx - \bx_{i} \|$ and $c \tau_{l-\frac{1}{2}}$ belong to $[c\tau_{l-1}, c\tau_l)$ it follows that
            \begin{equation*}
                \left| \frac{1}{\|\bx - \bx_i\|} - \frac{1}{c\tau_{l-\frac12}} \right| \leq \frac{c}{2UF_s}  \frac{1}{\dist_{\min}}  \frac{2}{\dist_{\min}} = \frac{c}{UF_s \dist_{\min}^2}.
        \end{equation*}
            Therefore,
                \begin{equation*}
                     \begin{aligned}
                             |\overline{\bbf}[l] - \bbf[l]| & \leq \sum_{i} | \bp_0[i]| \left| \mathcal{M}_i^l \right| \frac{c}{\dist_{\min}^2}.
                     \end{aligned}
              \end{equation*}
              Then, summing on $l$ gives:
                \begin{equation*}
                     \begin{aligned}
                            \left\|\overline{\bbf} - \bbf \right\|_1 & \leq \sum_{l} \sum_i | \bp_0[i]| \left| \mathcal{M}_i^l \right| \frac{c}{\dist_{\min}^2} \\
                            & \leq \|\bp_0\|_1 \left| \Sigma \right| \frac{c}{\dist_{\min}^2}
                     \end{aligned}
              \end{equation*}
              since $\sum_l \left| \mathcal{M}_i^l \right| \leq |\Sigma|$ for all $i$. Therefore,
              \begin{equation*}
                     \begin{aligned}
                           \err_3 & \lesssim |\Sigma| \| \bp_0\|_1 \|h\|_{L^\infty} c \dist_{\min}^{-2} (UF_s)^{-1}.
                     \end{aligned}
              \end{equation*}
       \end{description}
       The final bound is obtained by combining $\err_2$ and $\err_3$:
       \begin{equation*}
        \begin{aligned}
              \err_2 + \err_3 &\lesssim |\Sigma| \|\bp_0\|_1 \left( \|h'\|_{L^\infty} + \frac{c}{\dist_{\min}} \|h\|_{L^\infty} \right) \dist_{\min}^{-1} (UF_s)^{-1} \\
              & = C(h,\dist_{\min}) \, |\Sigma| \, \|\bp_0\|_1 \, \dist_{\min}^{-1} (UF_s)^{-1}.
        \end{aligned}
       \end{equation*}
\end{proof}

\subsection{Supporting lemmata}

The proof of \Cref{prop:approximation_points,prop:approximation_surface} will use the following results.

\begin{lemma}  \label{lem:smoothness_g}
       Assume that the extension of $\phi$ by zero belongs to $W^{k, \infty}(\R_+)$.
       Then if $k \geq 1$, the derivative of $g$ defined in \Cref{prop:point_value_pressure_field} is
       \begin{equation*}
              g'(t) = -\frac{1}{2} c \left( \phi(|ct|) + |ct| \phi'(|ct|) \right).
       \end{equation*}
       If $k \geq 2$ its second order weak derivative is
       \begin{equation*}
              g''(t) = -\frac{1}{2} c^2 \left( 2 \, \mathrm{sign}(t) \phi'(|ct|) + ct \phi''(|ct|) \right).
       \end{equation*}
       Moreover,
       \begin{equation*}
              \begin{aligned}
                     \| g \|_{L^\infty} &\leq  \frac{1}{2} \kappa \| \phi \|_{L^\infty} \\
                     \| g' \|_{L^\infty} &\leq \frac{1}{2} c (\|\phi\|_{L^\infty} + \kappa \|\phi'\|_{L^\infty} ) \\
                     \| g'' \|_{L^\infty} &\leq \frac{1}{2} c^2 ( 2 \| \phi ' \|_{L^\infty} + \kappa \|\phi'' \|_{L^\infty} ) \\
                     \|g\|_{L^1} & \leq \kappa^2 c^{-1} \| \phi \|_{L^\infty}
              \end{aligned}
       \end{equation*}
       These formulas are the ones used to transfer regularity from $\phi$ onto $g$.
\end{lemma}
\begin{proof}
       The derivatives are obtained using the definition of $g$ and the chain rule. The $L^\infty$ bounds directly follow. The $L^1$ bound is obtained as follows
              \begin{equation*}
                     \begin{aligned}
                           \|g\|_{L^1} & = \int_{|t| < \kappa c^{-1}} \frac{1}{2} |ct| \left| \phi(|ct|) \right| \dint t \\
                           & \leq \frac{1}{2} \kappa \| \phi \|_{L^\infty} \int_{|t| < \kappa c^{-1}} \dint t \\
                           & = \frac{1}{2} \kappa \| \phi \|_{L^\infty} \cdot 2\kappa c^{-1} \\
                           & = \kappa^2 c^{-1} \| \phi \|_{L^\infty}.
                     \end{aligned}
              \end{equation*}
\end{proof}

\begin{lemma} \label{lem:smoothness_f}
       Fix $t > 0$ and a voxel $\bx_i$, translate coordinates so that $\bx_i = 0$. Let $\varrho : D \to \R$, $G : (0,\infty) \to \R$ and $q_{t,i} : D \to \R$ be the functions defined by
       \begin{equation*}
              \begin{aligned}
                     \varrho(u, v) & = \| \sigma(u, v) \| \\
                     G(r) & = r^{-1} h\left(t - c^{-1}r\right) \\
                     q_{t,i}(u, v) &= G(\varrho(u, v)),
              \end{aligned}
       \end{equation*}
       with $h = e \star g$ as in \eqref{eq:def_h}, so that the map $q_t$ of \Cref{prop:approximation_points} decomposes as $q_t = \sum_{i \in [N]} \bp_0[i]\, q_{t,i}$.
       Then under \Cref{ass:smoothness_transducer},
       \begin{equation*}
              \begin{aligned}
                     \nabla q_{t,i}(u, v) &= G'(\varrho(u, v)) \nabla \varrho(u, v) \\
                     H[q_{t,i}](u, v) &= G''(\varrho(u, v)) \nabla \varrho(u, v) \nabla \varrho(u, v)^T + G'(\varrho(u, v)) H[\varrho](u, v) \\
              \end{aligned}
       \end{equation*}
       with
       \begin{equation*}
              \begin{aligned}
              G'(r) &= -\frac{1}{cr} h'\left(t - c^{-1}r\right) - \frac{1}{r^2} h\left(t - c^{-1}r\right) \\
              G''(r) &= -\frac{2}{r} G'(r) + \frac{1}{c^2r} h''\left(t - c^{-1}r\right) \\
              \nabla \varrho(u, v) & =J[\sigma](u, v)^T \frac{\sigma(u, v)}{\varrho(u, v)} \\
              H[\varrho](u, v) &= B(u, v) + \frac{1}{\varrho(u, v)} J[\sigma](u, v)^T \left( \Id -\frac{\sigma(u, v) \sigma(u, v)^T}{\varrho(u, v)^2} \right) J[\sigma](u, v) \\
              B(u, v) &= \begin{pmatrix} \left\langle \partial^2_u \sigma(u, v) , \frac{\sigma(u, v)}{\varrho(u, v)} \right\rangle & \left\langle \partial^2_{uv} \sigma(u, v) , \frac{\sigma(u, v)}{\varrho(u, v)} \right\rangle \\ \left\langle \partial^2_{uv} \sigma(u, v) , \frac{\sigma(u, v)}{\varrho(u, v)} \right\rangle & \left\langle \partial^2_v \sigma(u, v) , \frac{\sigma(u, v)}{\varrho(u, v)} \right\rangle\end{pmatrix}
              \end{aligned}
       \end{equation*}
       Moreover, under \Cref{ass:setting,ass:smoothness_transducer} and $h \in W^{2, \infty}(\R)$,
       \begin{equation*}
              \| H[q_{t,i}](u, v) \| \lesssim K(h,\dist_{\min}) \left( | \sigma |_{W^{1,\infty}}^2 + \dist_{\min} | \sigma |_{W^{2,\infty}} \right) \dist_{\min}^{-1}
       \end{equation*}
       with
       \begin{equation*}
            K(h,\dist_{\min}) = \frac{\|h''\|_{L^\infty}}{c^{2}} + \frac{\|h'\|_{L^\infty}}{c\,\dist_{\min}} + \frac{\|h\|_{L^\infty}}{\dist_{\min}^{2}}.
       \end{equation*}
       The constant of \Cref{prop:approximation_points} is defined as
       \begin{equation*}
        C(h,\dist_{\min},\sigma) = K(h,\dist_{\min}) \left( | \sigma |_{W^{1,\infty}}^2 + \dist_{\min} | \sigma |_{W^{2,\infty}} \right) \dist_{\min}^{-1}.
       \end{equation*}

\end{lemma}

\begin{proof}
       The expressions of the derivatives are obtained using chain rule.
       Since $\varrho(u,v) \geq \dist_{\min}$ by \Cref{ass:setting}\ref{ass:setting:separation} and the definition \eqref{eq:min_dist} of $\dist_{\min}$, the bounds
       \begin{equation*}
              \begin{aligned}
                     \|G'\|_{L^\infty} & \leq \frac{1}{c\,\dist_{\min}} \|h'\|_{L^\infty} + \frac{1}{\dist_{\min}^2} \|h\|_{L^\infty} \leq K(h,\dist_{\min}), \\
                     \|G''\|_{L^\infty} & \leq \frac{2}{\dist_{\min}} \|G'\|_{L^\infty} + \frac{1}{c^2 \dist_{\min}} \|h''\|_{L^\infty} \lesssim \frac{1}{\dist_{\min}} K(h,\dist_{\min}),
              \end{aligned}
       \end{equation*}
       follow directly, using $h \in W^{2, \infty}(\R)$ for the finiteness of $\|h\|_{L^\infty}, \|h'\|_{L^\infty}, \|h''\|_{L^\infty}$.
       Moreover,
       \begin{equation*}
              \begin{aligned}
                     \| B(u, v) \| & \leq \| B(u, v) \|_F \\
                     & \leq \sqrt{ \|\partial_u^2 \sigma(u, v) \|^2 + 2 \|\partial_{uv}^2 \sigma(u, v) \|^2 + \|\partial_v^2 \sigma(u, v) \|^2 } \\
                     & \leq 2 |\sigma|_{W^{2,\infty}}
              \end{aligned}
       \end{equation*}
       where $\| B(u, v) \|_F$ is the Frobenius norm and where the second inequality has been obtained via a Cauchy--Schwarz inequality and the fact that $\sigma(u, v) \varrho(u, v)^{-1}$ is a unit vector.
       This leads to
              \begin{equation*}
              \begin{aligned}
                     \| H[\varrho](u, v) \| & \leq \| B(u, v) \| + \frac{1}{\varrho(u, v)} \| J[\sigma](u, v)\|^2 \\
                     & \leq 2|\sigma|_{W^{2,\infty}} + \frac{2}{\dist_{\min}} | \sigma |_{W^{1,\infty}}^2,
              \end{aligned}
       \end{equation*}
since $\| J[\sigma](u, v)\|^2 \leq \| \partial_u \sigma(u, v)\| ^2 + \|\partial_v \sigma(u, v)\|^2 \leq 2|\sigma|_{W^{1,\infty}}^2$.
       Plugging these bounds into
       \begin{equation*}
              \begin{aligned}
                     \| H[q_{t,i}](u, v) \| &\leq \|G''\|_{L^\infty} \| \nabla \varrho(u, v) \|^2 + \|G'\|_{L^\infty} \| H[\varrho](u, v)\| \\
                     & \leq \left( \|G''\|_{L^\infty} + \frac{1}{\dist_{\min}} \|G'\|_{L^\infty}  \right)  \| J[\sigma](u, v)\|^2 + \|G'\|_{L^\infty} \| B(u, v) \| \\
                     & \lesssim \frac{K(h,\dist_{\min})}{\dist_{\min}} \, | \sigma |_{W^{1,\infty}}^2 + K(h,\dist_{\min}) \, | \sigma |_{W^{2,\infty}} \\
                     & = \frac{K(h,\dist_{\min})}{\dist_{\min}} \left( | \sigma |_{W^{1,\infty}}^2 + \dist_{\min} | \sigma |_{W^{2,\infty}} \right)
              \end{aligned}
       \end{equation*}
       which is the desired bound.
\end{proof}

\begin{lemma} \label{lem:quantization}
       Let $\mu$ be a measure supported on $\Sigma$ and of total mass $|\Sigma|$. Let
       \begin{equation*}
              \begin{aligned}
                     \tau_{i}(\bx) &= c^{-1} \|\bx - \bx_i\|, \\
                     \alpha_{i}(\bx) &= \bp_0[i] (c \tau_{i}(\bx))^{-1}, \\
                     \widetilde{s}(t) & = \sum_{i \in [N]} \int_{\Sigma} \alpha_{i}(\bx) h(t - \tau_{i}(\bx)) \dint \mu(\bx), \\
                     \overline{s}(t) & = \sum_{l \in [UL]} \overline{\bbf}[l] h \star \indic_{[\tau_{l-1}, \tau_l)}(t) \\
                     \overline{\bbf}[l] & = UF_s \sum_{i \in [N]} \int_\Sigma \alpha_{i}(\bx)\indic_{[\tau_{l-1}, \tau_l)}(\tau_{i}(\bx)) \dint \mu(\bx).
              \end{aligned}
       \end{equation*}
       Assume that $\tau_i(\bx) \in [\tau_0, \tau_{UL})$ for all $i \in [N]$ and $\bx \in \Sigma$, as ensured by the padding convention of \Cref{subsec:overview}.
       Then, under \Cref{ass:setting,ass:smoothness_transducer} and $h \in W^{1, \infty}(\R)$
       \begin{equation*}
              |\widetilde{s}(t) - \overline{s}(t)| \lesssim \dist_{\min}^{-1} |\Sigma| \|h'\|_{L^\infty} \|\bp_0\|_1 (UF_s)^{-1}.
       \end{equation*}
\end{lemma}

\begin{proof}
       Plugging the definition of $\overline{\bbf}[l]$ in $\overline{s}$ gives
              \begin{equation*}
                     \begin{aligned}
                            \overline{s}(t) &= \sum_{l \in [UL]} U F_s \sum_{i \in [N]} \int_{\Sigma} \alpha_{i}(\bx) \indic_{[\tau_{l-1}, \tau_l)} (\tau_{i}(\bx)) \dint\mu(\bx) h \star \indic_{[\tau_{l-1}, \tau_l)} (t)
                     \end{aligned}
              \end{equation*}
              Let $l(i,\bx)$ be the index such that $\tau_{i}(\bx) \in [\tau_{l(i,\bx)-1}, \tau_{l(i,\bx)})$, then
              \begin{equation*}
                     \begin{aligned}
                            \overline{s}(t) &= U F_s \sum_{i \in [N]} \int_{\Sigma} \alpha_{i}(\bx) h \star \indic_{[\tau_{l(i,\bx)-1}, \tau_{l(i,\bx)})} (t) \dint\mu(\bx)
                     \end{aligned}
              \end{equation*}
              This allows one to show that
              \begin{equation*}
                     \begin{aligned}
                            \left| \widetilde{s}(t) - \overline{s}(t) \right| &=  \left| \sum_{i \in [N]} \int_{\Sigma} \alpha_{i}(\bx)  \left( h(t - \tau_{i}(\bx)) - UF_s h \star \indic_{[\tau_{l(i,\bx)-1}, \tau_{l(i,\bx)})} (t)\right) \dint\mu(\bx)\right|.
                     \end{aligned}
              \end{equation*}
              Observing that
              \begin{equation*}
                     \begin{aligned}
                           h(t - \tau_{i}(\bx)) - UF_s h \star \indic_{[\tau_{l(i,\bx)-1}, \tau_{l(i,\bx)})} (t) &= h(t - \tau_{i}(\bx)) - UF_s \int_{\tau_{l(i,\bx)-1}}^{\tau_{l(i,\bx)}} h( t - \tau ) \dint \tau \\
                            &= UF_s \int_{\tau_{l(i,\bx)-1}}^{\tau_{l(i,\bx)}} h(t - \tau_{i}(\bx)) - h( t - \tau ) \dint \tau \\
                     \end{aligned}
              \end{equation*}
              allows one to bound, since $h \in W^{1,\infty}(\R)$
              \begin{equation*}
                     \begin{aligned}
                            \left| h(t - \tau_{i}(\bx)) - UF_s h \star \indic_{[\tau_{l(i,\bx)-1}, \tau_{l(i,\bx)})} (t) \right| & \leq \frac{\|h'\|_{L^\infty}}{UF_s},
                     \end{aligned}
              \end{equation*}
              since $\tau_{i}(\bx) \in [\tau_{l(i,\bx)-1}, \tau_{l(i,\bx)})$.
              This shows that
              \begin{equation*}
                     \begin{aligned}
                            \left| \widetilde{s}(t) - \overline{s}(t) \right| &\leq  \sum_{i} \int_\Sigma |\alpha_{i}(\bx)| \dint\mu(\bx) \frac{\|h'\|_{L^\infty}}{UF_s}.
                     \end{aligned}
              \end{equation*}
              Moreover, by definition of $\alpha$
              \begin{equation*}
                     \begin{aligned}
                           \sum_{i} \int_\Sigma |\alpha_{i}(\bx)| \dint\mu(\bx) & \leq \sum_{i} \int_{\Sigma} \left|\frac{\bp_0[i]}{c \tau_{i}(\bx)}\right| \dint \mu(\bx)\\
                            & \leq \frac{1}{\dist_{\min}} \sum_{i} | \bp_0[i] | \int_\Sigma \dint\mu(\bx) \\
                            & = \frac{1}{\dist_{\min}} |\Sigma| \| \bp_0 \|_1
                     \end{aligned}
              \end{equation*}
              giving the desired bound.
\end{proof}

\section{Elliptic integral evaluation of the arc integral}
\label{app:elliptic_integrals}
This appendix establishes the closed form \eqref{eq:cyl_Il} of the arc integral $I_l$ for the cylindrical transducer, with $\psi(u) = \rho^2 + R^2 - 2R\rho\cos u$ the function of its separable form (\Cref{subsec:elliptic_method}).

\begin{lemma} \label{lem:exact_integral}
       Let $b(r) = r^2 - \rho^2 - R^2$, so that $\Psi(u, r) = \sqrt{r^2 - \psi(u)} = \sqrt{2R\rho\cos u + b(r)}$. Let $0 \leq \alpha \leq \beta \leq \pi$ and $r \geq 0$ be such that $\psi(\beta) \leq r^2$. Then $2R\rho + b(r) = r^2 - \psi(0) \geq 0$ and, if $2R\rho + b(r) > 0$, the parameter $\nu(r) = \dfrac{4R\rho}{2R\rho + b(r)}$ satisfies $\nu(r)\sin^2 t \leq 1$ for all $t \in [0, \beta/2]$ and
       \begin{equation} \label{eq:arc_in_E}
              \int_{\alpha}^{\beta} \Psi(u, r)\,\dint u = 2\sqrt{2R\rho + b(r)} \left[ E\!\left(\tfrac{\beta}{2}, \nu(r)\right) - E\!\left(\tfrac{\alpha}{2}, \nu(r)\right) \right].
       \end{equation}
       If $2R\rho + b(r) = 0$, then $\alpha = \beta$ or $\Psi(\cdot, r) \equiv 0$, and the left-hand side of \eqref{eq:arc_in_E} vanishes.
\end{lemma}

\begin{proof}
       Since $\psi$ is non-decreasing on $[0, \pi]$, the hypothesis $\psi(\beta) \leq r^2$ gives $\psi(u) \leq r^2$ for all $u \in [0, \beta]$.
       In particular $\Psi(\cdot, r)$ is well-defined on $[\alpha, \beta]$ and $2R\rho + b(r) = r^2 - \psi(0) \geq 0$.
       In the degenerate case $2R\rho + b(r) = 0$, that is $r^2 = \psi(0)$: if $\rho > 0$, $\psi$ is strictly increasing and $\psi(\beta) \leq \psi(0)$ forces $\alpha = \beta = 0$; if $\rho = 0$, $\psi \equiv R^2$ and $\Psi(\cdot, r) \equiv 0$.

       In both cases the left-hand side of \eqref{eq:arc_in_E} vanishes. Otherwise, $2R\rho + b(r) > 0$ and, for all $t \in [0, \beta/2]$,
       \begin{equation*}
              1 - \nu(r)\sin^2 t = \frac{2R\rho\cos 2t + b(r)}{2R\rho + b(r)} = \frac{r^2 - \psi(2t)}{r^2 - \psi(0)} \geq 0,
       \end{equation*}
       which is the announced bound and makes the elliptic integrals below well-defined. The substitution $u = 2t$ with $\cos 2t = 1 - 2\sin^2 t$ then gives
       \begin{equation*}
              \begin{aligned}
                     \int_{\alpha}^{\beta} \sqrt{2R\rho\cos u + b(r)}\,\dint u
                     &= 2\int_{\alpha/2}^{\beta/2} \sqrt{\,2R\rho + b(r) - 4R\rho\sin^2 t\,}\,\dint t \\
                     &= 2\sqrt{2R\rho + b(r)} \int_{\alpha/2}^{\beta/2} \sqrt{1 - \nu(r)\sin^2 t}\,\dint t,
              \end{aligned}
       \end{equation*}
       which is \eqref{eq:arc_in_E} by definition of $E$ in \eqref{eq:cyl_Il}.
\end{proof}

The hypotheses are automatically satisfied when the lemma is applied in \eqref{eq:cyl_Il} with $r = c\tau_l$, $\alpha = \alpha_l$ and $\beta = \beta_l$, and likewise at the index $l-1$.
Indeed, after the reduction to the positive quadrant, $0 \leq u_{\min} \leq \alpha_l \leq \beta_l \leq u_{\max} \leq \pi$, the last bound following from $\theta_{\max} < \pi/2$ and the modulo of \Cref{subsec:elliptic_method}, and the ordering $\alpha_l \leq \beta_l$ from $v_{\min} \leq v_{\max}$ and the monotonicity of $\widehat{\Phi}$ in $v$.
As for $\psi(\beta_l) \leq (c\tau_l)^2$, either $\alpha_l = \beta_l$ and $I_l = 0$ trivially, or $\alpha_l < \beta_l$, which occurs only for $c\tau_l \in (r_{\min}, r_{\max})$.
In the latter case $\beta_l > u_{\min}$, so $\beta_l$ falls in one of the last two branches of the definition of $\widehat{\Phi}$, and in both $\psi(\beta_l) \leq (c\tau_l)^2 - v_{\min}^2 \leq (c\tau_l)^2$, since $v_{\min} \geq 0$.

Finally, since $2R\rho + b(r) = r^2 - (R-\rho)^2$, the parameter $\nu(r)$ exceeds $1$ if and only if $r < R + \rho$, that is, as long as the wavefront has not reached the far edge of the circle of radius $R$ in the plane of $\bx_i$. This is the regime in which the reciprocal-modulus transformation \eqref{eq:reciprocal_modulus} is required.

\section{Adjoint operators} \label{app:adjoint}

This appendix is dedicated to the adjoint counterparts of \Cref{alg:forward_global,alg:point_based,alg:surface_based}, given in \Cref{alg:adjoint_global,alg:adjoint_point,alg:adjoint_surface}.
\Cref{alg:forward_global} applies three linear stages in sequence (the geometric assembly $\bp_0 \mapsto \bbf_k$, the convolution with $\bh$ and the downsampling $\dsop$), so that its transpose applies the transposes of these stages in reverse order.
Being exact transposes of the implemented forward algorithms, the adjoints satisfy $\langle \bA \bp_0, \bs \rangle = \langle \bp_0, \bA^{*} \bs \rangle$ up to floating point accuracy, as required by the reconstruction algorithms of \Cref{sec:experiments}.

\paragraph{Upsampling}

Transposing the downsampling \eqref{eq:def_downsampling} replicates each time sample over the $U$ upsampled indices it averages, that is
\begin{equation*}
       \left( \dsop^{\top} \bv \right)\left[U(l-1) + l'\right] = \frac{1}{U}\,\bv[l], \qquad \forall l \in [L], \; \forall l' \in [U],
\end{equation*}
at a cost of $O(UL)$ operations, as for $\dsop$.

\paragraph{Correlation}

Transposing the discrete convolution \eqref{eq:discrete_conv_2} gives the correlation with the same kernel, $\bv \mapsto \bh^{-} \star \bv$ with $\bh^{-}[l] = \bh[-l]$.
Since $\bh$ is real-valued, the reversal amounts to a conjugation of its spectrum, so that the correlation is computed by the same batched FFT as the forward convolution, at the same cost of $O(UL \log(UL))$ operations per transducer.

\paragraph{Geometric assembly}

Transposing \eqref{eq:def_f_point} and \eqref{eq:ui_area} turns the scatter of the voxel contributions into a gather over the time samples, which is the only line of the inner loop by which \Cref{alg:adjoint_point,alg:adjoint_surface} differ from \Cref{alg:point_based,alg:surface_based}.
The padding convention of \Cref{subsec:overview} ensures that every index $l$ falls in $[UL]$, so that the scatter and the gather address the same entries.
The geometric quantities (the distances $\| \bq_j - \bx_i \|$, the time indices, the bounds $l_{\min}(i), l_{\max}(i)$ of \eqref{eq:rmin_rmax} and the areas $\left| \mathcal{M}_i^l \right|$ of \eqref{eq:def:bffi}) do not depend on $\bp_0$: they are computed by the same routines in the forward and in the adjoint pass, so that \Cref{alg:separable_area} is reused unchanged and requires no adjoint counterpart.
Similarly, the Piecewise-Planes operator \eqref{eq:plane_decomp_u} sums the contributions of the $N_p$ tangent planes in the adjoint as in the forward.

\paragraph{Coarser area grid}

When the strategy of \Cref{rmk:coarser_area_grid} is used, as in \Cref{sec:experiments}, with $\Delta l = U$, the forward pass spreads a single area $\left| \mathcal{M}_i^{l-\Delta l+1 : l} \right|$ evenly over the $\Delta l$ time intervals it covers. The inner loop of \Cref{alg:adjoint_surface} gathers them back accordingly,
\begin{equation*}
       \bz[i] \gets \bz[i] + \frac{\left| \mathcal{M}_i^{l-\Delta l+1 : l} \right|}{\Delta l} \sum_{l' = l-\Delta l+1}^{l} \frac{\bv[l']}{c\,\tau_{l'-\frac{1}{2}}},
\end{equation*}
with the same reduction in the number of area evaluations as in the forward.

\paragraph{Summary}

The adjoint has the complexity and the memory footprint of the forward operator, and differs from it only in its access pattern: the forward accumulates each voxel contribution into the shared vector $\bbf$, whereas the adjoint accumulates into the single entry $\bz[i]$ of its output $\bz = \bA^{*}\bs$, so that its loops over the voxels are free of write conflicts; only the sum over the $K$ transducers in \Cref{alg:adjoint_global} requires a reduction.

\begin{algorithm}[ht]
\caption{Adjoint operator $\bA^{*} : \bs \mapsto \bz$.}
\label{alg:adjoint_global}
\begin{algorithmic}[1]
       \Require $\bs \in \R^{KL}$ ; $(\bx_i)_{i \in [N]}$ ; $(\Sigma_k)_{k \in [K]}$ ; the same precomputed $\bh$ as in \Cref{alg:forward_global} ;
       \Ensure $\bz = \bA^{*} \bs \in \R^{N}$.
       \Statex
       \State \textbf{Initialization}: $\bz \gets 0 \in \R^{N}$
       \For{$k = 1, \ldots, K$} \Comment{Loop over transducers}
              \State $\bv_k \gets \dsop^{\top} \bs_k \in \R^{UL}$. \Comment{Upsampling}
              \State $\bv_k \gets \bh^{-} \star \bv_k$. \Comment{FFT correlation}
              \State Compute $\bz_k \in \R^N$ from $\bv_k$. \Comment{Method-specific: \Cref{alg:adjoint_point,alg:adjoint_surface}.}
              \State $\bz \gets \bz + \bz_k$.
       \EndFor
       \State \Return $\bz$.
\end{algorithmic}
\end{algorithm}

\begin{algorithm}[ht]
\caption{Adjoint of the point detector approximation for one transducer $\Sigma$.}
\label{alg:adjoint_point}
\begin{algorithmic}[1]
       \Require $\bv \in \R^{UL}$ ; $(\bx_i)_{i \in [N]}$ ; quadrature $(\bq_j, \Delta \bq_j)_{j \in [Q]}$ of $\Sigma$;
       \Ensure $\bz \in \R^{N}$, the transpose of \Cref{alg:point_based} applied to $\bv$.
       \Statex
       \State \textbf{Initialization}: $\bz \gets 0 \in \R^{N}$
       \For{$i \in [N]$} \Comment{Loop over voxels}
              \For{$j \in [Q]$} \Comment{Loop over point detectors of $\Sigma$}
                     \State $r \gets \|\bq_j - \bx_i\|$.
                     \State $l \gets \lfloor (r/c - t_0)\, U F_s \rfloor + 1$.
                     \State $\bz[i] \gets \bz[i] + \Delta \bq_j\, \bv[l] / r$. \Comment{Gather instead of scatter}
              \EndFor
       \EndFor
       \State $\bz \gets U F_s \cdot \bz$.
       \State \Return $\bz$.
\end{algorithmic}
\end{algorithm}

\begin{algorithm}[ht]
\caption{Adjoint of the surface approximation for one transducer $\Sigma$.}
\label{alg:adjoint_surface}
\begin{algorithmic}[1]
       \Require $\bv \in \R^{UL}$ ; $(\bx_i)_{i \in [N]}$ ; transducer surface $\Sigma$
       \Ensure $\bz \in \R^{N}$, the transpose of \Cref{alg:surface_based} applied to $\bv$.
       \State \textbf{Initialization:} $\bz \gets 0 \in \R^{N}$.
       \For{$i \in [N]$} \Comment{Loop over voxels}
       \State Compute $l_{\min}(i)$ and $l_{\max}(i)$ as in \eqref{eq:rmin_rmax}
              \For{$l = l_{\min}(i), \ldots, l_{\max}(i)$} \Comment{Loop over time steps}
                     \State Compute $\left| \mathcal{M}_i^l \right|$ in \eqref{eq:def:bffi} \Comment{\Cref{alg:separable_area}, unchanged}
                     \State $\bz[i] \gets \bz[i] + \left| \mathcal{M}_i^l \right| \bv[l] / (c\,\tau_{l-\frac{1}{2}})$.
              \EndFor
       \EndFor
       \State $\bz \gets U F_s \cdot \bz$.
       \State \Return $\bz$.
\end{algorithmic}
\end{algorithm}

\bibliographystyle{unsrt}
\bibliography{biblio}

\begin{thebibliography}{100}

\bibitem{beard2011biomedical}
Paul Beard.
\newblock Biomedical photoacoustic imaging.
\newblock {\em Interface Focus}, 1(4):602--631, 2011.

\bibitem{wang2012photoacoustic}
Lihong~V Wang and Song Hu.
\newblock Photoacoustic tomography: in vivo imaging from organelles to organs.
\newblock {\em Science}, 335(6075):1458--1462, 2012.

\bibitem{xia2014photoacoustic}
Jun Xia, Junjie Yao, and Lihong~V Wang.
\newblock Photoacoustic tomography: principles and advances.
\newblock {\em Progress In Electromagnetics Research}, 147:1--22, 2014.

\bibitem{yao2014photoacoustic}
Junjie Yao and Lihong~V Wang.
\newblock Photoacoustic brain imaging: from microscopic to macroscopic scales.
\newblock {\em Neurophotonics}, 1(1):011003--011003, 2014.

\bibitem{ntziachristos2010molecular}
Vasilis Ntziachristos and Daniel Razansky.
\newblock Molecular imaging by means of multispectral optoacoustic tomography
  ({MSOT}).
\newblock {\em Chemical Reviews}, 110(5):2783--2794, 2010.

\bibitem{perez2025dual}
Mailyn P{\'e}rez-Liva, Mar{\'\i}a Alonso~de Leci{\~n}ana, Mar{\'\i}a
  Guti{\'e}rrez-Fern{\'a}ndez, Jorge Camacho Sosa~Dias, Jorge F~Cruza, Jorge
  Rodr{\'\i}guez-Pardo, Iv{\'a}n Garc{\'\i}a-Su{\'a}rez, Fernando
  Laso-Garc{\'\i}a, Joaquin~L Herraiz, and Luis Elvira~Segura.
\newblock Dual photoacoustic/ultrasound technologies for preclinical research:
  current status and future trends.
\newblock {\em Physics in Medicine \& Biology}, 70(7):07TR01, 2025.

\bibitem{park2025clinical}
Jeongwoo Park, Seongwook Choi, Ferdinand Knieling, Bryan Clingman, Sarah
  Bohndiek, Lihong~V Wang, and Chulhong Kim.
\newblock Clinical translation of photoacoustic imaging.
\newblock {\em Nature Reviews Bioengineering}, 3(3):193--212, 2025.

\bibitem{attia2019review}
Amalina Binte~Ebrahim Attia, Ghayathri Balasundaram, Mohesh Moothanchery,
  US~Dinish, Renzhe Bi, Vasilis Ntziachristos, and Malini Olivo.
\newblock A review of clinical photoacoustic imaging: Current and future
  trends.
\newblock {\em Photoacoustics}, 16:100144, 2019.

\bibitem{wang2017photoacoustic}
Lihong Wang.
\newblock {\em Photoacoustic Imaging and Spectroscopy}.
\newblock CRC Press, 2017.

\bibitem{poudel2019survey}
Joemini Poudel, Yang Lou, and Mark~A Anastasio.
\newblock A survey of computational frameworks for solving the acoustic inverse
  problem in three-dimensional photoacoustic computed tomography.
\newblock {\em Physics in Medicine \& Biology}, 64(14):14TR01, 2019.

\bibitem{anastasio2007application}
Mark~A Anastasio, Jin Zhang, Dimple Modgil, and Patrick J~La Rivi{\`e}re.
\newblock Application of inverse source concepts to photoacoustic tomography.
\newblock {\em Inverse Problems}, 23(6):S21--S35, 2007.

\bibitem{kuchment2008mathematics}
Peter Kuchment and Leonid Kunyansky.
\newblock Mathematics of thermoacoustic tomography.
\newblock {\em European Journal of Applied Mathematics}, 19(2):191--224, 2008.

\bibitem{arridge2016adjoint}
Simon~R Arridge, Marta~M Betcke, Ben~T Cox, Felix Lucka, and Brad~E Treeby.
\newblock On the adjoint operator in photoacoustic tomography.
\newblock {\em Inverse Problems}, 32(11):115012, 2016.

\bibitem{Wang2011}
Kun Wang, Sergey~A. Ermilov, Richard Su, Hans-Peter Brecht, Alexander~A.
  Oraevsky, and Mark~A. Anastasio.
\newblock An imaging model incorporating ultrasonic transducer properties for
  three-dimensional optoacoustic tomography.
\newblock {\em IEEE Transactions on Medical Imaging}, 30(2):203--214, 2011.

\bibitem{rosenthal2011model}
Amir Rosenthal, Daniel Razansky, and Vasilis Ntziachristos.
\newblock Model-based optoacoustic inversion with arbitrary-shape detectors.
\newblock {\em Medical Physics}, 38(7):4285--4295, 2011.

\bibitem{mitsuhashi2014investigation}
Kenji Mitsuhashi, Kun Wang, and Mark~A Anastasio.
\newblock Investigation of the far-field approximation for modeling a
  transducer's spatial impulse response in photoacoustic computed tomography.
\newblock {\em Photoacoustics}, 2(1):21--32, 2014.

\bibitem{queiros2014modeling}
Daniel Queir{\'o}s, Xos{\'e}~Lu{\'\i}s D{\'e}an-Ben, Andreas Buehler, Daniel
  Razansky, Amir Rosenthal, and Vasilis Ntziachristos.
\newblock Modeling the shape of cylindrically focused transducers in
  three-dimensional optoacoustic tomography.
\newblock In {\em Photons Plus Ultrasound: Imaging and Sensing 2014}, volume
  8943, pages 335--341. SPIE, 2014.

\bibitem{roitner2014deblurring}
Heinz Roitner, Markus Haltmeier, Robert Nuster, Dianne~P O’Leary, Thomas
  Berer, Guenther Paltauf, Hubert Gr{\"u}n, and Peter Burgholzer.
\newblock Deblurring algorithms accounting for the finite detector size in
  photoacoustic tomography.
\newblock {\em Journal of Biomedical Optics}, 19(5):056011--056011, 2014.

\bibitem{piwakowski1989method}
Bogdan Piwakowski and Bertrand Delannoy.
\newblock Method for computing spatial pulse response: Time-domain approach.
\newblock {\em The Journal of the Acoustical Society of America},
  86(6):2422--2432, 1989.

\bibitem{piwakowski1999new}
Bogdan Piwakowski and Khalid Sbai.
\newblock A new approach to calculate the field radiated from arbitrarily
  structured transducer arrays.
\newblock {\em IEEE Transactions on Ultrasonics, Ferroelectrics, and Frequency
  Control}, 46(2):422--440, 1999.

\bibitem{ding2017efficient}
Lu~Ding, Xose~Luis Dean-Ben, and Daniel Razansky.
\newblock Efficient {3-D} model-based reconstruction scheme for arbitrary
  optoacoustic acquisition geometries.
\newblock {\em IEEE Transactions on Medical Imaging}, 36(9):1858--1867, 2017.

\bibitem{wise2019representing}
Elliott~S. Wise, Ben~T. Cox, Jiri Jaros, and Bradley~E. Treeby.
\newblock Representing arbitrary acoustic source and sensor distributions in
  {F}ourier collocation methods.
\newblock {\em The Journal of the Acoustical Society of America},
  146(1):278--288, 2019.

\bibitem{lockwood1973high}
JC~Lockwood and JG~Willette.
\newblock High-speed method for computing the exact solution for the pressure
  variations in the nearfield of a baffled piston.
\newblock {\em The Journal of the Acoustical Society of America},
  53(3):735--741, 1973.

\bibitem{san1992diffraction}
Jose~Luis San~Emeterio and Luis~G Ullate.
\newblock Diffraction impulse response of rectangular transducers.
\newblock {\em The Journal of the Acoustical Society of America},
  92(2):651--662, 1992.

\bibitem{schmerr2007ultrasonic}
Lester~W Schmerr~Jr and Sung-Jin Song.
\newblock {\em Ultrasonic Nondestructive Evaluation Systems: Models and
  Measurements}.
\newblock Springer, 2007.

\bibitem{hunt1983ultrasound}
John~W. Hunt, Marcel Arditi, and F.~Stuart Foster.
\newblock Ultrasound transducers for pulse-echo medical imaging.
\newblock {\em IEEE Transactions on Biomedical Engineering},
  BME-30(8):453--481, 1983.

\bibitem{baek2012spatial}
David~B{\"o}ttcher B{\ae}k, J{\o}rgen~Arendt Jensen, and Morten Willatzen.
\newblock Spatial impulse response of a rectangular double curved transducer.
\newblock {\em The Journal of the Acoustical Society of America},
  131(4):2730--2741, 2012.

\bibitem{jensen1992calculation}
J{\o}rgen~Arendt Jensen and Niels~Bruun Svendsen.
\newblock Calculation of pressure fields from arbitrarily shaped, apodized, and
  excited ultrasound transducers.
\newblock {\em IEEE Transactions on Ultrasonics, Ferroelectrics, and Frequency
  Control}, 39(2):262--267, 1992.

\bibitem{jensen1996field}
J{\o}rgen~Arendt Jensen.
\newblock {FIELD}: A program for simulating ultrasound systems.
\newblock {\em Medical \& Biological Engineering \& Computing}, 34(Suppl. 1,
  Pt. 1):351--353, 1996.
\newblock 10th Nordic-Baltic Conference on Biomedical Imaging.

\bibitem{wu1999spatial}
Ping Wu and Tadeusz Stepinski.
\newblock Spatial impulse response method for predicting pulse-echo fields from
  a linear array with cylindrically concave surface.
\newblock {\em IEEE Transactions on Ultrasonics, Ferroelectrics, and Frequency
  Control}, 46(5):1283--1297, 1999.

\bibitem{dean2022practical}
Xos{\'e}~Lu{\'\i}s De{\'a}n-Ben and Daniel Razansky.
\newblock A practical guide for model-based reconstruction in optoacoustic
  imaging.
\newblock {\em Frontiers in Physics}, 10:1028258, 2022.

\bibitem{do2025reconstruction}
Trung-Thai Do, Caroline Chaux, Paul Escande, and J{\'e}r{\^o}me G{\^a}teau.
\newblock Reconstruction d'images en tomographie photoacoustique avec
  r{\'e}gularisation combin{\'e}e variation totale-{C}auchy.
\newblock In {\em GRETSI 2025}, 2025.

\bibitem{paltauf2002iterative}
Guenther Paltauf, JA~Viator, SA~Prahl, and SL~Jacques.
\newblock Iterative reconstruction algorithm for optoacoustic imaging.
\newblock {\em The Journal of the Acoustical Society of America},
  112(4):1536--1544, 2002.

\bibitem{xu2006photoacoustic}
Minghua Xu and Lihong~V Wang.
\newblock Photoacoustic imaging in biomedicine.
\newblock {\em Review of Scientific Instruments}, 77(041101), 2006.

\bibitem{xu2005universal}
Minghua Xu and Lihong~V. Wang.
\newblock {Universal back-projection algorithm for photoacoustic computed
  tomography}.
\newblock In Alexander~A. Oraevsky and Lihong~V. Wang, editors, {\em Photons
  Plus Ultrasound: Imaging and Sensing 2005: The Sixth Conference on Biomedical
  Thermoacoustics, Optoacoustics, and Acousto-optics}, volume 5697, pages 251
  -- 254. International Society for Optics and Photonics, SPIE, 2005.

\bibitem{burgholzer2007exact}
Peter Burgholzer, Gebhard~J. Matt, Markus Haltmeier, and G\"unther Paltauf.
\newblock Exact and approximative imaging methods for photoacoustic tomography
  using an arbitrary detection surface.
\newblock {\em Physical Review E}, 75:046706, Apr 2007.

\bibitem{finch2004determining}
David Finch and Sarah~K Patch.
\newblock Determining a function from its mean values over a family of spheres.
\newblock {\em SIAM Journal on Mathematical Analysis}, 35(5):1213--1240, 2004.

\bibitem{kunyansky2007explicit}
Leonid~A Kunyansky.
\newblock Explicit inversion formulae for the spherical mean {R}adon transform.
\newblock {\em Inverse Problems}, 23(1):373--383, 2007.

\bibitem{haltmeier2014universal}
Markus Haltmeier.
\newblock Universal inversion formulas for recovering a function from spherical
  means.
\newblock {\em SIAM Journal on Mathematical Analysis}, 46(1):214--232, 2014.

\bibitem{paltauf2007experimental}
Guenther Paltauf, Robert Nuster, Markus Haltmeier, and Peter Burgholzer.
\newblock Experimental evaluation of reconstruction algorithms for limited view
  photoacoustic tomography with line detectors.
\newblock {\em Inverse Problems}, 23(6):S81--S94, 2007.

\bibitem{xu2004reconstructions}
Yuan Xu, Lihong~V Wang, Gaik Ambartsoumian, and Peter Kuchment.
\newblock Reconstructions in limited-view thermoacoustic tomography.
\newblock {\em Medical Physics}, 31(4):724--733, 2004.

\bibitem{sheu2008simulations}
Yae-Lin Sheu and Pai-Chi Li.
\newblock Simulations of photoacoustic wave propagation using a
  finite-difference time-domain method with {B}erenger’s perfectly matched
  layers.
\newblock {\em The Journal of the Acoustical Society of America},
  124(6):3471--3480, 2008.

\bibitem{huang2013full}
Chao Huang, Kun Wang, Liming Nie, Lihong~V Wang, and Mark~A Anastasio.
\newblock Full-wave iterative image reconstruction in photoacoustic tomography
  with acoustically inhomogeneous media.
\newblock {\em IEEE Transactions on Medical Imaging}, 32(6):1097--1110, 2013.

\bibitem{mitsuhashi2017forward}
Kenji Mitsuhashi, Joemini Poudel, Thomas~P Matthews, Alejandro Garcia-Uribe,
  Lihong~V Wang, and Mark~A Anastasio.
\newblock A forward-adjoint operator pair based on the elastic wave equation
  for use in transcranial photoacoustic computed tomography.
\newblock {\em SIAM Journal on Imaging Sciences}, 10(4):2022--2048, 2017.

\bibitem{Yuan2007threeD}
Zhen Yuan and Huabei Jiang.
\newblock Three-dimensional finite-element-based photoacoustic tomography:
  Reconstruction algorithm and simulations.
\newblock {\em Medical Physics}, 34(2):538--546, 2007.

\bibitem{Yao2009finite}
Lei Yao and Huabei Jiang.
\newblock Finite-element-based photoacoustic tomography in time domain.
\newblock {\em Journal of Optics A: Pure and Applied Optics}, 11(8):085301, May
  2009.

\bibitem{mast2001k}
T~Douglas Mast, Laurent~P Souriau, D-LD Liu, Makoto Tabei, Adrian~I Nachman,
  and Robert~C Waag.
\newblock A k-space method for large-scale models of wave propagation in
  tissue.
\newblock {\em IEEE Transactions on Ultrasonics, Ferroelectrics, and Frequency
  Control}, 48(2):341--354, 2001.

\bibitem{cox2007k}
Benjamin~T Cox, S~Kara, Simon~R Arridge, and Paul~C Beard.
\newblock k-space propagation models for acoustically heterogeneous media:
  Application to biomedical photoacoustics.
\newblock {\em The Journal of the Acoustical Society of America},
  121(6):3453--3464, 2007.

\bibitem{treeby2010kwave}
Bradley~E. Treeby and Benjamin~T. Cox.
\newblock {k-Wave}: {MATLAB} toolbox for the simulation and reconstruction of
  photoacoustic wave fields.
\newblock {\em Journal of Biomedical Optics}, 15(2):021314, 2010.

\bibitem{treeby2012modeling}
Bradley~E Treeby, Jiri Jaros, Alistair~P Rendell, and Benjamin~T Cox.
\newblock Modeling nonlinear ultrasound propagation in heterogeneous media with
  power law absorption using a k-space pseudospectral method.
\newblock {\em The Journal of the Acoustical Society of America},
  131(6):4324--4336, 2012.

\bibitem{else2024patato}
Thomas~R. Else, Janek Gröhl, Lina Hacker, and Sarah~E. Bohndiek.
\newblock {PATATO}: a {Python} photoacoustic tomography analysis toolkit.
\newblock {\em Journal of Open Source Software}, 9(93):5686, 2024.

\bibitem{stanziola2023j}
Antonio Stanziola, Simon~R Arridge, Ben~T Cox, and Bradley~E Treeby.
\newblock {j-Wave}: An open-source differentiable wave simulator.
\newblock {\em SoftwareX}, 22:101338, 2023.

\bibitem{hristova2008reconstruction}
Yulia Hristova, Peter Kuchment, and Linh Nguyen.
\newblock Reconstruction and time reversal in thermoacoustic tomography in
  acoustically homogeneous and inhomogeneous media.
\newblock {\em Inverse Problems}, 24(5):055006, 2008.

\bibitem{li2026slingbag}
Shuang Li, Yibing Wang, Jian Gao, Chulhong Kim, Seongwook Choi, Yu~Zhang, Qian
  Chen, Yao Yao, and Changhui Li.
\newblock {SlingBAG}: point cloud-based iterative algorithm for large-scale
  {3D} photoacoustic imaging.
\newblock {\em Nature Communications}, 17(1):128, 2026.

\bibitem{wang2026gpair}
Yibing Wang, Shuang Li, Tingting Huang, Yu~Zhang, Chulhong Kim, Seongwook Choi,
  and Changhui Li.
\newblock {GPAIR}: Gaussian-kernel-based ultrafast {3D} photoacoustic iterative
  reconstruction.
\newblock {\em arXiv preprint arXiv:2602.03893}, 2026.

\bibitem{jaros2016fullwave}
Jiri Jaros, Alistair~P. Rendell, and Bradley~E. Treeby.
\newblock Full-wave nonlinear ultrasound simulation on distributed clusters
  with applications in high-intensity focused ultrasound.
\newblock {\em The International Journal of High Performance Computing
  Applications}, 30(2):137--155, 2016.

\bibitem{rosenthal2010fast}
Amir Rosenthal, Daniel Razansky, and Vasilis Ntziachristos.
\newblock Fast semi-analytical model-based acoustic inversion for quantitative
  optoacoustic tomography.
\newblock {\em IEEE Transactions on Medical Imaging}, 29(6):1275--1285, 2010.

\bibitem{buehler2011model}
Andreas Buehler, Amir Rosenthal, Thomas Jetzfellner, Alexander Dima, Daniel
  Razansky, and Vasilis Ntziachristos.
\newblock Model-based optoacoustic inversions with incomplete projection data.
\newblock {\em Medical Physics}, 38(3):1694--1704, 2011.

\bibitem{caballero2013model}
Miguel Angel~Araque Caballero, J{\'e}r{\^o}me Gateau, Xos{\'e}-Luis
  D{\'e}an-Ben, and Vasilis Ntziachristos.
\newblock Model-based optoacoustic image reconstruction of large
  three-dimensional tomographic datasets acquired with an array of directional
  detectors.
\newblock {\em IEEE Transactions on Medical Imaging}, 33(2):433--443, 2013.

\bibitem{Hauptmann2018}
Andreas Hauptmann, Felix Lucka, Marta Betcke, Nam Huynh, Jonas Adler, Ben Cox,
  Paul Beard, Sebastien Ourselin, and Simon Arridge.
\newblock Model-based learning for accelerated, limited-view {3-D}
  photoacoustic tomography.
\newblock {\em IEEE Transactions on Medical Imaging}, 37(6):1382--1393, 2018.

\bibitem{ding2020model}
Lu~Ding, Daniel Razansky, and Xosé~Luís Deán-Ben.
\newblock Model-based reconstruction of large three-dimensional optoacoustic
  datasets.
\newblock {\em IEEE Transactions on Medical Imaging}, 39(9):2931--2940, 2020.

\bibitem{li2026slingbagpro}
Shuang Li, Yibing Wang, Jian Gao, Chulhong Kim, Seongwook Choi, Yu~Zhang, Qian
  Chen, Yao Yao, and Changhui Li.
\newblock {SlingBAG Pro}: Accelerating point cloud-based iterative
  reconstruction for {3D} photoacoustic imaging with arbitrary array
  geometries.
\newblock {\em arXiv preprint arXiv:2601.00551}, 2026.

\bibitem{diebold1991photoacoustic}
Gerald~J Diebold, T~Sun, and MI~Khan.
\newblock Photoacoustic monopole radiation in one, two, and three dimensions.
\newblock {\em Physical Review Letters}, 67(24):3384, 1991.

\bibitem{Lewitt1990attenuation}
Robert~M. Lewitt.
\newblock Multidimensional digital image representations using generalized
  {K}aiser--{B}essel window functions.
\newblock {\em Journal of the Optical Society of America A}, 7(10):1834--1846,
  Oct 1990.

\bibitem{wang2014discrete}
Kun Wang, Robert~W Schoonover, Richard Su, Alexander Oraevsky, and Mark~A
  Anastasio.
\newblock Discrete imaging models for three-dimensional optoacoustic tomography
  using radially symmetric expansion functions.
\newblock {\em IEEE Transactions on Medical Imaging}, 33(5):1180--1193, 2014.

\bibitem{kansa1990multiquadrics}
Edward~J Kansa.
\newblock Multiquadrics—a scattered data approximation scheme with
  applications to computational fluid-dynamics—{II} solutions to parabolic,
  hyperbolic and elliptic partial differential equations.
\newblock {\em Computers \& Mathematics with Applications}, 19(8-9):147--161,
  1990.

\bibitem{franke1998solving}
Carsten Franke and Robert Schaback.
\newblock Solving partial differential equations by collocation using radial
  basis functions.
\newblock {\em Applied Mathematics and Computation}, 93(1):73--82, 1998.

\bibitem{kostli2001temporal}
K.~P. K{\"o}stli, M.~Frenz, H.~Bebie, and H.~P. Weber.
\newblock Temporal backward projection of optoacoustic pressure transients
  using {Fourier} transform methods.
\newblock {\em Physics in Medicine and Biology}, 46:1863--1872, 2001.

\bibitem{cox2005fast}
B.~T. Cox and P.~C. Beard.
\newblock Fast calculation of pulsed photoacoustic fields in fluids using
  k-space methods.
\newblock {\em The Journal of the Acoustical Society of America},
  117(6):3616--3627, 2005.

\bibitem{hauptmann2018approximate}
Andreas Hauptmann, Ben Cox, Felix Lucka, Nam Huynh, Marta Betcke, Paul Beard,
  and Simon Arridge.
\newblock Approximate k-space models and deep learning for fast photoacoustic
  reconstruction.
\newblock In {\em Machine Learning for Medical Image Reconstruction (MLMIR
  2018)}, volume 11074 of {\em Lecture Notes in Computer Science}, pages
  103--111. Springer, 2018.

\bibitem{kucukkomurcu2026depth}
Ege K{\"u}{\c{c}}{\"u}kk{\"o}m{\"u}rc{\"u}, Simon Labouesse, Marc Allain, and
  Thomas Chaigne.
\newblock A depth-dependent, transverse shift-invariant operator for fast
  iterative {3D} photoacoustic tomography in planar geometry.
\newblock {\em arXiv preprint arXiv:2603.28150}, 2026.

\bibitem{antholzer2019deep}
Stephan Antholzer, Markus Haltmeier, and Johannes Schwab.
\newblock Deep learning for photoacoustic tomography from sparse data.
\newblock {\em Inverse Problems in Science and Engineering}, 27(7):987--1005,
  2019.

\bibitem{davoudi2019deep}
Neda Davoudi, Xos{\'e}~Lu{\'\i}s De{\'a}n-Ben, and Daniel Razansky.
\newblock Deep learning optoacoustic tomography with sparse data.
\newblock {\em Nature Machine Intelligence}, 1(10):453--460, 2019.

\bibitem{schwab2019learned}
Johannes Schwab, Stephan Antholzer, and Markus Haltmeier.
\newblock Learned backprojection for sparse and limited view photoacoustic
  tomography.
\newblock In {\em Photons Plus Ultrasound: Imaging and Sensing 2019}, volume
  10878, pages 263--271. SPIE, 2019.

\bibitem{allman2018photoacoustic}
Derek Allman, Austin Reiter, and Muyinatu A~Lediju Bell.
\newblock Photoacoustic source detection and reflection artifact removal
  enabled by deep learning.
\newblock {\em IEEE Transactions on Medical Imaging}, 37(6):1464--1477, 2018.

\bibitem{lan2020net}
Hengrong Lan, Daohuai Jiang, Changchun Yang, Feng Gao, and Fei Gao.
\newblock {Y-Net}: Hybrid deep learning image reconstruction for photoacoustic
  tomography in vivo.
\newblock {\em Photoacoustics}, 20:100197, 2020.

\bibitem{hauptmann2020deep}
Andreas Hauptmann and Ben Cox.
\newblock Deep learning in photoacoustic tomography: current approaches and
  future directions.
\newblock {\em Journal of Biomedical Optics}, 25(11):112903, 2020.

\bibitem{song2024solving}
Bowen Song, Soo~Min Kwon, Zecheng Zhang, Xinyu Hu, Qing Qu, and Liyue Shen.
\newblock Solving inverse problems with latent diffusion models via hard data
  consistency.
\newblock In {\em International Conference on Learning Representations}, volume
  2024, pages 7624--7654, 2024.

\bibitem{dey2024score}
Sreemanti Dey, Snigdha Saha, Berthy~T Feng, Manxiu Cui, Laure Delisle, Oscar
  Leong, Lihong~V Wang, and Katherine~L Bouman.
\newblock Score-based diffusion models for photoacoustic tomography image
  reconstruction.
\newblock In {\em ICASSP 2024-2024 IEEE International Conference on Acoustics,
  Speech and Signal Processing (ICASSP)}, pages 2470--2474. IEEE, 2024.

\bibitem{yang2021review}
Changchun Yang, Hengrong Lan, Feng Gao, and Fei Gao.
\newblock Review of deep learning for photoacoustic imaging.
\newblock {\em Photoacoustics}, 21:100215, 2021.

\bibitem{grohl2021deep}
Janek Gr{\"o}hl, Melanie Schellenberg, Kris Dreher, and Lena Maier-Hein.
\newblock Deep learning for biomedical photoacoustic imaging: A review.
\newblock {\em Photoacoustics}, 22:100241, 2021.

\bibitem{deng2021deep}
Handi Deng, Hui Qiao, Qionghai Dai, and Cheng Ma.
\newblock Deep learning in photoacoustic imaging: a review.
\newblock {\em Journal of Biomedical Optics}, 26(4):040901--040901, 2021.

\bibitem{bridal2024innovative}
Théotim Lucas and Jérôme Gateau.
\newblock Tomography and spectroscopy: photoacoustics.
\newblock In {\em Innovative Ultrasound Imaging Techniques: Biomedical
  Applications}, chapter~9. John Wiley \& Sons, 2024.

\bibitem{evans2022partial}
Lawrence~C Evans.
\newblock {\em Partial Differential Equations}, volume~19.
\newblock American Mathematical Society, 2022.

\bibitem{oppenheim2010discrete}
Alan~V Oppenheim and Ronald~W Schafer.
\newblock {\em Discrete-Time Signal Processing}.
\newblock Pearson, 2010.

\bibitem{arnau2004piezoelectric}
Antonio Arnau and David Soares.
\newblock {\em Piezoelectric Transducers and Applications}.
\newblock Springer, 2004.

\bibitem{Wang_2013}
Kun Wang, Chao Huang, Yu-Jiun Kao, Cheng-Ying Chou, Alexander~A. Oraevsky, and
  Mark~A. Anastasio.
\newblock Accelerating image reconstruction in three-dimensional optoacoustic
  tomography on graphics processing units.
\newblock {\em Medical Physics}, 40(2):023301, January 2013.

\bibitem{linger2023volumetric}
Cl{\'e}ment Linger, Yoann Atlas, Remy Winter, Marine Vandebrouck, Maxime Faure,
  Th{\'e}otim Lucas, S~Lori Bridal, and J{\'e}r{\^o}me Gateau.
\newblock Volumetric and simultaneous photoacoustic and ultrasound imaging with
  a conventional linear array in a multiview scanning scheme.
\newblock {\em IEEE Transactions on Ultrasonics, Ferroelectrics, and Frequency
  Control}, 70(12):1607--1620, 2023.

\bibitem{wendland2005scattered}
Holger Wendland.
\newblock {\em Scattered Data Approximation}.
\newblock Cambridge University Press, 2005.

\bibitem{wendland1995piecewise}
Holger Wendland.
\newblock Piecewise polynomial, positive definite and compactly supported
  radial functions of minimal degree.
\newblock {\em Advances in Computational Mathematics}, 4:389--396, 1995.

\bibitem{narcowich2006sobolev}
Francis~J. Narcowich, Joseph~D. Ward, and Holger Wendland.
\newblock Sobolev error estimates and a {B}ernstein inequality for scattered
  data interpolation via radial basis functions.
\newblock {\em Constructive Approximation}, 24:175--186, 2006.

\bibitem{buhmann2003radial}
Martin~D. Buhmann.
\newblock {\em Radial Basis Functions: Theory and Implementations}.
\newblock Cambridge University Press, 2003.

\bibitem{schaback2006kernel}
Robert Schaback and Holger Wendland.
\newblock Kernel techniques: from machine learning to meshless methods.
\newblock {\em Acta Numerica}, 15:543--639, 2006.

\bibitem{chacko2013discretization}
Nikhil Chacko, Michael Liebling, and Thierry Blu.
\newblock Discretization of continuous convolution operators for accurate
  modeling of wave propagation in digital holography.
\newblock {\em Journal of the Optical Society of America A}, 30(10):2012--2020,
  2013.

\bibitem{stepanishen1971transient}
Peter~R. Stepanishen.
\newblock Transient radiation from pistons in an infinite planar baffle.
\newblock {\em The Journal of the Acoustical Society of America},
  49(5B):1629--1638, 1971.

\bibitem{NIST:DLMF}
{\it NIST Digital Library of Mathematical Functions}.
\newblock \url{https://dlmf.nist.gov/}, Release 1.2.7 of 2026-06-15.
\newblock F.~W.~J. Olver, A.~B. {Olde Daalhuis}, D.~W. Lozier, B.~I. Schneider,
  R.~F. Boisvert, C.~W. Clark, B.~R. Miller, B.~V. Saunders, H.~S. Cohl, and
  M.~A. McClain, eds.

\bibitem{carlson1995numerical}
B.~C. Carlson.
\newblock Numerical computation of real or complex elliptic integrals.
\newblock {\em Numerical Algorithms}, 10(1):13–26, March 1995.

\bibitem{abramowitz1966handbook}
Milton Abramowitz and Irene~A Stegun.
\newblock {\em {Handbook of Mathematical Functions, with Formulas, Graphs, and
  Mathematical Tables}}.
\newblock American Institute of Physics, 1966.

\bibitem{ilyina2017extension}
Natalia Ilyina, Jeroen Hermans, Koen Van Den~Abeele, and Jan D'hooge.
\newblock Extension of the angular spectrum method to model the pressure field
  of a cylindrically curved array transducer.
\newblock {\em The Journal of the Acoustical Society of America},
  141(3):EL262--EL266, 2017.

\bibitem{byrd1995limited}
Richard~H Byrd, Peihuang Lu, Jorge Nocedal, and Ciyou Zhu.
\newblock A limited memory algorithm for bound constrained optimization.
\newblock {\em SIAM Journal on Scientific Computing}, 16(5):1190--1208, 1995.

\bibitem{zhu1997algorithm}
Ciyou Zhu, Richard~H Byrd, Peihuang Lu, and Jorge Nocedal.
\newblock Algorithm 778: {L-BFGS-B}: {Fortran} subroutines for large-scale
  bound-constrained optimization.
\newblock {\em ACM Transactions on Mathematical Software (TOMS)},
  23(4):550--560, 1997.

\bibitem{wang2004image}
Zhou Wang, Alan~C Bovik, Hamid~R Sheikh, and Eero~P Simoncelli.
\newblock Image quality assessment: from error visibility to structural
  similarity.
\newblock {\em IEEE Transactions on Image Processing}, 13(4):600--612, 2004.

\bibitem{blood-vessel-segmentation}
Yashvardhan Jain, Katy Borner, Claire Walsh, Nancy Ruschman, Peter~D. Lee,
  Griffin~M. Weber, Ryan Holbrook, and Addison Howard.
\newblock {SenNet + HOA - Hacking the Human Vasculature in 3D}.
\newblock \url{https://kaggle.com/competitions/blood-vessel-segmentation},
  2023.
\newblock Kaggle.

\bibitem{van2026assessing}
Melle Van Der~Brugge, Kalloor~Joseph Francis, and Navchetan Awasthi.
\newblock Assessing image quality in photoacoustic imaging: A metric-based and
  deep learning-based evaluation.
\newblock {\em Photoacoustics}, 49(100800), 2026.

\bibitem{van2005fourier}
Marin Van~Heel and Michael Schatz.
\newblock {F}ourier shell correlation threshold criteria.
\newblock {\em Journal of Structural Biology}, 151(3):250--262, 2005.

\bibitem{schwartz1951theorie}
Laurent Schwartz.
\newblock {\em Th{\'e}orie des Distributions}.
\newblock Hermann, 1951.

\bibitem{friedlander1998introduction}
Friedrich~Gerard Friedlander.
\newblock {\em Introduction to the Theory of Distributions}.
\newblock Cambridge University Press, 1998.

\end{thebibliography}

\end{document}